\documentclass[journal, twocolumn]{IEEEtran}

\usepackage[utf8]{inputenc} 
\usepackage[T1]{fontenc}
\usepackage{url}
\usepackage{ifthen}
\usepackage{cite}
\usepackage[cmex10]{amsmath} 
\usepackage{amsmath, amssymb, amsthm}
\usepackage{bm}
\usepackage{bbm}
\usepackage{color}
\usepackage{graphicx}
\usepackage{mathrsfs}
\usepackage{xcolor}
\usepackage{caption}
\usepackage{subcaption}
\usepackage[rightcaption]{sidecap}
\sidecaptionvpos{figure}{c}
\usepackage{multirow,tabularx}
\usepackage[normalem]{ulem}

\usepackage{cleveref}

\newcommand{\lp}{\left(}
\newcommand{\rp}{\right)}
\newcommand{\lb}{\left[}
\newcommand{\rb}{\right]}
\newcommand{\lbp}{\left\{}
\newcommand{\rbp}{\right\}}
\newcommand{\lnorm}{\left\lVert}
\newcommand{\rnorm}{\right\rVert}
\newcommand{\mv}{\,\middle\vert\,}
\newcommand{\mcal}{\mathcal}
\newcommand{\mscr}{\mathscr}

\newcommand{\mbb}{\mathbb}

\newcommand{\lce}{\left\lceil}
\newcommand{\rce}{\right\rceil}
\newcommand{\diid}{\overset{\text{i.i.d.}}{\sim}}
\newcommand{\Indc}[1]{\mathbbm{1}{\lbp #1\rbp}}
\newcommand{\argmin}{\mathop{\mathrm{argmin}}}
\newcommand{\KLD}[2]{\mathrm{D}\mkern-1.5mu\left( #1  \middle\Vert #2 \right)}
\newcommand{\RenyiD}[3]{\mathrm{D}_{#3}\mkern-1.5mu\left( #1  \middle\Vert #2 \right)}
\newcommand{\GJS}[3]{\mathrm{GJS}\mkern-1.5mu\left( #1, #2, #3 \right)}

\newcommand{\supp}[1]{\mathsf{supp}(#1)}
\newcommand{\BKLD}[2]{\mathrm{d}\mkern-1.5mu\left( #1  , #2 \right)}
\newcommand{\eT}[1]{e_{#1}}

\newcommand{\poly}[1]{\mathsf{poly}\mkern-1.5mu\left( #1 \right)}
\newcommand{\comp}[1]{#1^\mathrm{c}}
\newcommand{\closure}[1]{\mathrm{cl}(#1)}
\newcommand{\interior}[1]{\mathrm{int}(#1)}

\newcommand*{\dotleq}{\mathrel{\dot{\leq}}}

\newtheorem{theorem}{Theorem}
\newtheorem{assumption}{Assumption}
\newtheorem{lemma}{Lemma}
\newtheorem{proposition}{Proposition}
\newtheorem{corollary}{Corollary}

\theoremstyle{definition}
\newtheorem{definition}{Definition}
\theoremstyle{remark}
\newtheorem{remark}{Remark}
\newtheorem{case}{Case}

\newcommand{\distset}{\mcal{D}_\varepsilon}

\newcommand{\fullseq}{$\mathsf{Fully\text{-}Sequential}$}
\newcommand{\semione}{$\mathsf{Semi\text{-}Sequential\text{-}1}$}
\newcommand{\semitwo}{$\mathsf{Semi\text{-}Sequential\text{-}2}$}
\newcommand{\fixed}{$\mathsf{Fixed\text{-}Length}$}
\newcommand{\eseq}{e_{1,\mathsf{seq}}^*}
\newcommand{\esemione}{e_{1,\mathsf{semi1}}^*}
\newcommand{\esemitwo}{e_{1,\mathsf{semi2}}^*}
\newcommand{\efix}{e_{1,\mathsf{fix}}^*}

\allowdisplaybreaks 

\begin{document}
\title{A Unified Study on Sequentiality in Universal Classification with Empirically Observed Statistics}

\author{Ching-Fang~Li,~\IEEEmembership{Student Member,~IEEE} and
I-Hsiang~Wang,~\IEEEmembership{Member,~IEEE}
\thanks{This work was supported by NSTC of Taiwan under Grant 111-2628-E-002-005-MY2 and 113-2628-E-002-022-MY4, and NTU under Grant 113L7764, 113L891404, and 113L900902. 
The material in this paper was presented in part at the 2024 International Zurich Seminar
on Information and Communication, March 2024, and the 2024 IEEE International Symposium on Information Theory, July 2024.}%
\thanks{C.-F. Li was with the Graduate Institute of Electrical Engineering, 
National Taiwan University, Taipei 10617, Taiwan. She is now with the Department of Electrical Engineering, Stanford University (email: cfli@stanford.edu).}% <-this % stops a space
\thanks{I.-H. Wang is with the Department of Electrical Engineering and the Graduate Institute of Communication Engineering, National Taiwan University, Taipei 10617, Taiwan (email: ihwang@ntu.edu.tw).}
}

\maketitle

\begin{abstract}
In the binary hypothesis testing problem, it is well known that sequentiality in taking samples eradicates the trade-off between two error exponents, yet implementing the optimal test requires the knowledge of the underlying distributions, say $P_0$ and $P_1$. In the scenario where the knowledge of distributions is replaced by empirically observed statistics from the respective distributions, the gain of sequentiality is less understood when subject to universality constraints over all possible $P_0,P_1$. 
In this work, the gap is mended by a unified study on sequentiality in the universal binary classification problem, 
where the universality constraints are set on the expected stopping time as well as the type-I error exponent. 
The type-I error exponent is required to achieve a pre-set distribution-dependent constraint $\lambda(P_0,P_1)$ for all $P_0,P_1$. 
Under the proposed framework, different sequential setups are investigated so that fair comparisons can be made with the fixed-length counterpart. 
By viewing these sequential classification problems as special cases of a general sequential composite hypothesis testing problem, the optimal type-II error exponents are characterized. 
Specifically, in the general sequential composite hypothesis testing problem subject to universality constraints, upper and lower bounds on the type-II error exponent are proved, and a sufficient condition for which the bounds coincide is given. The results for sequential classification problems are then obtained accordingly. 
With the characterization of the optimal error exponents, the benefit of sequentiality is shown both analytically and numerically by comparing the sequential and the fixed-length cases in representative examples of type-I exponent constraint $\lambda$. 
\end{abstract}

\section{Introduction}\label{sec:intro}
It is known that sequentiality in taking samples greatly enhances reliability in statistical inference. Take the binary hypothesis testing problem as an example. The decision maker observes a sequence of samples drawn i.i.d. from one of the two known distributions $P_0$ or $P_1$. It aims to infer from which of the two distributions the sequence is generated. Both type-I and type-II error probabilities vanish exponentially fast as the number of samples $n$ tends to infinity. The exponential rates are denoted as the error exponents $\eT{0}$ and $\eT{1}$ respectively, between which there exists a fundamental trade-off \cite{Blahut_74}. When samples are taken sequentially and the decision maker is free to decide when to stop as long as the expected stopping time is less than a given constraint $n$, Wald's Sequential Probability Ratio Test (SPRT) \cite{Wald_45} is shown to be optimal \cite{WaldWolfowitz_48} and the error exponents can simultaneously achieve the two extremes, namely, the two KL divergences $\KLD{P_1}{P_0}$ and $\KLD{P_0}{P_1}$. In words, sequentiality in taking samples eradicates the trade-off between error exponents. The optimal design of the stopping and inference strategies, however, critically relies on the knowledge of the underlying distributions \cite{Wald_45,WaldWolfowitz_48,DragalinTartakovsky_99}. 

In this work, we aim to investigate the benefit of sequentiality when the underlying distributions are unknown to the decision maker. 
Instead, the decision maker only has access to empirical statistics of training sequences sampled from these distributions, say, $P_0$ and $P_1$ in the binary-hypothesis case. 
Since the underlying distributions are unknown, it is natural to ask for a \emph{universal} guarantee on certain performances. One such framework focused on the asymptotic performance as the number of samples in all sequences tend to infinity was proposed and studied by Ziv \cite{Ziv_88}, where a universality constraint was set on the type-I error exponent to be no less than a given constant $\lambda_0>0$ regardless of the underlying distributions. Gutman \cite{Gutman_89} later proved the asymptotic optimality of Ziv's universal test in the sense that it attains the optimal type-II error exponent (as a function of $P_0, P_1, \lambda_0$). 
Further extension was made by Levitan and Merhav \cite{LevitanMerhav_02} where a competitive criterion was taken, replacing the constant constraint $\lambda_0$ by a \emph{distribution-dependent} one, $\lambda(P_0,P_1)$. From these results in the fixed-length settings, it can be seen that there also exists a trade-off between the two error exponents. 
A natural question emerges: in this universal classification problem, can sequentiality improve the trade-off, and even eradicate them \emph{without} knowing the underlying distributions?

Full resolution to the question above is provided in this paper, and as we will see, the answer depends on the distribution-dependent constraint $\lambda(P_0,P_1)$ and the setup of sequentiality. 
Since the decision maker observes three sequences (one testing and two training), there are different ways of considering sequentiality. The main focus is on the following setups: (1) the \fullseq\ setup where testing and training samples all arrive sequentially, (2) the \semione\ setup where the training sequences have fixed lengths and testing samples arrive sequentially, and (3) the \semitwo\ setup where the testing sequence has a fixed length and training samples arrive sequentially. 
To have a fair comparison with the \fixed\ setup in previous works, universality constraints are set on the expected stopping time and the type-I error exponent. 

Our main contribution is the characterization of the optimal type-II error exponent when the type-I error exponent universality constraint $\lambda(P_0,P_1)$ satisfies mild continuity conditions. 
For the \fullseq\ setup, the optimal type-II error exponent is $\eseq(P_0,P_1) =\min\big\{ \RenyiD{P_0}{P_1}{\frac{\alpha}{1+\alpha}},\kappa(P_0,P_1)\big\}$. Here the first term $\RenyiD{P_0}{P_1}{\frac{\alpha}{1+\alpha}}$ is the R\'enyi divergence and $\alpha$ is the expected ratio of the number of the $P_0$-training samples to that of testing samples. It corresponds to the error event that the test stops at around time $n$ with testing and $P_0$-training samples looking like they are generated by the same distribution. 
Likewise, the second term $\kappa(P_0,P_1)$ corresponds to the error event that the test stops at around time $n$ while the testing samples look like $P_1$-training samples instead but the empirical distributions fall in the optimal \fixed\ decision region of the null hypothesis. 
For the \semione\ setup, there is an additional term as the performance is further restricted by the fixed-length training sequences. Specifically, the optimal type-II error exponent is $\esemione(P_0,P_1) =\min\big\{ \eseq(P_0,P_1),\mu(P_0,P_1)\big\}$.  
The additional term $\mu(P_0,P_1)$ corresponds to an error event in the regime where the number of testing samples is much larger than the fixed lengths of training sequences and is hence the optimal type-II error exponent in an induced fixed-length composite hypothesis testing problem \cite{LevitanMerhav_02} when the distribution of the testing samples is fixed as $P_1$ and known. 
Similarly, for the \semitwo\ setup, the additional restriction is from the fixed-length testing sequence, and $\esemitwo(P_0,P_1) =\min\big\{ \eseq(P_0,P_1),\nu(P_0,P_1)\big\}$.  
Here $\nu(P_0,P_1)$ instead corresponds to an error event in the regime where the numbers of training samples are both much larger than the fixed length of the testing sequence so that effectively the distributions of the two hypotheses are both known. Hence, it is the optimal type-II error exponent for the binary hypothesis problem when the type-I error exponent is $\lambda(P_0,P_1)$.
For the \fixed\ setup, the optimal type-II error exponent is given in \cite{LevitanMerhav_02} and denoted as $\efix(P_0,P_1)$. To illustrate the benefit of sequentiality and the optimal type-II error exponents in different setups, a numerical example is given in Figure~\ref{fig:compare_weird}, where \fullseq\ strictly outperforms \semione /\semitwo, and \semione /\semitwo\ strictly outperforms \fixed. Meanwhile, \semione\ is sometimes better than \semitwo\ and sometimes worse. 

\begin{figure*}[ht]
    \centering
    \includegraphics[scale=0.6]{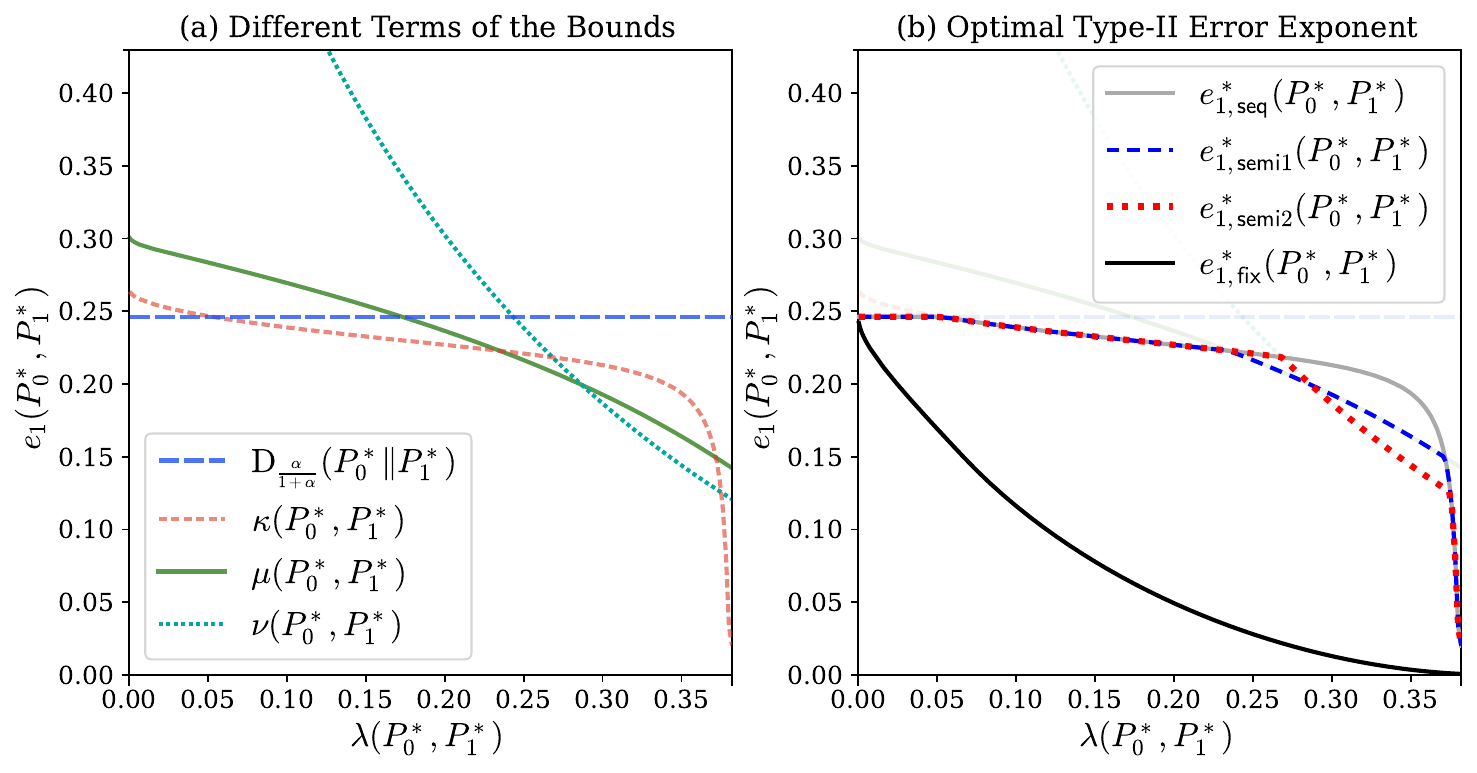}
    \caption{The optimal type-II error exponents under type-I error exponent constraint $\textstyle\lambda(P_0,P_1)=\xi\big(\RenyiD{P_1}{P_0}{\frac{\beta}{1+\beta}}+0.003\big)$, where $\beta$ is the expected ratio of the number of the $P_1$-training samples to that of testing samples. 
    Here $\mcal X = \{0,1\}$,
    $\alpha=0.38$, $\beta=0.6$, and $P_0^*=[0.6, 0.4]$, $P_1^*=[0.1, 0.9]$.
    Let $\xi$ increase from $0.001$ to $1$ to obtain a curve.
    The optimal type-II error exponents in different setups can be identified by taking the minimum of the corresponding terms, as demonstrated in plot (b). 
    This particular example is chosen judiciously so that all terms are active in the considered regime.
    } 
    \label{fig:compare_weird}
\end{figure*}

Analytical comparisons are also made besides numerical evaluation. 
In particular, for the choice of $\lambda(P_0,P_1)$ being a constant $\lambda_0 > 0$, we show that \fixed\ tests have strictly smaller type-II error exponent than \fullseq\ ones. Moreover, \semione\ and \fullseq\ tests have the same optimal type-II error exponents, indicating that there is no additional gain due to sequentiality in taking training samples when testing samples are sequentially observed.
On the other hand, when the choice of $\lambda(P_0,P_1)$ permits exponentially vanishing error probabilities for all $P_0,P_1$, it is shown that the trade-off between error exponents is eradicated in the \fullseq\ case, demonstrating strict error exponent gain over \fixed\ tests. Regarding the additional benefit of sequentiality in taking training samples when testing samples are already sequential, we characterize the necessary and sufficient condition of whether or not there exists a strict gap between the optimal exponents of \semione\ and \fullseq\ tests. The condition pertains to $\alpha$ and $\beta$, the expected ratios of the number of the $P_0/P_1$-training samples to that of the testing samples, respectively. 

A unified proof of the optimal type-II error exponents in the classification problem is developed by investigating a more general composite hypothesis testing problem with $s$ 
independent sequences, each of which is either fixed-length or sequentially observed. In this problem, there are two disjoint composite hypothesis sets, $\mscr P_0$ and $\mscr P_1$, 
collecting the possible underlying distributions of the sequences under two hypotheses.
When all sequences are fixed-length, the problem was investigated in \cite{LevitanMerhav_02}, and hence the proposed problem can be viewed as a generalization of that in \cite{LevitanMerhav_02} with sequential components. Similar to the setup in the classification problem, the universality constraints are set on the expected stopping time and the type-I error exponent. For this general composite hypothesis testing problem with sequential components, we establish upper and lower bounds on the type-II error exponents, as well as a sufficient condition 
such that the upper and lower bounds coincide. 
The optimal type-II error exponents of the classification problem can be obtained by identifying the different sequential setups as special cases of the composite hypothesis testing problem. 
We show that the aforementioned sufficient condition holds in these cases, and characterize the the optimal type-II error exponents in the classification problems.

\subsection{Related works}
For the sequential version of binary classification problem, 
Haghifam \textit{et al.} \cite{HaghifamTan_21} considered the \semione\ setup as well. They proposed a test and showed that it achieves larger Bayesian error exponent over the \fixed\ case. Under the same setting, Bai \textit{et al.} \cite{BaiZhou22} proposed an \emph{almost fixed-length} two-phase test with performance lying between Gutman's \fixed\ test and the \semione\ test in \cite{HaghifamTan_21}. However, in both \cite{HaghifamTan_21} and \cite{BaiZhou22}, they did not have a universality constraint on the expected stopping time, nor other universal guarantees over all possible distributions. For example, in \cite{HaghifamTan_21}, the expected stopping time of their test depends implicitly on the unknown distributions $P_0,P_1$. Moreover, in order to achieve certain performance guarantees, parameters have to be chosen to satisfy some conditions that depend on the underlying distribution. 
A problem of lacking universality guarantees is that there may be a test that performs well under a specific pair of distributions and bad under another pair of distributions, while another test does the opposite. Then it cannot be said which test is better. Even if a test may be instance optimal, there might not exist a universally optimal test. As a result, it would be difficult or even impossible to derive a tight converse bound matching the achievability results. Such difficulties are experienced in both \cite{HaghifamTan_21} and \cite{BaiZhou22}. In comparison, by setting universality constraints, we can focus on tests that are (at least to certain extent) comparable to each other. For such tests, optimality can be clearly defined and there even exist tests that are universally optimal.

Hsu \textit{et al.} \cite{HsuWang_22} considered the same \fullseq\ setup as in this work, while the universality constraints are set on the expected stopping time and vanishing error probabilities.
Hence the result is not directly comparable with those works in the \fixed\ setting \cite{Ziv_88,Gutman_89,LevitanMerhav_02}. 
In a preliminary version of this work \cite{LiWang_24}, the treatment only pertains to universal tests that can achieve exponentially vanishing error probabilities, and thus faces a similar problem as \cite{HsuWang_22}, not comparable with results under a constant type-I error exponent constraint \cite{Ziv_88,Gutman_89}. Besides, the lengths of training sequences are assumed to be the same in \cite{LiWang_24}, while in this work, the assumption is lifted.
Partial results of this work also appear in \cite{LiWang_24_isit}, where the focus is on \semione\ and \fullseq\ classification. Here the results are generalized to a composite hypothesis testing problem and the proofs are unified. This provides the flexibility to consider more possibilities regarding the sequentiality in taking samples, and the potential to apply to other classical problems.
For the sake of completeness, results in \cite{LiWang_24, LiWang_24_isit} will also be included in exposition.

\subsection{Organization}
In Section~\ref{sec:formulation}, we formally define the binary classification problem, including different sequential setups, the universality constraints, and the performance metrics. The general composite hypothesis testing problem is also formulated similarly. 
In Section~\ref{sec:general_results}, we present the results for the general composite hypothesis testing problem, serving as the main tool for characterization of error exponents for the binary classification problem. 
In \Cref{sec:results}, we present the optimal error exponents for the binary classification problem and show that they can be derived as a special case of the general composite hypothesis testing problem.
The section is completed with comparisons between different sequential setups.
Section~\ref{sec:converse} and Section~\ref{sec:achievability} contain respectively the proof of converse and achievability for the results in the general composite hypothesis testing problem. Section~\ref{sec:discussion} collects discussions on various aspects related to our problem, including the technical assumptions, comparison with prior works, and extensions. 
Finally, Section~\ref{sec:conclusion} concludes the paper.

\subsection{Notation} 
A finite-length sequence $(x_1,x_2,...,x_n)$ is denoted as $x^n$. Logarithms are of base $2$ if not specified. 
For a subset $S$ of a topological space, $\interior{S}$ and $\closure{S}$ denote the interior and closure of $S$, respectively. 
$\mcal{P}(\mcal{X})$ is the set of all probability distributions over alphabet $\mcal{X}$. 
When $|\mcal X|=d$, the probability simplex $\mcal{P}(\mcal{X})$ can be embedded into $\mbb R^{d-1}$. 
Hence we use $\interior{\mcal{P}(\mcal{X})}$ to denote the set of all probability distributions over alphabet $\mcal{X}$ with full support.
An indicator function is written as $\Indc{\cdot}$. 
Given positive sequences $\{a_n\}$ and $\{b_n\}$, we write $a_n\doteq b_n$ if $\lim_{n\to\infty}\frac{1}{n}\log \frac{a_n}{b_n}=0$. The relation $\dotleq$ is defined similarly. 
Also, $o(1)$ is used to denote vanishing terms, and $\poly{n}$ means some polynomial in $n$.

\section{Problem Formulation}\label{sec:formulation}
\subsection{Binary Classification}

Let $\mcal{X}$ be a finite alphabet with $|\mcal{X}|=d\geq2$ and 
consider the set $\mcal P_\varepsilon = \lbp P\in\mcal P(\mcal X) \mv \forall\,x\in\mcal X,\ P(x)\geq\varepsilon \rbp$ for some $\varepsilon>0$. 
Note that $\mcal P_\varepsilon$ is compact and chosen to ensure that the KL divergences between these distributions are bounded and uniformly continuous. The underlying distributions are described by a pair of distinct distributions $(P_0,P_1)\in\distset = \lbp(P,Q)\mv P,Q\in\mcal P_\varepsilon,\ P\neq Q\rbp$,
and $(P_0,P_1)$ is \emph{unknown} to the decision maker. 
Under ground truth $\theta\in\{0,1\}$, the decision maker observes i.i.d. testing samples $X_k$'s following $P_\theta$, along with i.i.d. training samples $T_{0,k}$'s and $T_{1,k}$'s following $P_0$ and $P_1$ respectively to learn about the unknown underlying distributions. Note that the testing and training samples are mutually independent. 
Let $n\in\mbb{N}$ be an index of the problem, which will later be related to the number of testing samples in expectation.
Let $\alpha,\beta>0$ be two fixed problem parameters indicating the asymptotic ratios of the number of training samples to $n$. % (in expectation).
The objective of the decision maker is to output $\hat{\theta}\in\{0,1\}$ as an estimation of the unknown ground truth $\theta$, based on the observed samples.

Next let us specify different setups regarding the sequentiality of taking samples.
\begin{itemize}
    \item \fullseq: the testing and training samples are all sequentially observed. At time $k$, there are $k$ testing samples and $N_k=\lce\alpha k\rce, M_k=\lce\beta k\rce$ training samples from each distribution. 
    A test is a pair $\Phi_n = (\tau_n,\delta_n)$ where $\tau_n\in\mbb{N}$ is a Markov stopping time with respect to the filtration $\mcal F_k = \sigma(X^k,T_0^{N_k},T_1^{M_k})$. Recall that $T_0^{N_k}$ denote the first $N_k$ training samples $T_{0,i}$'s. We may write $\tau_n$ as $\tau$ when it is clear from the context. The decision rule $\delta_n:\mcal X^\tau\times\mcal X^{N_\tau}\times \mcal X^{M_\tau}\to\{0,1\}$ is a $\mcal F_\tau$-measurable function, and the output is denoted as $\hat{\theta}$. 
    \item \semione: the testing samples $X_k$'s arrive sequentially, while the numbers of $P_0$/$P_1$-training samples are fixed to $N=\lce\alpha n\rce$ and $M=\lce\beta n\rce$. Denote the two training sequences as $T_0^N$ and $T_1^M$. 
    A test $\Phi_n = (\tau_n,\delta_n)$ is similarly defined as in the \fullseq\ 
    setup, with time-dependent variables $N_k,M_k$ replaced by the constants $N, M$.
    \item \semitwo: the number of the testing samples is fixed to $n$ and the training samples are sequentially observed. That is, at time $k$, there are $N_k=\lce\alpha k\rce$ and $M_k=\lce\beta k\rce$ training samples from each distribution. A test $\Phi_n = (\tau_n,\delta_n)$ is similarly defined as in the \fullseq\   
    setup, with the number of testing samples $k$ replaced by the constant $n$.
    \item \fixed: all three sequences have fixed-length. Note that it can be viewed as restricting $\tau_n=n$ in either one of the above sequential setups. 
\end{itemize}
When the collection of training samples are done way earlier than the collection of testing samples, it is reasonable to consider the \semione\ setup, as it may be infeasible to collect the training samples once the collection of testing samples starts. On the other hand, if the collection of training samples can be done on-the-fly or after the collection of testing samples, it may be reasonable to consider the \semitwo\ and \fullseq\ setups. 
For example, a real-world application is mentioned in \cite{GerberPolyanskiy_24} about determining whether the Higgs boson exists. In this case, the actual data from experiment is expensive to obtain. However, one can generate simulated data under different hypotheses. Since simulation has lower cost and easy to resume, both the \semitwo\ and \fullseq\ setups apply. 

In the following, when $X^k$ is observed, denote the empirical distribution as $\hat{P}^k$, where $\hat{P}^k(x) = \frac{1}{k}\sum_{i=1}^k \Indc{X_i=x}$ for $x\in\mcal X$. 
The empirical distributions of $T_0^N$ and $T_1^M$ are denoted as $\hat{P}_0^N$ and $\hat{P}_1^M$ (we omit $N, M$ if it is clear from the context).

The performance of tests is measured by the error probability and the number of samples used. Given $(P_0,P_1)\in\distset$ and $\theta\in\{0,1\}$, the error probability is defined as $\pi_\theta(\Phi_n|P_0,P_1) = \mbb P_\theta\{\hat{\theta}\neq\theta\}$, where $\mbb P_\theta$ is the shorthand notation for the joint probability law of the testing sequence and training sequences. The average number of samples used can be described by the expected stopping time $\mbb E_\theta[\tau_n|P_0,P_1]$, where the expectation is taken under $\mbb P_\theta$. 
Since the underlying distributions are unknown, it is natural to ask for some universal guarantees on the performance. The universality constraints are twofold. 
\begin{itemize}
\item 
First, to compare with fixed-length tests, we set a \emph{universality constraint on the expected stopping time} to be at most $n$. 
Let $\{\Phi_n\}$ be a sequence of tests where $\Phi_n$ satisfies $\mbb E_\theta[\tau_n|P_0,P_1]\leq n$ for all underlying distributions $(P_0,P_1)\in\distset$ and ground truth $\theta$. The type-I and type-II error exponents of $\{\Phi_n\}$ given $(P_0,P_1)\in\distset$ are defined as
\begin{equation}
    \eT{\theta}(P_0,P_1) = \liminf\limits_{n\to\infty}\frac{-\log{\pi_\theta(\Phi_n|P_0,P_1)}}{n},\ \text{for }\theta=0,1.
\end{equation}

\item
Second, we adopt the competitive Neyman-Pearson criterion proposed in \cite{LevitanMerhav_02} and set a \emph{universality constraint on the type-I error exponent}.
Let $\lambda:\distset\to (0,\infty)$ be a pre-set distribution-dependent constraint function. We focus on tests satisfying
\begin{equation}
    \eT{0}(P_0,P_1)\geq \lambda(P_0,P_1),\quad\forall\,(P_0,P_1)\in\distset. 
\end{equation} 
\end{itemize}
The goal is to characterize the maximum $\eT{1}(P_0,P_1)$ that can be achieved for tests satisfying these universality constraints and to find such a test, if possible, that achieves the maximum \emph{uniformly} over all possible underlying distributions. 

We close this sub-section with the following assumption on $\lambda$ that will be used in proving the achievability in the \semione\ and \semitwo\ setup.

\begin{assumption}\label{assumption:lambda}
    The function $\lambda:\distset\to (0,\infty)$ can be extended to a continuous function $\bar\lambda:\mcal P_\varepsilon\times \mcal P_\varepsilon\to[0,\infty)$.
\end{assumption}

\begin{remark}
    The competitive Neyman-Pearson criterion can be viewed as a generalization of the generalized Neyman-Pearson criterion which requires the type-I error exponent to be lower bounded by a constant regardless of the underlying distributions. Since $P_0$ and $P_1$ can be arbitrarily close to each other, in some cases it might be too stringent to ask for a constant type-I error exponent as that will inevitably make the type-II error tend to $1$. In \cite{LevitanMerhav_02}, the competitive Neyman-Pearson criterion is proposed to allow more flexibility and make it possible that both error probabilities decay to zero exponentially for all the underlying distributions. Two examples of the threshold functions are $\lambda \equiv \lambda_0$ and $\lambda(P_0,P_1) = \xi \RenyiD{P_1}{P_0}{\frac{\beta}{1+\beta}}$ for some $\xi\in(0,1)$, corresponding to the two scenarios mentioned above.
    Further discussions can be found in Section~\ref{subsec:compare}. It is worth noting that since the distributions are unknown, one should define $\lambda(P_0,P_1)$ for all $(P_0,P_1)\in\distset$, specifying the performance requirement for every possibly encountered case.
\end{remark}

\subsection{Composite Hypothesis Testing}\label{subsec:formulation_general}
In \cite{LevitanMerhav_02}, a more general composite hypothesis testing problem was considered, and the binary classification problem in the \fixed\ setup can be viewed as a special case. 
To provide a unified proof of our main results, the optimal error exponents of binary classification problem, we propose a more general version of the problem in \cite{LevitanMerhav_02} so that the setups with sequentially arrived samples such as \semione, \semitwo, and \fullseq\ can be incorporated. 

\begin{figure*}[ht]    
    \centering
    \includegraphics[scale=1]{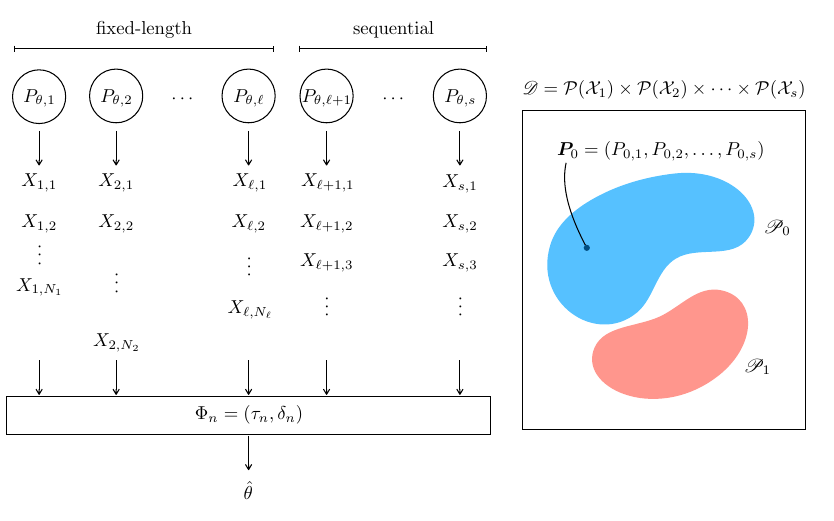}
    \caption{The setup of the generalized problem.} 
    \label{fig:explain}
\end{figure*}

Let us first formalize the general setup. 
For convenience, in the context of the generalized problem, some notations are reused and overridden to avoid introducing too many new ones. The definitions of the reused notations should be clear in the respective contexts.  
Let $s\in\mbb N$ be the number of sequences. For each $i=1,2,\dots,s$, $\mcal X_i$ is a finite alphabet and $\alpha_i>0$ is a constant. Let $\mscr D:=\mcal P(\mcal X_1)\times \mcal P(\mcal X_2)\times\dots\times \mcal P(\mcal X_s)$ and disjoint sets $\mscr P_0, \mscr P_1\subseteq \interior{\mscr D}$ be the collections of possible underlying distributions when the ground truth is $0$ and $1$, respectively.
Fix some integer $0\leq\ell\leq s$, indicating the number of fixed-length sequences. 

Given an index $n\in\mbb N$, the decision maker observes $\ell$ fixed-length sequences $X_1^{N_1},X_2^{N_2},\dots, X_\ell^{N_\ell}$, where $N_i = \lce\alpha_i n\rce$ is the length of the $i$-th sequence $X_i^{N_i}=(X_{i,1},X_{i,2},\dots,X_{i,N_i})\in\mcal X_i^{N_i}$ for $i=1,2,\dots,\ell$. In addition, some samples are sequentially observed. Specifically, for $i=\ell+1,\dots,s$, at time $k$, there are $\lce\alpha_i k\rce$ samples in the $i$-th sequence. 
Given the unknown ground truth $\theta\in\{0,1\}$, each sample $X_{i,k}$ of the $i$-th sequence i.i.d. follow $P_{\theta,i}$, which is the $i$-th component of the underlying distribution $\bm P_\theta = (P_{\theta,1},P_{\theta,2},\dots,P_{\theta,s})\in \mscr P_\theta$. 
Note that all the sequences are assumed to be mutually independent. 
The objective of the decision maker is to output $\hat{\theta}\in\{0,1\}$ as an estimation of the ground truth $\theta$, based on the observed samples. 
A test $\Phi_n=(\tau_n,\delta_n)$ consists of a Markov stopping time with respect to the filtration $\mcal F_k=\sigma(X_1^{N_1},\dots,X_\ell^{N_\ell}, X_{\ell+1}^{\lce\alpha_{\ell+1} k\rce},\dots,X_s^{\lce\alpha_s k\rce})$ and a decision rule that maps observations to $\hat{\theta}$. Given ground truth $\theta$ and underlying distributions $\bm P_\theta$, the error probability $\pi_\theta(\Phi_n|\bm P_\theta)=\mbb P_\theta\{\hat{\theta}\neq\theta\}$, where $\mbb P_\theta$ is the shorthand notation for the joint probability law of all the sequences. The expected stopping time is $\mbb E_\theta[\tau|\bm P_\theta]$, where the expectation is taken under $\mbb P_\theta$.

Again, we set a constraint on the expected stopping time to be at most $n$, that is, $\mbb E_\theta[\tau|\bm P_\theta]\leq n$ for all $\theta=0,1$ and $\bm P_\theta\in\mcal P_\theta$.
The error exponent is defined as
\begin{equation}
    e_\theta(\bm P_\theta) = \liminf_{n\to\infty}\frac{-\log \pi_\theta(\Phi_n|\bm P_\theta)}{n}.
\end{equation}
We also consider the competitive Neyman-Pearson criterion, where a constraint is set on the type-I error exponent. Let $\lambda:\mscr P_0 \to (0,\infty)$, and we focus on tests satisfying $e_0(\bm P_0)\geq \lambda(\bm P_0)$ for all $\bm P_0\in\mscr P_0$. 

It is not hard to see that the different sequential binary classification problems are instances of this general composite hypothesis testing problem. Detailed specifications can be found in \Cref{subsec:proof_of_thm1-2,subsec:proof_of_thm3}.

\section{Results for Composite Hypothesis Testing}\label{sec:general_results}
In this section, we present converse (\Cref{prop:composite_converse}) and achievability (\Cref{prop:composite_achievability}) bounds on the type-II error exponent in the general composite hypothesis testing problem. Then we provide sufficient conditions for which the two bounds coincide.

Let us first take a look at the special case $\ell=s$, which is essentially the fixed-length setup in \cite{LevitanMerhav_02}. The slight difference is that there is only one sequence in the original setup \cite{LevitanMerhav_02}. 
It is not difficult to modify the arguments in \cite{LevitanMerhav_02} to show that the optimal type-II error exponent is
\begin{equation}
    e^*_{1,\mathsf{fix}}(\bm P_1) = \inf_{\substack{\bm Q\in\mscr D \\ g_1(\bm Q)<0}} \sum_{i=1}^s \alpha_i \KLD{Q_i}{P_{1,i}},\label{eq:optimal_fixed-length_e1}
\end{equation}
where $\bm Q$ denotes the tuple $(Q_1,Q_2,\dots,Q_s)$, and
\begin{equation}
    g_1(\bm Q) = \inf_{\bm P_0'\in\mscr P_0} \Big( \sum_{i=1}^s \alpha_i \KLD{Q_i}{P_{0,i}'} - \lambda(\bm P_0') \Big).
\end{equation}
Also, there exist tests satisfying the constraint on the type-I error exponent and achieve $e^*_{1,\mathsf{fix}}(\bm P_1)$ simultaneously for all $\bm P_1\in\mscr P_1$. 
Notice that $\sum_{i=1}^s \alpha_i \KLD{Q_i}{P_{1,i}}$ in \eqref{eq:optimal_fixed-length_e1} is the divergence of $\bm Q$ from the real distribution $\bm P_1$, and $g_1(\bm Q)<0$ means $\bm Q$ is ``close'' to some distribution $\bm P_0'\in\mscr P_0$. Here ``closeness'' is measured by the KL divergence, and the function $\lambda$ is used as the threshold. We say $\bm Q$ is ``$\lambda$-close'' to some $\bm P_0'\in\mscr P_0$, and such $\bm Q$ constitutes the decision region for hypothesis $0$.

In the following treatment, we shall focus on the case $\ell<s$. 
The results are summarized in the following propositions.
\begin{proposition}[Converse]\label{prop:composite_converse}
    Let $\lambda:\mscr P_0 \to (0,\infty)$, and $\{\Phi_n\}$ be a sequence of tests such that
    \begin{itemize}
        \item for each $\Phi_n = (\tau_n,\delta_n)$, the expected stopping time 
        $\mbb E_\theta[\tau_n|\bm P_\theta]\leq n$ for all $\theta\in\{0,1\}$ and $\bm P_\theta\in\mscr P_\theta$,
        \item %type-I error exponent satisfies 
        $e_0(\bm P_0)\geq \lambda(\bm P_0)$ for all $\bm P_0\in\mscr P_0$.
    \end{itemize}
    Then for any $\bm P_1\in\mscr P_1$,
    \begin{equation}\label{eq:Expt_CHT_converse}
        e_1(\bm P_1) \leq e_1^*(\bm P_1):=\inf_{\bm Q\in\mscr P_0\cup \Gamma_0\cup\Omega(\bm P_1)} \sum_{i=1}^s \alpha_i \KLD{Q_i}{P_{1,i}},
    \end{equation}
    where $\Gamma_0 = \lbp\bm Q\in\mscr P_1\mv g_1(\bm Q) < 0\rbp$ and
    \begin{align}
\label{eq:def_Omega_P1}
        \Omega(\bm P_1) &= 
            \lbp\bm Q\in \mcal{S}(\bm P_1) \mv g^{\bm P_1}(\bm Q) < 0\rbp, \\
\label{eq:def_g_P1}
        g^{\bm P_1}(\bm Q) &= \inf_{\bm P_0'\in\mscr P_0\cap \mcal{S}(\bm P_1)} \Big( \sum_{i=1}^{\ell} \alpha_i \KLD{Q_i}{P_{0,i}'} - \lambda(\bm P_0') \Big),\\       
        \mcal{S}(\bm P_1) &= \lbp \bm Q\in\mscr D \mv Q_i=P_{1,i}\ \forall\,i=\ell+1,\dots,s\rbp.
    \end{align}
\end{proposition}

The type-II error exponent \eqref{eq:Expt_CHT_converse} is the KL divergence of some distribution $\bm Q\in\mscr P_0\cup \Gamma_0\cup\Omega(\bm P_1)$ from the real underlying distribution $\bm P_1$. The first set $\mscr P_0$ collects the distributions under hypothesis 0. 
The second set $\Gamma_0$ consists of distributions in $\mscr P_1$ that are ``$\lambda$-close'' to some $\bm P_0'\in\mscr P_0$.
The idea behind these two sets is that if the empirical distributions of the observed sequences are close to the real underlying distributions, then in order to satisfy the universality constraint on the expected stopping time, the test should not stop too late as this event happens with high probability. However, since the real underlying distributions are unknown, the test should stop early as long as the empirical distributions are close to any possible distributions in $\mscr P_0$ or $\mscr P_1$. 
Moreover, if the test stops at some time around time $n$, it should follow the optimal decision rule in the fixed-length setup \cite{LevitanMerhav_02} in order to satisfy the universality constraint on the type-I error exponent. That is, output $0$ if the empirical distributions are ``$\lambda$-close'' to some $\bm P_0'\in\mscr P_0$. The main technique to make these arguments rigorous is the method of types. 

The third set $\Omega(\bm P_1)$ represents the restriction caused by the first $\ell$ fixed-length sequences, and is obtained by a reduction from a new fixed-length composite hypothesis testing problem where the distributions behind sequentially observed samples are perfectly learned. 
Apparently the decision maker cannot do better than when it exactly knows those distributions. 
To understand the notation, notice that the set $\mcal{S}(\bm P_1)$ is a slice where the last $s-\ell$ coordinates are restricted to be the same as $\bm P_1$. This corresponds to the situation that the distributions of the last $s-\ell$ sequences are known (perfectly learned from the sequentially observed samples). The set $\Omega(\bm P_1)$ consists of distributions that are too close to some distribution $\bm P_0'\in\mscr P_0$, with everything restricted to the slice $\mcal{S}(\bm P_1)$. 
The complete proof of \Cref{prop:composite_converse} is provided in \Cref{sec:converse}.

\begin{proposition}[Achievability]\label{prop:composite_achievability}
    Suppose the following two conditions hold: 
    \begin{enumerate}
    \item $\mscr P_0$ is bounded away from the boundary of $\mscr D$.
    \item $\lambda:\mscr P_0 \to (0,\infty)$ can be extended to a continuous function $\bar\lambda:\closure{\mscr P_0}\to[0,\infty)$. 
    \end{enumerate}
    Then there exists a sequence of tests $\{\Phi_n\}$ satisfying the universality constraints \begin{itemize}
        \item for each $\Phi_n = (\tau_n,\delta_n)$, the expected stopping time 
        $\mbb E_\theta[\tau_n|\bm P_\theta]\leq n$ for all $\theta\in\{0,1\}$ and $\bm P_\theta\in\mscr P_\theta$,
        \item 
        $e_0(\bm P_0)\geq \lambda(\bm P_0)$ for all $\bm P_0\in\mscr P_0$,
    \end{itemize}
    and achieves the the following type-II error exponent simultaneously for all $\bm P_1\in\mscr P_1$:
    \begin{equation}\label{eq:Expt_CHT_achievable}
        e_1(\bm P_1) \geq\inf_{\bm Q\in\mscr P_0\cup \Gamma_0\cup\bar\Omega (\bm P_1)} \sum_{i=1}^s \alpha_i \KLD{Q_i}{P_{1,i}},
    \end{equation}
    where
    \begin{align}
\label{eq:def_Omega_bar_P1}
        \bar\Omega (\bm P_1) &= 
            \lbp\bm Q\in \mcal{S}(\bm P_1) \mv \bar g^{\bm P_1}(\bm Q) \leq 0\rbp,\\
        \bar g^{\bm P_1}(\bm Q) &= \inf_{\bm P_0'\in\closure{\mscr P_0}\cap \mcal{S}(\bm P_1)} \Big( \sum_{i=1}^{\ell} \alpha_i \KLD{Q_i}{P_{0,i}'} - \bar\lambda(\bm P_0') \Big).
        \label{eq:def_g_bar_P1}
    \end{align}
\end{proposition}

To prove \Cref{prop:composite_achievability}, we propose a two-phase test that satisfies the universality constraints. 
Following the same idea as in the converse part, to ensure the expected stopping time is bounded by $n$, the test stops at time $n-1$ if the empirical distributions of the observed sequences are close to any distributions in $\mscr P_0$ or $\mscr P_1$. 
The decision rule follows the optimal fixed-length test in \cite{LevitanMerhav_02} to meet the type-I error exponent constraint. The contribution to the type-II error exponent in this part unsurprisingly corresponds to $\mscr P_0$ and $\Gamma_0$. 

On the other hand, if the empirical distributions are not close to any possible underlying distributions, it means that the samples are inadequate to make a good estimation. 
Again, motivated by the idea in the converse part, we aim to approach the case when the distributions behind sequentially observed samples are known. Intuitively, if we observe infinitely many samples, then the distributions behind those sequences can be known by the law of large numbers. 
Thus in this case the test continues to take more samples until time $n^2$, which can be viewed as mimicking having infinitely many samples.
The decision rule again leverages the idea of fixed-length test in \cite{LevitanMerhav_02}. In this part, the contribution to the type-II error exponent corresponds to $\bar{\Omega}(\bm P_1)$.

Observe that $\bar{\Omega}(\bm P_1)$ is the only part that \eqref{eq:Expt_CHT_achievable} in \Cref{prop:composite_achievability} differs from \eqref{eq:Expt_CHT_converse} in \Cref{prop:composite_converse}.
In particular, $<$ in \eqref{eq:def_Omega_P1} is replaced by $\leq$ in \eqref{eq:def_Omega_bar_P1}, and $g^{\bm P_1}(\bm Q)$ in \eqref{eq:def_g_P1} is replaced by $\bar g^{\bm P_1}(\bm Q)$ in \eqref{eq:def_g_bar_P1}, where the infimum is taken over $\mscr P_0\cap \mcal{S}(\bm P_1)$ and $\closure{\mscr P_0}\cap \mcal{S}(\bm P_1)$ respectively. Also note that in \eqref{eq:def_g_bar_P1}, $\lambda$ is replaced by its extension $\bar{\lambda}$ so that the definition is valid. 
The reason behind this gap is that in the proof of \Cref{prop:composite_converse}, 
we assume the distributions of the last $s-\ell$ sequences are known, which is possible only if infinitely many samples are observed.
However, in the proof of \Cref{prop:composite_achievability}, the constructed test only observes finitely many samples (in the order of $n^2$), so those distributions are just approximated instead of perfectly known. As a result, the closure emerges when taking limit. 
The two conditions in \Cref{prop:composite_achievability} are also added to keep the proof tractable when taking limit. Specifically, the first condition is to avoid the case that KL divergence being possibly infinity on the boundary of the probability simplex. The second condition means the threshold function $\lambda$ should be continuous even on the boundary of its domain.
The detailed proof of \Cref{prop:composite_achievability} can be found in \Cref{sec:achievability}.

Let us now investigate when the converse and achievability can coincide. In the following lemma, we first show that if the two conditions in \Cref{prop:composite_achievability} hold, $\Omega(\bm P_1)$ 
can be replaced with another set $\Omega' (\bm P_1)$ which differs from $\bar\Omega(\bm P_1)$ 
only in substituting $\bar g^{\bm P_1}(\bm Q)$ with $g^{\bm P_1}(\bm Q)$. The proof is given in Appendix~\ref{app:rewrite_Omega}.
\begin{lemma}\label{lemma:rewrite_Omega}
    Suppose the two conditions in \Cref{prop:composite_achievability} hold, then 
    \begin{equation}
        \inf_{\bm Q\in\Omega(\bm P_1)} \sum_{i=1}^s \alpha_i \KLD{Q_i}{P_{1,i}} = \inf_{\bm Q\in\Omega'(\bm P_1)} \sum_{i=1}^s \alpha_i \KLD{Q_i}{P_{1,i}},
    \end{equation}
    where $\Omega'(\bm P_1) = \lbp\bm Q\in \mcal{S}(\bm P_1) \mv g^{\bm P_1}(\bm Q) \leq 0\rbp$.
\end{lemma}
To proceed, further note that $g^{\bm P_1}(\bm Q)$ in \eqref{eq:def_g_P1} can be rewritten as 
\begin{equation}\label{eq:def_g_P1_ext}
g^{\bm P_1}(\bm Q) = \inf_{\bm P_0'\in\mscr P_0\cap \mcal{S}(\bm P_1)} \Big( \sum_{i=1}^{\ell} \alpha_i \KLD{Q_i}{P_{0,i}'} - \bar{\lambda}(\bm P_0') \Big).
\end{equation}
In general, $\mscr P_0\cap \mcal{S}(\bm P_1) \subseteq \closure{\mscr P_0\cap \mcal{S}(\bm P_1)} \subseteq \closure{\mscr P_0}\cap \mcal{S}(\bm P_1)$. Hence,
\begin{align}
g^{\bm P_1}(\bm Q) &= \inf_{\bm P_0'\in\mscr P_0\cap \mcal{S}(\bm P_1)} \Big( \sum_{i=1}^{\ell} \alpha_i \KLD{Q_i}{P_{0,i}'} - \bar{\lambda}(\bm P_0') \Big) \\
\ge \bar g^{\bm P_1}(\bm Q) &= \inf_{\bm P_0'\in\closure{\mscr P_0}\cap \mcal{S}(\bm P_1)} \Big( \sum_{i=1}^{\ell} \alpha_i \KLD{Q_i}{P_{0,i}'} - \bar\lambda(\bm P_0') \Big).
\end{align}

As we will show in \Cref{lemma:inf_continuity}, since the KL divergence and $\bar{\lambda}$ are continuous, taking infimum over $\mscr P_0\cap \mcal{S}(\bm P_1)$ in \eqref{eq:def_g_P1_ext} is the same as taking infimum over $\closure{\mscr P_0\cap \mcal{S}(\bm P_1)}$. Thus we can obtain a sufficient condition for which the converse and achievability match, that is, $\closure{\mscr P_0\cap \mcal{S}(\bm P_1)} = \closure{\mscr P_0}\cap \mcal{S}(\bm P_1)$. 
This sufficient condition will be leveraged to prove the main results for binary classification problems in \Cref{sec:results}. 

Let us summarize the above sufficient condition in the following proposition for convenience.

\begin{proposition}[Characterization of the optimal type-II error exponent]\label{prop:match}
Suppose the two conditions in \Cref{prop:composite_achievability} hold and $\closure{\mscr P_0\cap \mcal{S}(\bm P_1)} = \closure{\mscr P_0}\cap \mcal{S}(\bm P_1)$, then the optimal type-II error exponent of all tests that satisfy the universality constraints in \Cref{prop:composite_converse,prop:composite_achievability} is $e_1^*(\bm P_1)$ in \eqref{eq:Expt_CHT_converse}, and it can be achieved by a sequence of tests that do not depend on the actual $\bm{P}_1$.
\end{proposition}

Before we proceed to the specialization to the binary classification problem, let us close this sub-section by taking a look at a special case to better understand the general result. When $\ell=0$, all the samples arrive sequentially with a fixed ratio. In this case, 
it can be shown that the optimal type-II error exponent can be achieved and simplified, even if $\closure{\mscr P_0\cap \mcal{S}(\bm P_1)} \neq \closure{\mscr P_0}\cap \mcal{S}(\bm P_1)$.
\begin{corollary}[Fully-sequential composite hypothesis testing]\label{cor:fully-seq_CHT}
    Suppose the two conditions in \Cref{prop:composite_achievability} hold and $\ell=0$, then the  optimal type-II error exponent of all tests that satisfy the universality constraints in \Cref{prop:composite_converse,prop:composite_achievability} is
    \begin{equation}
        e_1^*(\bm P_1)=\inf_{\bm Q\in\mscr P_0\cup \Gamma_0} \sum_{i=1}^s \alpha_i \KLD{Q_i}{P_{1,i}},
    \end{equation}
    and it can be achieved by a sequence of tests that do not depend on the actual $\bm{P}_1$.
\end{corollary}
\begin{proof}
    Since $\mscr P_0$ and $\mscr P_1$ are disjoint, we have $\mcal{S}(\bm P_1)=\{\bm P_1\}$ and $\mscr P_0\cap \mcal{S}(\bm P_1)=\emptyset$. 
    If $\closure{\mscr P_0}\cap \mcal{S}(\bm P_1)=\emptyset$, simply apply \Cref{prop:match}. Otherwise, $\closure{\mscr P_0}\cap \mcal{S}(\bm P_1)=\{\bm P_1\}$, and thus $\bm P_1\in\closure{\mscr P_0}$. With the help of \Cref{lemma:inf_continuity} (presented later in \Cref{sec:achievability}), we can observe that
    \begin{align}
        0\leq e_1^*(\bm P_1) &\leq \inf_{\bm Q\in\mscr P_0\cup \Gamma_0} \sum_{i=1}^s \alpha_i \KLD{Q_i}{P_{1,i}}\\
        &\leq\inf_{\bm Q\in\mscr P_0} \sum_{i=1}^s \alpha_i \KLD{Q_i}{P_{1,i}}\\
        &=\inf_{\bm Q\in\closure{\mscr P_0}} \sum_{i=1}^s \alpha_i \KLD{Q_i}{P_{1,i}}=0,
    \end{align}
    where the last equality holds because $\bm P_1\in\closure{\mscr P_0}$.
\end{proof}
\begin{remark}
    In the fixed-length setup ($\ell=s$), the result in \cite{LevitanMerhav_02} shows matching converse and achievability without the additional conditions in \Cref{prop:composite_achievability}. It is worth notice that if $\mscr P_0$ is bounded away from the boundary of $\mscr D$, then the optimal type-II error exponent can be rewritten as the following:
    \begin{equation}
        e^*_{1,\mathsf{fix}}(\bm P_1) = \inf_{\substack{\bm Q\in\mscr D \\ g_1(\bm Q)\leq 0}} \sum_{i=1}^s \alpha_i \KLD{Q_i}{P_{1,i}}.
    \end{equation}
    This result follows from the same argument as in \Cref{lemma:rewrite_Omega} and some observations in \Cref{sec:discussion}.
\end{remark}

\section{Main Results on Universal Classification with Empirically Observed Statistics}\label{sec:results}
Now we switch our focus back to the binary classification problem and present the optimal type-II error exponents in different sequential setups.
To present the results, we first introduce two divergences. 
For $P,Q\in\mcal{P}(\mcal{X})$, the R\'{e}nyi Divergence of order $\frac{\alpha}{1+\alpha}$ of $P$ from $Q$ can be expressed as
\begin{equation}
    \RenyiD{P}{Q}{\frac{\alpha}{1+\alpha}} = \min_{V\in\mcal{P}(\mcal{X})} \{\alpha\KLD{V}{P} + \KLD{V}{Q} \},\label{eq:Renyi}
\end{equation}
and the minimizer is $V^*$, where
\begin{equation}\label{eq:Renyi_minimizer}
    V^*(x)=\frac{P(x)^{\frac{\alpha}{1+\alpha}} Q(x)^{\frac{1}{1+\alpha}}}{\sum_{x'\in\mcal X}P(x')^{\frac{\alpha}{1+\alpha}} Q(x')^{\frac{1}{1+\alpha}}} \quad\text{for } x\in\mcal X.
\end{equation}
The $\alpha$-weighted generalized Jensen-Shannon (GJS) divergence of $Q$ from $P$ is defined as
\begin{equation}
    \textstyle\GJS{P}{Q}{\alpha} = \alpha\KLD{P}{\frac{\alpha P+Q}{\alpha+1}} + \KLD{Q}{\frac{\alpha P+Q}{\alpha+1}}.
\end{equation}
We start with the \fullseq\ setup where all the samples arrive sequentially.
\begin{theorem}[\fullseq]\label{thm:complete_fully-seq}
    Given $\lambda:\distset\to (0,\infty)$ and let $\{\Phi_n\}$ be a sequence of \fullseq\ tests such that for any underlying distributions $(P_0,P_1)\in\distset$,
    \begin{itemize}
        \item for each $\Phi_n = (\tau_n,\delta_n)$, %the expected stopping time 
        $\mbb E_\theta[\tau_n|P_0,P_1]\leq n\ \forall\,\theta\in\{0,1\}$,
        \item %type-I error exponent satisfies 
        $\eT{0}(P_0,P_1)\geq \lambda(P_0,P_1)$.
    \end{itemize}
    Then for any $(P_0,P_1)\in\distset$, $\eT{1}(P_0,P_1)$ is upper bounded by
    \begin{equation}
         \eseq(P_0,P_1):= \min\lbp \RenyiD{P_0}{P_1}{\frac{\alpha}{1+\alpha}},\ \kappa(P_0,P_1)\rbp,\label{eq:fully-sequential_converse}
    \end{equation}
    where $\kappa(P_0,P_1)=$
    \begin{align}
\inf_{\substack{(Q_0,Q_1)\in\distset\\g_1(Q_1, Q_0,Q_1)<0}} \Big( \KLD{Q_1}{P_1} + \alpha\KLD{Q_0}{P_0} + \beta\KLD{Q_1}{P_1} \Big),
    \end{align}
    and $g_1(Q,Q_0,Q_1)=$
    \begin{align}
    \inf_{(P_0',P_1')\in\distset} \Big(\KLD{Q}{P_0'} &+ \alpha\KLD{Q_0}{P_0'}\nonumber\\
         &+ \beta\KLD{Q_1}{P_1'} - \lambda(P_0',P_1')\Big).
    \end{align}
    Moreover, the upper bound $\eseq(P_0,P_1)$ can be achieved simultaneously for all $(P_0,P_1)\in\distset$.
\end{theorem}
\Cref{thm:complete_fully-seq} is proved as a corollary of \Cref{prop:match} by identifying it as a special case of the general composite hypothesis testing problem. Detailed derivation can be found in \Cref{subsec:proof_of_thm1-2}.
In \eqref{eq:fully-sequential_converse}, $\RenyiD{P_0}{P_1}{\frac{\alpha}{1+\alpha}}$ corresponds to the region of true distributions under ground truth $0$, that is $\mscr P_0$ in the general setup. 
On the other hand, $\kappa(P_0,P_1)$  corresponds to the region of true distributions under ground truth $1$ that falls in the $\lambda$-balls centered at distributions under ground truth $0$, that is $\Gamma_0$ in the general setup.

Next we look at the \semione\ setup where the testing samples $X_k$'s arrive sequentially and the numbers of $P_0$/$P_1$-training samples are fixed.
\begin{theorem}[\semione]\label{thm:complete_semi-seq}
    Given $\lambda:\distset\to (0,\infty)$ and let $\{\Phi_n\}$ be a sequence of \semione\ tests such that for any underlying distributions $(P_0,P_1)\in\distset$,
    \begin{itemize}
        \item for each $\Phi_n = (\tau_n,\delta_n)$, %the expected stopping time 
        $\mbb E_\theta[\tau_n|P_0,P_1]\leq n\ \forall\,\theta\in\{0,1\}$,
        \item %type-I error exponent satisfies 
        $\eT{0}(P_0,P_1)\geq \lambda(P_0,P_1)$.
    \end{itemize}
    Then for any $(P_0,P_1)\in\distset$, $\eT{1}(P_0,P_1)$ is upper bounded by
    \begin{align}
        &\esemione(P_0,P_1):=\nonumber\\
        &\hspace{1cm}\min\lbp \RenyiD{P_0}{P_1}{\frac{\alpha}{1+\alpha}},\ \kappa(P_0,P_1),\ 
        \mu(P_0,P_1)\rbp,\label{eq:semi-sequential_converse}
    \end{align}
    where
    \begin{equation}
        \mu(P_0,P_1) = \inf_{\substack{Q_0,Q_1\in\mcal{P}(\mcal X)\\g(Q_0,Q_1) < 0}} \Big(\alpha\KLD{Q_0}{P_0} + \beta\KLD{Q_1}{P_1}\Big),
    \end{equation}
    and $g(Q_0,Q_1)=$
    \begin{align}
    \inf_{P_1'\in\mcal{P}_\varepsilon\setminus \{P_1\}} \Big(\alpha\KLD{Q_0}{P_1} + \beta\KLD{Q_1}{P_1'} - \lambda(P_1,P_1')\Big).
    \end{align}
    Moreover, if $\lambda$ satisfies \Cref{assumption:lambda}, then $\esemione(P_0,P_1)$ can be achieved simultaneously for all $(P_0,P_1)\in\distset$.
\end{theorem}

The term $\mu(P_0,P_1)$ comes from the limitation by the fixed-length training sequences, which marks the difference between \Cref{thm:complete_semi-seq} and the \fullseq\ result in \Cref{thm:complete_fully-seq}.
Likewise, \Cref{thm:complete_semi-seq} is proved as a corollary of \Cref{prop:match} in \Cref{subsec:proof_of_thm1-2}.

The next theorem summarizes the result for the \semitwo\ setup where the testing samples have fixed length and the training samples arrive sequentially.
\begin{theorem}[\semitwo]\label{thm:complete_semi-seq-2}
    Given $\lambda:\distset\to (0,\infty)$ and let $\{\Phi_n\}$ be a sequence of \semitwo\ tests such that for any underlying distributions $(P_0,P_1)\in\distset$,
    \begin{itemize}
        \item for each $\Phi_n = (\tau_n,\delta_n)$, %the expected stopping time 
        $\mbb E_\theta[\tau_n|P_0,P_1]\leq n\ \forall\,\theta\in\{0,1\}$,
        \item %type-I error exponent satisfies 
        $\eT{0}(P_0,P_1)\geq \lambda(P_0,P_1)$.
    \end{itemize}
    Then for any $(P_0,P_1)\in\distset$, $\eT{1}(P_0,P_1)$ is upper bounded by
    \begin{align}
         &\esemitwo(P_0,P_1):= \nonumber\\
         &\hspace{1cm}\min\lbp \RenyiD{P_0}{P_1}{\frac{\alpha}{1+\alpha}},\ \kappa(P_0,P_1),\ \nu(P_0,P_1)\rbp.\label{eq:semi-sequential-2_converse}
    \end{align}
    where
    \begin{equation}
        \nu(P_0,P_1) = \inf_{\substack{Q\in\mcal{P}(\mcal X)\\ \KLD{Q}{P_0}<\lambda(P_0,P_1)}}\KLD{Q}{P_1}.
    \end{equation}
    Moreover, if $\lambda$ satisfies \Cref{assumption:lambda}, then $\esemitwo(P_0,P_1)$ can be achieved simultaneously for all $(P_0,P_1)\in\distset$.
\end{theorem}
Similar as in \Cref{thm:complete_semi-seq}, the third term comes from the limitation limitation by the fixed-length testing sequences. 
Notice that this term $\nu(P_0,P_1)$ is exactly the optimal type-II error exponent for the binary hypothesis problem, given the type-I error exponent $e_0(P_0,P_1)\geq \lambda(P_0,P_1)$. This upper bound is natural because the test cannot do better than when the underlying distributions $(P_0,P_1)$ are known. \Cref{thm:complete_semi-seq-2} is proved as a corollary of \Cref{prop:match} in \Cref{subsec:proof_of_thm3}.

Finally, for comparison, we also restate the \fixed\ result in \cite{LevitanMerhav_02}, restricting the underlying distributions within $\distset$.
\begin{theorem}[\fixed \cite{LevitanMerhav_02}]\label{thm:complete_fixed-length}
    Given $\lambda:\distset\to (0,\infty)$ and let $\{\Phi_n\}$ be a sequence of \fixed\ tests such that for any underlying distributions $(P_0,P_1)\in\distset$, the type-I error exponent satisfies $\eT{0}(P_0,P_1)\geq \lambda(P_0,P_1)$. Then for any $(P_0,P_1)\in\distset$, $\eT{1}(P_0,P_1)$ is upper bounded by
    \begin{align}
        &\efix(P_0,P_1):= \nonumber\\
        &\inf_{\substack{Q,Q_0,Q_1\in\mcal P(\mcal X)\\g_1(Q,Q_0,Q_1) < 0}} \Big(\KLD{Q}{P_1} + \alpha\KLD{Q_0}{P_0} + \beta\KLD{Q_1}{P_1}\Big),
    \end{align}
    Moreover, the upper bound $\efix(P_0,P_1)$ can be achieved simultaneously for all $(P_0,P_1)\in\distset$.
\end{theorem}
\begin{remark}
    For \Cref{thm:complete_fully-seq} and \Cref{thm:complete_fixed-length}, the results still hold even if we replace the distribution set $\distset$ by $\mcal D=\lbp(P,Q)\mv \supp{P}=\supp{Q}=\mcal{X},\ P\neq Q\rbp$. 
\end{remark}

Next we prove \Cref{thm:complete_fully-seq,thm:complete_semi-seq,thm:complete_semi-seq-2} as corollaries of \Cref{prop:match}. 

\subsection{Proof of \Cref{thm:complete_fully-seq,thm:complete_semi-seq}}\label{subsec:proof_of_thm1-2}
Arrange the three sequences by $T_0, T_1, X$. Then the binary classification problem can be viewed as a special case of the composite hypothesis testing problem by taking $s=3$, $\mcal X_i=\mcal X$ for $i=1,2,3$, $\alpha_1=\alpha$, $\alpha_2=\beta$, and $\alpha_3=1$. The two hypothesis sets are 
\begin{align*}
    \mscr P_0&=\lbp(P_0,P_1,P_0)\mv (P_0,P_1)\in\distset\rbp\\
    \text{and }\mscr P_1&=\lbp(P_0,P_1,P_1)\mv (P_0,P_1)\in\distset\rbp.
\end{align*}
For the \fullseq\ setup, let $\ell=0$, and for the \semione\ setup, let $\ell=2$.
Note that we use $\lambda$ with an abuse of notation. We could then apply \Cref{prop:composite_converse} and \Cref{prop:composite_achievability}.
It can be easily verified that $\closure{\mscr P_0}\cap \mcal{S}(\bm P_1) = \closure{\mscr P_0\cap \mcal{S}(\bm P_1)}$ for all $\bm P_1\in\mscr P_1$ in both setups, so 
\Cref{prop:match} is also applicable. 
Specifically, notice that $\closure{\mscr P_0} = \lbp(P_0,P_1,P_0)\mv P_0,P_1\in\mcal P_\varepsilon\rbp$, and the only difference from $\mscr P_0$ is that $P_0$ and $P_1$ can be the same. For the \fullseq\ setup, given any $\bm P_1=(P_0,P_1,P_1)\in\mscr P_1$, we have $P_0\neq P_1$ and hence $\closure{\mscr P_0}\cap \mcal{S}(\bm P_1)=\emptyset$. Moreover, one can observe that $\Omega(\bm P_1)=\emptyset$ for all $\bm P_1=(P_0,P_1,P_1)\in\mscr P_1$. For the \semione\ setup, given any $\bm P_1=(P_0,P_1,P_1)\in\mscr P_1$, $\mcal{S}(\bm P_1)=\lbp(Q_1,Q_2,P_1)\mv Q_1,Q_2\in\mcal P(\mcal X)\rbp$. Observe that $\mscr P_0\cap \mcal{S}(\bm P_1)=\lbp(P_1,Q_2,P_1)\mv Q_2\in\mcal P_\varepsilon\setminus\{P_1\}\rbp$ and $\closure{\mscr P_0}\cap \mcal{S}(\bm P_1)=\lbp(P_1,Q_2,P_1)\mv Q_2\in\mcal P_\varepsilon\rbp$. The closure of the former set is clearly the same as the latter since $Q_2$ can be arbitrarily close to $P_1$.

The bounds in \Cref{thm:complete_fully-seq} and \Cref{thm:complete_semi-seq} can be obtained via simple calculations. The three terms $\RenyiD{P_0}{P_1}{\frac{\alpha}{1+\alpha}}$, $\kappa(P_0,P_1)$ and $\mu(P_0,P_1)$ correspond to the divergences between the real underlying distribution $\bm P_1=(P_0,P_1,P_1)$ and the three sets $\mscr P_0$, $\Gamma_0$ and $\Omega(\bm P_1)$, respectively. 
In particular, to obtain the R\'enyi divergence $\RenyiD{P_0}{P_1}{\frac{\alpha}{1+\alpha}}$, 
\begin{align}
    &\inf_{\bm Q\in\mscr P_0} \Big( \alpha \KLD{Q_1}{P_0} + \beta\KLD{Q_2}{P_1} + \KLD{Q_3}{P_1}\Big) \\
    &= \inf_{P_0',P_1'\in\mcal P_\varepsilon} \Big( \alpha \KLD{P_0'}{P_0} + \beta\KLD{P_1'}{P_1} + \KLD{P_0'}{P_1}\Big)\\
    &= \inf_{P_0'\in\mcal P_\varepsilon} \Big( \alpha \KLD{P_0'}{P_0} + \KLD{P_0'}{P_1}\Big)
    = \RenyiD{P_0}{P_1}{\frac{\alpha}{1+\alpha}}.
\end{align}
The last equality follows from \eqref{eq:Renyi}, where the R\'enyi divergence is expressed as a minimization problem. By H\"older's inequality, since $P_0,P_1\in\mcal P_\varepsilon$, the minimizer described in \eqref{eq:Renyi_minimizer} is also in $\mcal P_\varepsilon$. Thus taking the infimum over $\mcal P_\varepsilon$ instead of $\mcal P(\mcal X)$ in \eqref{eq:Renyi} still gives the R\'enyi divergence.
%The term $\kappa(P_0,P_1)$ is rather easy to obtain. \commentcfl{In fact, it is not easy to explicitly write down the process due to abuse of notation.} 
To obtain $\kappa(P_0,P_1)$, first observe that for $\bm Q=(Q_1,Q_2,Q_3)$, %which corresponds to the set $\Gamma_0$.
\begin{align}
    g_1(\bm Q) &= \inf_{\bm P_0'\in\mscr P_0}\Big( \alpha \KLD{Q_1}{P_{0,1}'} + \beta\KLD{Q_2}{P_{0,2}'} \nonumber\\
    &\hspace{3cm}+ \KLD{Q_3}{P_{0,3}'} - \lambda(\bm P_0')\Big)\\
    &= \inf_{(P_0',P_1')\in\distset} \Big(\alpha\KLD{Q_1}{P_0'} + \beta\KLD{Q_2}{P_1'}  \nonumber\\
    &\hspace{2.8cm}+\KLD{Q_3}{P_0'} - \lambda(P_0',P_1')\Big)\\
    &=g_1(Q_3,Q_1,Q_2) \quad\text{defined in \Cref{thm:complete_fully-seq}.}
\end{align}
As a result, 
\begin{align}
    &\inf_{\bm Q\in\Gamma_0} \Big( \alpha \KLD{Q_1}{P_0} + \beta\KLD{Q_2}{P_1} + \KLD{Q_3}{P_1}\Big) \\
    &=\hspace{-0.2cm}\inf_{\bm Q\in\mscr P_1: g_1(\bm Q)<0} \Big( \alpha \KLD{Q_1}{P_0} + \beta\KLD{Q_2}{P_1} + \KLD{Q_3}{P_1}\Big) \\
    &= \hspace{-0.2cm}\inf_{\substack{(Q_0',Q_1')\in\distset\\g_1(Q_1', Q_0',Q_1')<0}} \hspace{-0.4cm}\Big( \alpha \KLD{Q_0'}{P_0} + \beta\KLD{Q_1'}{P_1} + \KLD{Q_1'}{P_1}\Big)\\
    &= \kappa(P_0,P_1).
\end{align}
For the \semione\ setup, observe that for $\bm Q=(Q_1,Q_2,Q_3)\in \mcal{S}(\bm P_1)$, we have $Q_3=P_1$ and
\begin{align}
    &\quad\ g^{\bm P_1}(\bm Q) \nonumber\\
    &= \hspace{-0.1cm}\inf_{\bm P_0'\in\mscr P_0\cap \mcal{S}(\bm P_1)} \hspace{-0.1cm}\Big( \alpha \KLD{Q_1}{P_{0,1}'} + \beta\KLD{Q_2}{P_{0,2}'} - \lambda(\bm P_0') \Big)\\
    &= \hspace{-0.1cm}\inf_{P_{0,2}'\in\mcal{P}_\varepsilon\setminus \{P_1\}} \hspace{-0.1cm}\Big( \alpha \KLD{Q_1}{P_1} + \beta\KLD{Q_2}{P_{0,2}'} - \lambda(P_1, P_{0,2}') \Big) \\
    &= g(Q_1,Q_2).
\end{align}
Hence,
\begin{align}
    &\inf_{\bm Q\in\Omega(\bm P_1)} \Big( \alpha \KLD{Q_1}{P_0} + \beta\KLD{Q_2}{P_1} + \KLD{Q_3}{P_1}\Big) \\
    &= \inf_{\substack{\bm Q\in \mcal{S}(\bm P_1)\\g^{\bm P_1}(\bm Q)<0}} \Big( \alpha \KLD{Q_1}{P_0} + \beta\KLD{Q_2}{P_1} + \KLD{Q_3}{P_1}\Big) \\
    &= \inf_{\substack{Q_1,Q_2\in\mcal{P}(\mcal X)\\g(Q_0,Q_1) < 0}} \Big(\alpha\KLD{Q_1}{P_0} + \beta\KLD{Q_2}{P_1}\Big) = \mu(P_0,P_1).
\end{align}

\subsection{Proof of \Cref{thm:complete_semi-seq-2}}\label{subsec:proof_of_thm3}
Arrange the three sequences by $X, T_0, T_1$, and take $s=3$, $\ell=1$, $\mcal X_i=\mcal X$ for $i=1,2,3$, $\alpha_1=1$, $\alpha_2=\alpha$, and $\alpha_3=\beta$. The two hypothesis sets are 
\begin{align*}
    \mscr P_0&=\lbp(P_0,P_0,P_1)\mv (P_0,P_1)\in\distset\rbp\\
    \text{and }\mscr P_1&=\lbp(P_1,P_0,P_1)\mv (P_0,P_1)\in\distset\rbp.
\end{align*}
When $\bm P_1=(P_1,P_0,P_1)$, we have $\closure{\mscr P_0}\cap \mcal{S}(\bm P_1) = \closure{\mscr P_0\cap \mcal{S}(\bm P_1)}=\{(P_0,P_0,P_1)\}$. The result follows immediately from \Cref{prop:composite_converse},  \Cref{prop:composite_achievability}, and \Cref{prop:match}. Specifically, $\RenyiD{P_0}{P_1}{\frac{\alpha}{1+\alpha}}$ and $\kappa(P_0,P_1)$ can be obtained using the same arguments as above with a permutation of the three sequences. As for $\nu(P_0,P_1)$, consider $\bm P_1=(P_1,P_0,P_1)$ and observe that for $\bm Q=(Q_1,Q_2,Q_3)\in \mcal{S}(\bm P_1)$, we have $Q_2=P_0, Q_3=P_1$. Also, 
\begin{align}
    g^{\bm P_1}(\bm Q) &= \inf_{\bm P_0'\in\mscr P_0\cap \mcal{S}(\bm P_1)} \Big( \KLD{Q_1}{P_{0,1}'} - \lambda(\bm P_0') \Big)\\
    &= \KLD{Q_1}{P_0} - \lambda(P_0,P_1),
\end{align}
as there is only one element $(P_0,P_0,P_1)$ in $\mscr P_0\cap \mcal{S}(\bm P_1)$.
Thus,
\begin{align}
    &\inf_{\bm Q\in\Omega(\bm P_1)} \Big( \KLD{Q_1}{P_1} + \alpha\KLD{Q_2}{P_0} + \beta\KLD{Q_3}{P_1} \Big) \\
    &= \inf_{\substack{\bm Q\in \mcal{S}(\bm P_1)\\g^{\bm P_1}(\bm Q)<0}} \Big( \KLD{Q_1}{P_1} + \alpha\KLD{Q_2}{P_0} + \beta\KLD{Q_3}{P_1}\Big) \\
    &= \inf_{\substack{Q_1\in\mcal{P}(\mcal X)\\ \KLD{Q_1}{P_0} < \lambda(P_0,P_1)}}\KLD{Q_1}{P_1} = \nu(P_0,P_1).
\end{align}

\subsection{Comparison}\label{subsec:compare}
Here we take a closer look at these results and compare the optimal achievable error exponents for \fullseq, \semione, \semitwo\ and \fixed\ settings to show the benefit of sequentiality. 
Observe that for any $(P_0,P_1)\in\distset$,
\begin{align}
\efix(P_0,P_1) \leq \esemione(P_0,P_1) \leq \eseq(P_0,P_1),\label{ineq:compare1}\\
\efix(P_0,P_1) \leq \esemitwo(P_0,P_1) \leq \eseq(P_0,P_1).\label{ineq:compare2}
\end{align}
The first inequality of \eqref{ineq:compare1} holds because a \fixed\ test is a special case of \semione\ tests. However, \semione\ tests are not a special case of \fullseq\ tests as the whole training sequences are observed even when stopping before time $n$. Instead, the second inequality of \eqref{ineq:compare1} is obtained by directly comparing the expressions defined in \eqref{eq:semi-sequential_converse} and \eqref{eq:fully-sequential_converse}. Following similar arguments, we get the inequalities \eqref{ineq:compare2} for \semitwo\ tests.
A natural question is whether these inequalities are strict. In the following, we analytically address the question under two representative classes of $\lambda$: $\lambda \equiv \lambda_0$ being a constant in \Cref{subsec:constant_lambda}, and $\lambda$ allowing \emph{efficient} universal tests in \Cref{subsec:efficient_tests_lambda}.

\subsubsection{Constant Constraint}\label{subsec:constant_lambda}
Take $\lambda(P_0,P_1)\equiv\lambda_0$ for some $\lambda_0>0$. It is clear that such $\lambda$ satisfies \Cref{assumption:lambda}. For the \fixed\ case, this is exactly the setup in Ziv \cite{Ziv_88} and Gutman \cite{Gutman_89} (generalized Neyman-Pearson criterion). 
The following proposition shows that \semione\ and \fullseq\ tests achieve the same optimal error exponents under constant constraint.
\begin{proposition}\label{prop:constant_compare}
    If $\lambda(P_0,P_1) \equiv \lambda_0$ for some $\lambda_0>0$, then $\kappa(P_0,P_1) \leq \mu(P_0,P_1)$ for all $(P_0,P_1)\in\distset$, and hence
    \begin{align}
        &\quad\ \esemione(P_0,P_1)=\eseq(P_0,P_1)\\
        &=\min\lbp \RenyiD{P_0}{P_1}{\frac{\alpha}{1+\alpha}},\ \kappa(P_0,P_1)\rbp\quad\forall\,(P_0,P_1)\in\distset.
    \end{align}
\end{proposition}
\Cref{prop:constant_compare} shows that there is no additional gain due to sequentiality in taking training samples in the constant constraint case. The proof is in Appendix~\ref{app:constant_compare}. Next we provide some numerical comparisons. 
In resemble to the error exponent trade-off in the binary hypothesis testing problem, 
choose a pair of distributions $(P_0^*,P_1^*)\in\distset$ and plot the upper bounds on the type-II error exponent versus the type-I error exponent constraint $\lambda(P_0^*,P_1^*)$. 

\begin{figure*}[ht]
    \centering
    \includegraphics[scale=0.57]{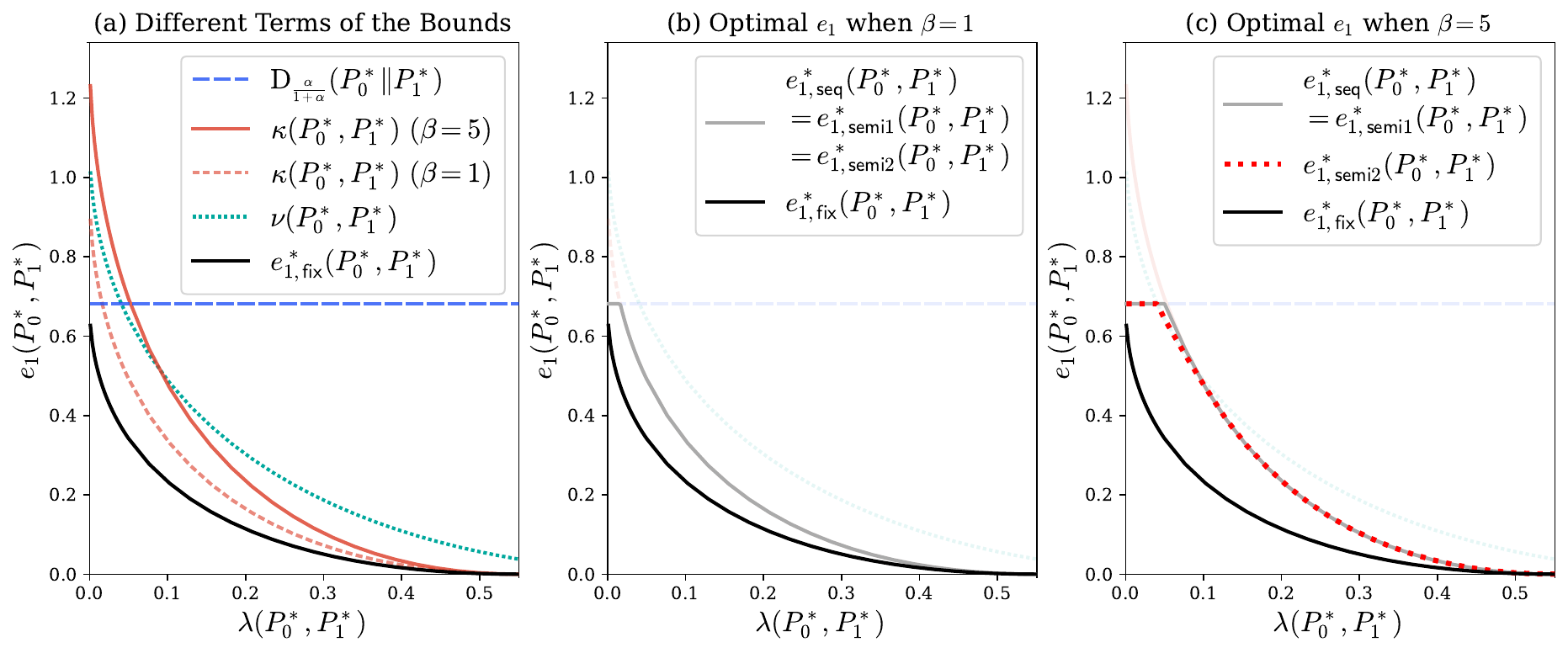}
    \caption{The optimal type-II error exponents under constant type-I error exponent constraint $\lambda(P_0,P_1)\equiv\lambda_0$. Fix $\mcal X = \{0,1\}$, $\varepsilon = 0.01$, $\alpha=2$, and choose $P_0^*=[0.6, 0.4]$, $P_1^*=[0.1, 0.9]$.
    Let $\lambda_0$ increase from $0.001$ to $\GJS{P_0^*}{P_1^*}{\alpha}$ to obtain a curve. Note that in (a), $\kappa(P_0^*,P_1^*)$ is plot under two different values of $\beta$, whereas $\efix(P_0^*,P_1^*)$ and $\RenyiD{P_0^*}{P_1^*}{\frac{\alpha}{1+\alpha}}$ are independent of $\beta$.
    } 
    \label{fig:compare_const}
\end{figure*}

In \Cref{fig:compare_const}, there is a trade-off between the two error exponents for every setup. For \semione\ and \fullseq\ tests, the optimal type-II error exponent is the minimum of the red and blue curve in (a). We also plot them directly in (b) and (c). Comparing to \fixed\ test, the benefit of sequentiality is more significant when $\beta$ is larger as shown in (b) and (c). 
This can also be observed analytically, and the results are summarized later in \Cref{prop:constant_strict_gain}.
On the other hand, \Cref{fig:compare_const} shows that \semitwo\ tests may perform strictly worse than \semione\ and \fullseq\ tests sometimes, depending on the value of $\beta$. Recall that $\nu(P_0,P_1)$ does not depend on $\alpha$ nor $\beta$. The numerical example suggests that it is not possible to get a result in the same form of \Cref{prop:constant_compare} for the \semitwo\ setup. The solution of which $\nu$ intersects with $\kappa$ does not seem easy to obtain analytically either.

\begin{proposition}\label{prop:constant_strict_gain}
    If $\lambda(P_0,P_1) \equiv \lambda_0$ for some $\lambda_0>0$, then $\efix(P_0,P_1)$ and $\kappa(P_0,P_1)$ can be written in the same form, with coefficient $1$ replaced by $1+\beta$:
    \begin{align}
        \efix(P_0,P_1)&=\hspace{-0.2cm}\inf_{\substack{Q,Q_0\in\mcal P_\varepsilon\\\GJS{Q_0}{Q}{\alpha}\leq\lambda_0}} \hspace{-0.3cm}\KLD{Q}{P_1} + \alpha\KLD{Q_0}{P_0},\\
        \kappa(P_0,P_1)&=\nonumber\\
        &\hspace{-0.8cm}\inf_{\substack{Q_0,Q_1\in\mcal P_\varepsilon\\\GJS{Q_0}{Q_1}{\alpha}\leq\lambda_0}} \hspace{-0.3cm}(1+\beta)\KLD{Q_1}{P_1} + \alpha\KLD{Q_0}{P_0}.
    \end{align}
    Moreover, for a pair of distributions $(P_0,P_1)\in\distset$:
    \begin{itemize}
        \item If $\GJS{P_0}{P_1}{\alpha}\leq \lambda_0$, then
        \begin{align}
            &\efix(P_0,P_1)=\esemione(P_0,P_1)\\
            &=\esemitwo(P_0,P_1)=\eseq(P_0,P_1)=0
        \end{align}
        \item If $\GJS{P_0}{P_1}{\alpha}>\lambda_0$, then $\kappa(P_0,P_1)$ is strictly increasing in $\beta$ and
        \begin{equation}
            0<\efix(P_0,P_1)<\esemione(P_0,P_1)=\eseq(P_0,P_1).
        \end{equation}
    \end{itemize}
\end{proposition}
As a result, unlike in the \fixed\ setup, the training sequence $T_1$ can improve the error exponent in \semione\ and \fullseq\ setups, and the gain is more significant as $\beta$ gets larger. 
The proof of \Cref{prop:constant_strict_gain} can be found in Appendix~\ref{app:constant_strict_gain}. 

\subsubsection{Efficient Tests}\label{subsec:efficient_tests_lambda}
Following \cite{LevitanMerhav_02, HsuWang_22}, we consider a special class of tests.
A sequence of tests is said to be \emph{efficient} \cite{LevitanMerhav_02} or \emph{universally exponentially consistent} \cite{LiVeeravalli_17} if the error probabilities decay to zero exponentially for all the underlying distributions as the number of samples goes to infinity. 
Under our problem formulation, the formal definition of efficient tests is provided as follows. 
\begin{definition}[Efficient Tests]
    Given $\lambda:\distset\to (0,\infty)$ and let $\{\Phi_n\}$ be a sequence of tests such that for any underlying distributions $(P_0,P_1)\in\distset$,
    \begin{itemize}
        \item for each $\Phi_n = (\tau_n,\delta_n)$, 
        $\mbb E_\theta[\tau_n|P_0,P_1]\leq n\ \forall\,\theta\in\{0,1\}$,
        \item 
        $\eT{0}(P_0,P_1)\geq \lambda(P_0,P_1)$.
    \end{itemize}
    We say $\{\Phi_n\}$ is efficient if the type-II error exponent $e_1(P_0,P_1)>0$ for all $(P_0,P_1)\in\distset$.
\end{definition}

In order to have tests satisfying the universality constraints and being efficient, the constraint function $\lambda$ has to satisfy certain conditions, and this will lead to further simplification of the optimal error exponents.
In the following, we will introduce a necessary condition and the consequent results. Then we state a sufficient condition and show some numerical examples. 
\begin{proposition}\label{prop:kappa=infty}
    Given $\lambda:\distset\to (0,\infty)$, if there exists a sequence of tests satisfying the universality constraints on the expected stopping time and the type-I error exponent, 
    furthermore being efficient; then for all $(P_0,P_1)\in\distset$, $\lambda(P_0,P_1)\leq \RenyiD{P_1}{P_0}{\frac{\beta}{1+\beta}}$, and hence $\kappa(P_0,P_1)=\infty$. Thus  
    \begin{align}
        \eseq(P_0,P_1) &= \RenyiD{P_0}{P_1}{\frac{\alpha}{1+\alpha}},\\
        \esemione(P_0,P_1)&= \min\lbp \RenyiD{P_0}{P_1}{\frac{\alpha}{1+\alpha}},\ \mu(P_0,P_1)\rbp,\\
        \esemitwo(P_0,P_1)&= \min\lbp \RenyiD{P_0}{P_1}{\frac{\alpha}{1+\alpha}},\ \nu(P_0,P_1)\rbp.
    \end{align}
\end{proposition}
Following the proof in \cite{HsuWang_22}, it can be shown that the error exponents of efficient tests are upper bounded by the R\'enyi divergence, which imposes upper bounds on $\lambda(P_0,P_1)$. 
As a result, $g_1(Q_1,Q_0,Q_1)\geq 0$ for all $(Q_0,Q_1)\in\distset$, and hence $\kappa(P_0,P_1)=\infty$.
Detailed proof is provided in Appendix~\ref{app:kappa=infty}.

Next, we proceed with a sufficient condition for efficient tests.

\begin{proposition}[Sufficient Condition for Efficient Tests \cite{LevitanMerhav_02}]\label{prop:efficient}
    Take $\xi\in(0,1)$ and let $\lambda(P_0,P_1)=\xi\RenyiD{P_1}{P_0}{\frac{\beta}{1+\beta}}$, then there exists efficient \fixed\ tests, i.e. $\efix(P_0,P_1)>0$ for all $(P_0,P_1)\in\distset$.
\end{proposition}
This proposition occurs in \cite{LevitanMerhav_02}, and the proof is omitted here.
Since the \fixed\ tests are a special case of any types of sequential tests, it is obvious that there also exist efficient sequential tests.

\begin{figure*}[ht]
    \centering
    \includegraphics[scale=0.57]{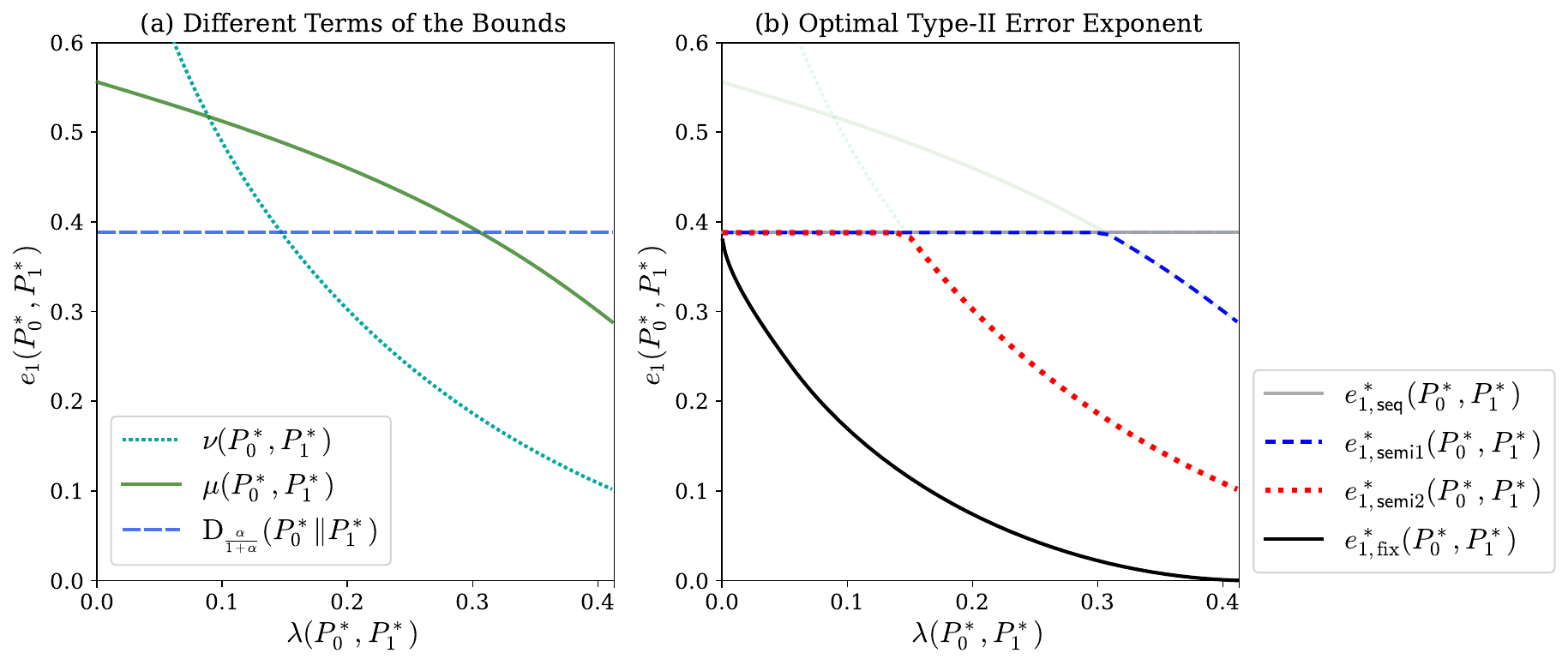}
    \caption{The optimal type-II error exponents when $\textstyle\lambda(P_0,P_1)=\xi\RenyiD{P_1}{P_0}{\frac{\beta}{1+\beta}}$. Fix $\mcal X = \{0,1\}$, $\varepsilon = 0.01$, $\alpha=\beta=0.7$, and choose $P_0^*=[0.6, 0.4]$, $P_1^*=[0.1, 0.9]$.
    Let $\xi$ increase from $0.001$ to $0.999$ to obtain a curve.
    } 
    \label{fig:compare_xi}
\end{figure*}

\Cref{fig:compare_xi} demonstrates a special class of $\lambda$ that permits efficient tests. 
As shown in \cite{LevitanMerhav_02}, there is a trade-off between the type-I and type-II error exponents of \fixed\ tests. Note that the trade-off curve depends on the function $\lambda(\cdot,\cdot)$. In particular, given two different constraint functions $\lambda(\cdot,\cdot)$ and $\tilde\lambda(\cdot,\cdot)$, even if $\lambda(P_0^*,P_1^*) = \tilde\lambda(P_0^*,P_1^*)$, the corresponding optimal $e_1(P_0^*,P_1^*)$ may be different.
On the other hand, \Cref{prop:kappa=infty} implies that \fullseq\ tests completely eradicate the trade-off, consistent with the result in \cite{HsuWang_22}.
Given the parameters in \Cref{fig:compare_xi}, there is a trade-off for \semione\ and \semitwo\ tests.

Finally, we %slightly generalize the results in \cite{LiWang_24} and 
provide a necessary and sufficient condition on $\alpha$ and $\beta$ such that \semione\ tests can achieve the same error exponents as \fullseq\ tests.%and \fullseq\ tests can both eradicate the trade-off between error exponents.
\begin{proposition}\label{prop:efficient_compare}
    Let $\lambda(P_0,P_1)=\RenyiD{P_1}{P_0}{\frac{\beta}{1+\beta}}$ for all $(P_0,P_1)\in\distset$, then \eqref{eq:renyi<mu} holds if and only if $\alpha\beta\geq1$.
    \begin{equation}
        \RenyiD{P_0}{P_1}{\frac{\alpha}{1+\alpha}} \leq \mu(P_0,P_1)\quad\forall\,(P_0,P_1)\in\distset\label{eq:renyi<mu}
    \end{equation}
\end{proposition}
The proof can be found in Appendix~\ref{app:efficient_compare}. 
\Cref{prop:efficient_compare} implies that if $\alpha\beta\geq1$, then \semione\ tests and \fullseq\ tests both achieve the optimal error exponents of efficient tests $\lp\RenyiD{P_1}{P_0}{\frac{\beta}{1+\beta}}, \RenyiD{P_0}{P_1}{\frac{\alpha}{1+\alpha}}\rp$ for all $(P_0,P_1)\in\distset$. On the other hand, 
if $\alpha\beta<1$, then for any \semione\ tests, there exists some distributions $(P_0,P_1)\in\distset$ such that the point $\lp\RenyiD{P_1}{P_0}{\frac{\beta}{1+\beta}}, \RenyiD{P_0}{P_1}{\frac{\alpha}{1+\alpha}}\rp$ is not achievable, and there remains a trade-off between the two error exponents. As for \semitwo\ tests, we do not obtain a similar result as \Cref{prop:efficient_compare}. We discuss the challenges and present some comparisons with \fullseq\ tests in \Cref{subsec:efficient_semitwo}.

\section{Proof of \Cref{prop:composite_converse}}\label{sec:converse}
For $\bm P_1\in\mscr P_1$, we want to show \eqref{eq:Expt_CHT_converse}, that is
\begin{align}
    &e_1(\bm P_1) \leq \inf_{\bm Q\in\mscr P_0\cup \Gamma_0\cup\Omega(\bm P_1)} \sum_{i=1}^s \alpha_i \KLD{Q_i}{P_{1,i}}\\
    &= \min\Big\{\inf_{\bm Q\in\mscr P_0} \sum_{i=1}^s \alpha_i \KLD{Q_i}{P_{1,i}}, \inf_{\bm Q\in\Gamma_0} \sum_{i=1}^s \alpha_i \KLD{Q_i}{P_{1,i}}, \nonumber\\
    &\hspace{3.5cm}\inf_{\bm Q\in\Omega(\bm P_1)} \sum_{i=1}^s \alpha_i \KLD{Q_i}{P_{1,i}}\Big\},
\end{align}
Notice that the bound is decomposed into three parts, corresponding to the sets $\mscr P_0$, $\Gamma_0$ and $\Omega(\bm P_1)$. In the following, we prove them in order.

\subsubsection{$\eT{1}(\bm P_1) \leq \inf_{\bm Q\in\mscr P_0} \sum_{i=1}^s \alpha_i \KLD{Q_i}{P_{1,i}}$}
Let us introduce two lemmas.
\begin{lemma}\label{lemma: converse}
Let $\bm P_0\in\mscr P_0$, $\bm P_1\in\mscr P_1$, and $\Phi=(\tau,\delta)$ be a test. If $\mbb{E}_0\lb\tau|\bm P_0\rb<\infty$, then for any $\mcal{E}\in\mcal F_\tau$, 
\begin{align}
    &\quad\ \BKLD{\mbb P_0\{\mcal{E}\}}{\mbb{P}_{1}\{\mcal{E}\}}\\
    &\leq \sum_{i=1}^\ell N_i \KLD{P_{0,i}}{P_{1,i}} + \sum_{i=\ell+1}^s \mbb{E}_0\lb\lce\alpha_i\tau\rce|\bm P_0\rb \KLD{P_{0,i}}{P_{1,i}},
\end{align}
where $\BKLD{p}{q} = p\log\frac{p}{q} + (1-p)\log\frac{1-p}{1-q}$ is the binary KL divergence. 
\end{lemma}

Lemma~\ref{lemma: converse} is proved by the data processing inequality of KL divergences and Wald's identity.
Details can be found in Appendix~\ref{app:converse_lemma}.

\begin{lemma}\label{lemma: continuity}
Define the function $h:(0,1)\times(0, 1)\to\mbb{R}$ as $h(p,q)=\frac{\BKLD{1-q}{p}}{-\log p}$. Then, $\lim_{(p,q)\to(0,0)}h(p,q)=1$.
\end{lemma}
The proof is easy and stated in Appendix~\ref{app:continuity} for completeness.

Given $\Phi_n = (\tau_n,\delta_n)$, apply Lemma~\ref{lemma: converse} with $\mcal{E} = \{\hat\theta=0\}$. For any $\bm Q\in\mscr P_0$, by the assumption in Theorem~\ref{thm:complete_semi-seq}, $\mbb{E}_0\lb\tau_n|\bm Q\rb\leq n$ and thus
\begin{align}
    &\quad\ \BKLD{1-\pi_{0}(\Phi_n|\bm Q)}{\pi_{1}(\Phi_n|\bm P_1)}\\
    &\leq \sum_{i=1}^\ell N_i \KLD{Q_i}{P_{1,i}} + \sum_{i=\ell+1}^s \mbb{E}_0\lb\lce\alpha_i\tau\rce|\bm Q\rb \KLD{Q_i}{P_{1,i}}\\
    &\leq n\sum_{i=1}^s \alpha_i \KLD{Q_i}{P_{1,i}} + \sum_{i=\ell+1}^s \KLD{Q_i}{P_{1,i}}.\label{eq:64}
\end{align}
The last inequality is due to the fact that $\lce\alpha_i\tau\rce\leq \alpha_i\tau+1$. Notice that since $\bm P_1\in\mscr P_1\subseteq \interior{\mscr D}$, the KL divergence $\KLD{Q_i}{P_{1,i}}$ is bounded by some constant.
For the asymptotic behavior, observe that if $\pi_{1}(\Phi_n|\bm P_1)$ does not vanish as $n$ goes to infinity, we simply have $\eT{1}(\bm P_1)=0$. So we may assume that $\lim_{n\to\infty}\pi_{1}(\Phi_n|\bm P_1) =0$. Also by the constraint on the type-I error exponent, we know that $\pi_{0}(\Phi_n|\bm P_0)$ goes to $0$ exponentially. Using Lemma~\ref{lemma: continuity}, we have
\begin{equation}
    \lim_{n\to\infty} \frac{\BKLD{1-\pi_{0}(\Phi_n|\bm Q)}{\pi_{1}(\Phi_n|\bm P_1)}}{-\log \pi_1(\Phi_n|\bm P_1)} = 1.\label{eq:65}
\end{equation}
By \eqref{eq:64}, 
\begin{equation}
    \liminf_{n\to\infty} \frac{\BKLD{1-\pi_{0}(\Phi_n|\bm Q)}{\pi_{1}(\Phi_n|\bm P_1)}}{n} \leq \sum_{i=1}^s \alpha_i \KLD{Q_i}{P_{1,i}}.\label{eq:66}
\end{equation}
By definition,
\begin{align}
    \eT{1}(\bm P_1) &= \liminf_{n\to\infty}\frac{-\log \pi_1(\Phi_n|\bm P_1)}{n}\\
    &= \liminf_{n\to\infty} \frac{\BKLD{1-\pi_{0}(\Phi_n|\bm Q)}{\pi_{1}(\Phi_n|\bm P_1)}}{n} \nonumber\\    &\hspace{0.5cm}\times\lim_{n\to\infty}\frac{-\log \pi_1(\Phi_n|\bm P_1)}{\BKLD{1-\pi_{0}(\Phi_n|\bm Q)}{\pi_{1}(\Phi_n|\bm P_1)}}\\
    &\leq \sum_{i=1}^s \alpha_i \KLD{Q_i}{P_{1,i}},
\end{align}
where the last inequality follows from \eqref{eq:66} and \eqref{eq:65}.
Since this is true for all $\bm Q\in\mscr P_0$, the upper bound is obtained.

\subsubsection{$\eT{1}(\bm P_1) \leq \inf_{\bm Q\in\Gamma_0} \sum_{i=1}^s \alpha_i \KLD{Q_i}{P_{1,i}}$}
Denote the empirical distribution as $\hat{\bm P} = (\hat{P_1},\hat{P_2},\dots,\hat{P_s})$ where $\hat{P_i}$ corresponds to the $i$-th sequence for $i=1,2,\dots,s$.
Observe that when $\theta=1$ and the underlying distribution is $\bm Q\in\mscr P_1$, $\hat{\bm P}$ will be close to $\bm Q$ with high probability. So if the empirical distribution is close to $\bm Q$, the test must not stop too late in order to ensure $\mbb{E}_1\lb\tau_n|\bm Q\rb\leq n$. However, stopping around $n$, if $g_1(\bm Q)<0$, the test must output $0$ in order to satisfy the universality constraint on the type-I error exponent according to the fixed-length result \cite{LevitanMerhav_02}.
This leads to upper bounds on the type-II error exponent.

We proceed by showing the technical details. If $\Gamma_0\neq\emptyset$, take $\bm Q\in\Gamma_0= \lbp\bm Q\in\mscr P_1\mv g_1(\bm Q) < 0\rbp$. By definition, there exist some $\bm P_0'\in\mscr P_0$ such that $\sum_{i=1}^s \alpha_i \KLD{Q_i}{P_{0,i}'} < \lambda(\bm P_0')$. Given $\epsilon>0$, define 
\begin{align}
    \mcal B_\epsilon(\bm Q) &= \Big\{ \bm R\in\mscr D \,\Big\vert\, \sum_{i=1}^s \alpha_i \KLD{R_i}{Q_i} < \epsilon \Big\},\\
    \mcal B_\lambda^\epsilon(\bm P_0') &= \Big\{ \bm R\in\mscr D \,\Big\vert\, \sum_{i=1}^s \alpha_i \KLD{R_i}{P_{0,i}'} < (1-\epsilon)\lambda(\bm P_0')\Big\}.
\end{align}
By the choice of $\bm P_0'$, we can take $\epsilon$ small enough such that $\mcal B_\epsilon(\bm Q)\subseteq \mcal B_\lambda^\epsilon(\bm P_0')$. For each test $\Phi_n = (\tau_n, \delta_n)$, by assumption $\mbb{E}_1\lb\tau_n|\bm Q\rb\leq n$. Thus by Markov's inequality,
\begin{equation}
    \mbb{Q}\{\tau_n> n\} = \mbb{Q}\{\tau_n\geq n+1\}\leq \frac{\mbb{E}_1\lb\tau_n|\bm Q\rb}{n+1}\leq \frac{n}{n+1}.\label{eq:markov}
\end{equation}
Here $\mbb Q$ is the shorthand notation for the joint probability law of the sequences when the underlying distribution is $\bm Q$.

Now look at the empirical distribution $\hat{\bm P}$ at time $n$. 
Note that even if the test stops earlier, we can still consider the samples until $n$. Let $\bm{\mcal{P}}^n$ denote the collection of all possible empirical distributions corresponding to the sequences observed at time $n$. It is not hard to show that
\begin{equation}
    |\bm{\mcal{P}}^n|\leq \prod_{i=1}^s\big(\lce\alpha_i n\rce + 1\big)^{|\mcal X_i|},
\end{equation}
which is polynomial in $n$.
By the method of types,
\begin{align}
    \mbb{Q}\lbp \hat{\bm P}\notin\mcal B_\epsilon(\bm Q) \rbp
    &\leq \sum_{\hat{\bm P}\in \bm{\mcal{P}}^n\cap \comp{\mcal B_\epsilon(\bm Q)}} 2^{-\sum_{i=1}^s \lce\alpha_i n\rce \KLD{\hat P_{i}}{Q_{i}}}  \\
    &\leq \sum_{\hat{\bm P}\in \bm{\mcal{P}}^n\cap \comp{\mcal B_\epsilon(\bm Q)}} 2^{-n\sum_{i=1}^s \alpha_i \KLD{\hat P_{i}}{Q_{i}}}  \\
    &\leq |\bm{\mcal{P}}^n|2^{-n\epsilon}\\
    &\leq \frac{1}{2(n+1)} \text{ for all $n$ large enough.}\label{eq:too_far}
\end{align}
Combine \eqref{eq:markov} and \eqref{eq:too_far}, then for all $n$ large enough,
\begin{align}
    &\mbb{Q}\lbp \hat{\bm P}\in\mcal B_\epsilon(\bm Q) \text{ and } \tau_n\leq n\rbp \\
    &\geq 1 - \frac{n}{n+1} - \frac{1}{2(n+1)} = \frac{1}{2(n+1)}.\label{eq:78}
\end{align}
Notice that $\mbb{Q}\lbp \hat{\bm P}\in\mcal B_\epsilon(\bm Q) \text{ and } \tau_n\leq n\rbp =$
\begin{equation}
     \sum_{\bm P^n\in \bm{\mcal{P}}^n\cap \mcal B_\epsilon(\bm Q)} \mbb{Q}\lbp \hat{\bm P} = \bm P^n\rbp \mbb{Q}\lbp \tau_n\leq n\mv \hat{\bm P}=\bm P^n\rbp. 
\end{equation}
Hence for all $n$ large enough, there exist some empirical distribution $\bm P^n=(P_1^n,P_2^n,\dots,P_s^n)\in \bm{\mcal{P}}^n\cap \mcal B_\epsilon(\bm Q)$ such that
\begin{equation}
    \mbb{Q}\lbp \tau_n\leq n\mv \hat{\bm P} = \bm P^n\rbp \geq \frac{1}{2(n+1)},\label{eq:geq_1/poly}
\end{equation}
otherwise \eqref{eq:78} will be violated.
Since all the sequences in a given type class have the same probability, a key observation is that this conditional probability is independent of the underlying distributions and ground truth $\theta$. Now we consider the situation where the underlying distributions are $\bm P_0'$ and $\theta=0$. By the method of types and the fact that $\bm P^n\in \mcal B_\epsilon(\bm Q)\subseteq \mcal B_\lambda^\epsilon(\bm P_0')$, for $n$ large enough,
\begin{align}
    \mbb{P}_0'\lbp \hat{\bm P} = \bm P^n\rbp
    &\geq \frac{1}{|\bm{\mcal{P}}^n|}2^{-\sum_{i=1}^s \lce\alpha_i n\rce \KLD{P_{i}^n}{P_{0,i}'}} \\
    &\geq \frac{c}{|\bm{\mcal{P}}^n|}2^{-n(1-\epsilon)\lambda(\bm P_0')}, \label{eq:type_prob}
\end{align}
where $c>0$ is a constant emerging from using $\lce\alpha_i n\rce\leq \alpha_i n+1$ and the boundedness of KL divergence for fixed $\bm P_0'$.
Using the key observation and the inequalities \eqref{eq:geq_1/poly}, \eqref{eq:type_prob}, we know that for $n$ large enough,
\begin{align}
    &\quad\ \mbb{P}_0'\lbp \hat{\bm P} = \bm P^n \text{ and } \tau_n\leq n\rbp  \\
    &= \mbb{P}_0'\lbp \hat{\bm P} = \bm P^n\rbp \mbb{P}_0'\lbp \tau_n\leq n\mv \hat{\bm P} = \bm P^n\rbp\\
    &=\mbb{P}_0'\lbp \hat{\bm P} = \bm P^n\rbp \mbb{Q}\lbp \tau_n\leq n\mv \hat{\bm P} = \bm P^n\rbp\\
    &\geq \frac{c}{2(n+1)|\bm{\mcal{P}}^n|}2^{-n(1-\epsilon)\lambda(\bm P_0')}.
\end{align}
Consider the type-I error probability
\begin{align}
    \pi_0(\Phi_n|\bm P_0')
    &\geq \mbb{P}_0'\lbp \hat\theta=1,\ \hat{\bm P} = \bm P^n \text{ and } \tau_n\leq n\rbp\\
    &\geq \frac{c}{2(n+1)|\bm{\mcal{P}}^n|}2^{-n(1-\epsilon)\lambda(\bm P_0')}\nonumber\\
    &\hspace{0.5cm}\times \mbb{P}_0'\lbp \hat\theta=1\mv\hat{\bm P} = \bm P^n \text{ and } \tau_n\leq n\rbp .
\end{align}
Recall that $|\bm{\mcal{P}}^n|$ is polynomial in $n$. Since $\{\Phi_n\}$ satisfies the constraint on the type-I error exponent, that is $\pi_0(\Phi_n|\bm P_0')\dotleq 2^{-n\lambda(\bm P_0')}$, we must have
\begin{equation}
    \mbb{P}_0'\lbp \hat\theta=1\mv\hat{\bm P} = \bm P^n \text{ and } \tau_n\leq n\rbp = o(1).\label{eq:1-o(1)}
\end{equation}
Again, since all the sequences in a given type class have the same probability, the above conditional probability is independent of the underlying distributions and ground truth $\theta$. When the underlying distributions are $\bm P_1$ and $\theta=1$, the error probability can be lower bounded using \eqref{eq:geq_1/poly}, \eqref{eq:1-o(1)} and the method of types. Specifically, for $n$ large enough,
\begin{align}
    &\quad\ \pi_1(\Phi_n|\bm P_1)\\
    &\geq \mbb{P}_1\lbp \hat\theta=0,\ \hat{\bm P} = \bm P^n \text{ and } \tau_n\leq n\rbp\\
    &= \mbb{P}_1\lbp \hat{\bm P} = \bm P^n\rbp \mbb{P}_1\lbp \tau_n\leq n\mv \hat{\bm P} = \bm P^n\rbp \nonumber\\
    &\hspace{2.2cm}\times\mbb{P}_1\lbp \hat\theta=1\mv\hat{\bm P} = \bm P^n \text{ and } \tau_n\leq n\rbp\\
    &= \mbb{P}_1\lbp \hat{\bm P} = \bm P^n\rbp \mbb{Q}\lbp \tau_n\leq n\mv \hat{\bm P} = \bm P^n\rbp \nonumber\\
    &\hspace{2.2cm}\times\mbb{P}_0'\lbp \hat\theta=1\mv\hat{\bm P} = \bm P^n \text{ and } \tau_n\leq n\rbp\\
    &= \frac{c'}{|\bm{\mcal{P}}^n|}2^{-n\sum_{i=1}^s \alpha_i \KLD{P_{i}^n}{P_{1,i}}}\times\frac{1}{2(n+1)}\times(1-o(1))\\
    &\geq \frac{1}{\poly{n}} 2^{-n\sum_{i=1}^s \alpha_i \KLD{P_{i}^n}{P_{1,i}}}.
\end{align}
Hence
\begin{equation}
    e_1(\bm P_1)\leq \liminf_{n\to\infty} \sum_{i=1}^s \alpha_i \KLD{P_{i}^n}{P_{1,i}}.
\end{equation}

Let $\bm P^\epsilon = (P_1^\epsilon,P_2^\epsilon,\dots,P_s^\epsilon)$ be a limit point of $\{\bm P^n\}$, then by the continuity of KL divergence,
\begin{equation}
    e_1(\bm P_1)\leq \sum_{i=1}^s \alpha_i \KLD{P_{i}^\epsilon}{P_{1,i}}.
\end{equation}
Also by the continuity of KL divergence, $\sum_{i=1}^s \alpha_i \KLD{P_{i}^\epsilon}{Q_{i}} \leq \epsilon$. As $\epsilon$ goes to $0$, $\bm P^\epsilon$ converge to $\bm Q$, and
\begin{align}
    e_1(\bm P_1)&\leq \liminf_{\epsilon\to0} \sum_{i=1}^s \alpha_i \KLD{P_{i}^\epsilon}{P_{1,i}}= \sum_{i=1}^s \alpha_i \KLD{Q_{i}}{P_{1,i}}.
\end{align}
Since this hold for all $\bm Q\in\Gamma_0$, we get the desired result.

\subsubsection{$\eT{1}(\bm P_1) \leq \inf_{\bm Q\in\Omega(\bm P_1)} \sum_{i=1}^s \alpha_i \KLD{Q_i}{P_{1,i}}$}
Observe that the set $\Omega(\bm P_1)$ corresponds to the limitation imposed by the fixed-length training sequences. 
The bound is obtained by a reduction from a fixed-length composite hypothesis testing problem. Intuitively, if we are allowed to drop the constraint on the expected stopping time and take infinitely many samples of the $(\ell+1)$-th to $s$-th sequences, then their underlying distributions can be fully known. Specifically, when $\theta=1$ and the underlying distribution is $\bm P_1$, this means $(P_{1,\ell+1},\dots,P_{1,s})$ are known.
In this situation, consider the following equivalent problem. Suppose $(P_{1,\ell+1},\dots,P_{1,s})$ is \emph{fixed} and \emph{known}. 
The decision maker observes $\ell$ independent fixed-length sequences $X_1^{N_1},\dots,X_\ell^{N_\ell}$, and the objective is to decide between the following two hypotheses:
\begin{align*}
    &\mcal H_0:X_{i,k}\diid P_{0,i}\ \forall\, 1\leq i\leq \ell\text{ for some }\bm P_0\in\mscr P_0\cap \mcal{S}(\bm P_1)\\ 
    &\mcal H_1:X_{i,k}\diid P_{1,i}'\ \forall\, 1\leq i\leq \ell\text{ for some }\bm P_1'\in\mscr P_1\cap \mcal{S}(\bm P_1) 
\end{align*}

Recall that $ \mcal{S}(\bm P_1) = \lbp \bm Q\in\mscr D \mv Q_i=P_{1,i}\ \forall\,i=\ell+1,\dots,s\rbp$ represents the collection of distributions with the $(\ell+1)$-th to $s$-th elements equal to $P_{1,\ell+1},\dots,P_{1,s}$.
Observe that this is a fixed-length composite hypothesis testing problem\footnote{The case $\mscr P_0\cap \mcal{S}(\bm P_1)=\emptyset$ is trivial, so here we focus on the case when it is nonempty.}.
Given a sequence of tests $\{\Phi_n\}$ for the original problem, we can apply it to this new problem. Specifically, generate the $(\ell+1)$-th to $s$-th sequences using the knowledge of $(P_{1,\ell+1},\dots,P_{1,s})$. Along with the observed fixed-length sequences $X_1^{N_1},\dots,X_\ell^{N_\ell}$, the test $\Phi_n$ will output a decision. The above method gives a randomized test, yet theoretically it can be easily derandomized without affecting the exponential rate. One can observe a clear correspondence between the error probabilities of the two problems, and hence a correspondence between the exponential rates of error probabilities.
Since $\{\Phi_n\}$ satisfies the universality constraints on the type-I error exponent, we can apply the result for fixed-length composite hypothesis testing \eqref{eq:optimal_fixed-length_e1} and get the desired bound.

\section{Proof of \Cref{prop:composite_achievability}}\label{sec:achievability}
In this section we prove \Cref{prop:composite_achievability} by proposing a sequence of tests satisfying the universality constraints and achieve the desired error exponents.

\subsection{Proposed Test}
When $\ell<s$, the decision maker has the flexibility to decide when to stop, so intuitively we would like to take more samples before making the decision.
Nevertheless, the expected stopping time should not exceed $n$, meaning the probability of taking more samples should be kept small. 
First consider the empirical distributions at time $n-1$, written as $\hat {\bm P}^{n-1}$ for simplicity. 
Observe that with high probability, the empirical distributions are close to the true underlying distributions.
Define two subsets of $\mscr D$, for $\theta=0,1$,
\begin{align}
    &\Lambda_\theta^n = \bigcup_{\bm P_\theta'\in\mscr P_\theta} \Big\{ \bm Q\in\mscr D\,\Big\vert\,\sum_{i=1}^s \alpha_i \KLD{Q_{i}}{P_{\theta,i}'} <  \eta_n\Big\}, \label{eq:def_Lambda}
\end{align}
where $\eta_n = \big[2\log n + \sum_{i=1}^s |\mcal X_i|\log(\lce\alpha_i n\rce+1)\big]/(n-1)$ is a margin vanishing in $n$. By the method of types, the empirical distributions lie in these sets with high probability. Hence the universality constraint on the expected stopping time can be satisfied with the following stopping time
\begin{equation}\label{eq:stopping_time}
    \tau_n = \begin{cases}
        n-1&\quad\text{if }\hat {\bm P}^{n-1}\in\Lambda_0^n \cup \Lambda_1^n,\\
        n^2&\quad\text{if }\hat {\bm P}^{n-1}\in\lp\Lambda_0^n \cup \Lambda_1^n\rp^{\mathrm{c}}.
    \end{cases}
\end{equation}
Here, $\hat {\bm P}^{n-1}\in\lp\Lambda_0^n \cup \Lambda_1^n\rp^{\mathrm{c}}$ implies that the observed samples are not a good representation of the real underlying distribution, and thus more samples are needed.
Stopping at time $n^2$ can be viewed as approaching ``having infinitely many samples''.

To specify the decision rule, in addition to the set $\Gamma_0= \lbp\bm Q\in\mscr P_1\mv g_1(\bm Q) < 0\rbp$, let $\Gamma_1= \lbp\bm Q\in\mscr P_1\mv g_1(\bm Q) \geq 0\rbp$.
Note that $\Gamma_0$ consists of the distributions under ground truth $1$ that are ``$\lambda$-close'' to some possible distributions $\bm P_0\in\mscr P_0$ under ground truth $0$. On the other hand, $\Gamma_1$ consists of those that are \emph{not} too close to $\mscr P_0$.
For $\theta=0,1$, let
\begin{align}
    \mcal B_n(\Gamma_\theta) = &\bigcup_{\bm Q\in\Gamma_\theta} \Big\{ \bm R\in\mscr D\,\Big\vert\,
    \sum_{i=1}^s \alpha_i \KLD{R_{i}}{Q_{i}} < \eta_n \Big\}
\end{align}
be a subset of $\mcal P(\mcal X)^3$ that slightly extends $\Gamma_\theta$.
When $\tau_n=n-1$, the decision rule maps the observed samples to $\hat\theta$, where
\begin{equation}\label{eq:decision_rule_n-1}
    \hat\theta=\begin{cases}
        0\quad\text{if } \hat {\bm P}^{n-1}\in \Lambda_0^n \cup \mcal B_n(\Gamma_0),\\%\Omega\cap\Lambda_0^n,\\
        1\quad\text{if } \hat {\bm P}^{n-1}\in \mcal B_n(\Gamma_1).
    \end{cases}
\end{equation}
Notice that $\mcal B_n(\Gamma_0)\cup\mcal B_n(\Gamma_1)=\Lambda_1^n$. 

When $\tau_n=n^2$, the decision maker observes $(X_1^{N_1}, \dots,X_\ell^{N_\ell},X_{\ell+1}^{\lce\alpha_{\ell+1}n^2\rce},\dots,X_s^{\lce\alpha_{s}n^2\rce})$.
Denote the corresponding empirical distribution as $\hat {\bm P}^{n^2}$, and use the fixed-length test in \cite{LevitanMerhav_02}. 
As a result, we can ensure that the type-I error probability is of the same order as when $\tau_n=n-1$. Specifically, the decision rules is 
\begin{equation}
    \hat\theta = \begin{cases}
        0\quad\text{if } g_n(\hat {\bm P}^{n^2}) < 0,\\
        1\quad\text{otherwise},
    \end{cases}
\end{equation}
where $g_n(\bm Q) =$
\begin{equation}
     \inf_{\bm P_0'\in\mscr P_0} \Big( \sum_{i=1}^\ell \alpha_i \KLD{Q_i}{P_{0,i}'} + n\hspace{-0.2cm}\sum_{i=\ell+1}^s \alpha_i \KLD{Q_i}{P_{0,i}'} - \lambda(\bm P_0') \Big).
\end{equation}

\subsection{Performance Analysis}

Now the test is clearly defined, we are ready to analyze its performance. 

\subsubsection{Expected Stopping Time}
First we show that the proposed test satisfies the universality constraint on the expected stopping time. Given $\theta\in\{0,1\}$ and $\bm P_\theta\in\mscr P_\theta$, we have
\begin{align}
    &\quad\ \mbb E_\theta[\tau_n|\bm P_\theta]\\ 
    &\leq n-1 + \mbb P_\theta\lbp\hat{\bm P}^{n-1}\in(\Lambda_\theta^n)^{\mathrm c}\rbp\times n^2\\
    &\leq n-1 +\hspace{-0.2cm}\sum_{\hat{\bm P}\in \bm{\mcal{P}}^{n-1}\cap (\Lambda_\theta^n)^{\mathrm c}} \hspace{-0.2cm}\Big(2^{-\sum_{i=1}^\ell \lce\alpha_i n\rce \KLD{\hat P_{i}}{P_{\theta,i}}} \nonumber\\
    &\hspace{3cm}\times2^{-\sum_{i=\ell+1}^s \lce\alpha_i (n-1)\rce \KLD{\hat P_{i}}{P_{\theta,i}}}\Big) \times n^2\label{eq:106}\\
    &\leq n-1 + \prod_{i=1}^s (\lce\alpha_i n\rce+1)^{|\mcal X_i|}\times 2^{-(n-1)\eta_n}\times n^2 =n.
\end{align}
Note that \eqref{eq:106} follows from the method of types. Recall that $\bm{\mcal{P}}^{n-1}$ is the collection of all possible empirical distributions corresponding to the sequences observed at time $n-1$, and $|\bm{\mcal{P}}^{n-1}|\leq \prod_{i=1}^s (\lce\alpha_i n\rce+1)^{|\mcal X_i|}$.

\subsubsection{Error exponents}
To calculate the error exponents, we introduce a useful lemma and define some sets.
\begin{lemma}\label{lemma:inf_continuity}
    Let $S\subseteq \mbb R^m$ be a compact set, $f:S\to \mbb R$ be a continuous function and $A\subseteq S$. Given $\epsilon>0$ and $x\in\mbb R^m$, let $\mcal B'_\epsilon(x) = \lbp y\in\mbb R^m : \lnorm y-x\rnorm_1<\epsilon\rbp$ be the $\epsilon$-ball centered at $x$. Let $\mcal B'_\epsilon(A) = \cup_{x\in A}\mcal B'_\epsilon(x)$. Then
    \begin{enumerate}
        \item $\inf_{x\in A} f(x) = \inf_{x\in\closure{A}} f(x)$,
        \item $\lim_{\epsilon\to0}\inf_{x\in \mcal B'_\epsilon(A)\cap S} f(x) = \inf_{x\in\closure{A}} f(x)$
    \end{enumerate}
\end{lemma}
The proof is easy and the details can be found in Appendix~\ref{app:inf_continuity}.
Since $\mscr D$ can be embedded into $\mbb R^{\sum_{i=1}^s (|\mcal X_i|-1)}$, consider for $\theta=0,1$, % and some constant $c>0$,
\begin{equation}
    \mcal B'_n(\Gamma_\theta) = \bigcup_{\bm Q\in\Gamma_\theta} \Big\{ \bm R\in\mscr D \,\Big\vert\, \lnorm \bm R-\bm Q \rnorm_1 < \sqrt{\eta_n}\sum_{i=1}^s \sqrt{2/\alpha_i} \Big\},
\end{equation}
and
\begin{equation}
    \tilde\Lambda_0^n = \bigcup_{\bm P_0'\in\mscr P_0} \Big\{ \bm Q\in\mscr D\,\Big\vert\,\lnorm \bm Q-\bm P_0' \rnorm_1 < \sqrt{\eta_n}\sum_{i=1}^s \sqrt{2/\alpha_i}\Big\}.
\end{equation}
By Pinsker's inequality, we know that $\mcal B_n(\Gamma_\theta)\subseteq\mcal B'_n(\Gamma_\theta)$ and $\Lambda_0^n\subseteq\tilde\Lambda_0^n$. 
Specifically, for $\bm Q\in\Gamma_\theta$ and $\bm R\in\mscr D$, if $\sum_{i=1}^s \alpha_i \KLD{R_{i}}{Q_{i}} < \eta_n$, then
\begin{align}
    \lnorm \bm R-\bm Q \rnorm_1 &= \sum_{i=1}^s \lnorm  R_i- Q_i \rnorm_1 \\
    &\leq \sum_{i=1}^s \sqrt{2\KLD{R_i}{Q_i}}<  \sum_{i=1}^s \sqrt{2\eta_n/\alpha_i}.
\end{align}
Similar argument works for $\lnorm \bm Q-\bm P_0' \rnorm_1$. Notice that $\sum_{i=1}^s \sqrt{2/\alpha_i}$ is constant, so $\sqrt{\eta_n}\sum_{i=1}^s \sqrt{2/\alpha_i}$ goes to $0$ as $n$ goes to infinity.

For the type-I error exponent, let $\theta=0$ and $\bm P_0\in\mscr P_0$. Based on the stopping time, the error events can be divided into two parts. When $\tau_n=n-1$, there is an error only if $\hat {\bm P}^{n-1}\in \mcal B_n(\Gamma_1)$. 
If $\Gamma_1\neq\emptyset$, using the method of types,
\begin{align}
    &\quad\ \mbb P_0\lbp\hat{\bm P}^{n-1}\in\mcal B_n(\Gamma_1)\rbp\\
    &\leq\mbb P_0\lbp\hat {\bm P}^{n-1}\in\mcal B'_n(\Gamma_1)\rbp\\
    &\leq \sum_{\hat{\bm P}\in \bm{\mcal{P}}^{n-1}\cap \mcal B'_n(\Gamma_1)} \Big(2^{-\sum_{i=1}^\ell \lce\alpha_i n\rce \KLD{\hat P_{i}}{P_{0,i}}}\nonumber\\
    &\hspace{2.8cm} \times2^{-\sum_{i=\ell+1}^s \lce\alpha_i (n-1)\rce \KLD{\hat P_{i}}{P_{0,i}}}\Big)\\
    &\leq \prod_{i=1}^s (\lce\alpha_i n\rce+1)^{|\mcal X_i|}\times2^{-(n-1)\tilde\lambda_n(\bm P_0)},\label{eq:113}
\end{align}
where
\begin{equation}
    \tilde\lambda_n(P_0,P_1) = \inf_{\bm R\in\mcal B'_n(\Gamma_1)} \sum_{i=1}^s \alpha_i \KLD{R_{i}}{P_{0,i}}.
\end{equation}
Since $\sqrt{\eta_n}\sum_{i=1}^s \sqrt{2/\alpha_i}$ vanishes as $n$ goes to infinity, using Lemma~\ref{lemma:inf_continuity} and the definition of $\Gamma_1$,
\begin{equation}
    \lim_{n\to\infty}\tilde\lambda_n(\bm P_0) = \inf_{\bm R\in\Gamma_1} \sum_{i=1}^s \alpha_i \KLD{R_{i}}{P_{0,i}} \geq \lambda(\bm P_0).\label{eq:115}
\end{equation}
Combining \eqref{eq:113} and \eqref{eq:115}, we can see that the error when stopping at $n-1$ has exponential rate at least $\lambda(\bm P_0)$.
When $\tau_n=n^2$, there is an error if $g_n(\hat {\bm P}^{n^2}) \geq 0$, by the method of types, 
\begin{align}
    &\quad\ \mbb P_0\lbp \tau_n=n^2,\ \hat\theta=1 \rbp \\
    &= \sum_{\substack{\hat{\bm P}\in \bm{\mcal{P}}^{n^2}\\g_n(\hat{\bm P})\geq0}} 2^{-\sum_{i=1}^\ell \lce\alpha_i n\rce \KLD{\hat P_{i}}{P_{0,i}}-\sum_{i=\ell+1}^s \lce\alpha_i n^2\rce \KLD{\hat P_{i}}{P_{0,i}}}\\
    &\leq \prod_{i=1}^\ell (\lce\alpha_i n\rce+1)^{|\mcal X_i|}\prod_{i=\ell+1}^s (\lce\alpha_i n^2\rce+1)^{|\mcal X_i|}\times 2^{-n\lambda(\bm P_0)}\\
    &\doteq 2^{-n\lambda(\bm P_0)}.
\end{align}
As a result, the total type-I error probability has exponential rate $\lambda(\bm P_0)$, which satisfies the universality constraint on the type-I error exponent.

For the type-II error exponent, let $\theta=1$ and $\bm P_1\in\mscr P_1$. When $\tau_n=n-1$, the error probability is upper bounded by
\begin{align}
    &\mbb P_1\lbp\hat{\bm P}^{n-1}\in\Lambda_0^n \cup \mcal B_n(\Gamma_0)\rbp
    \leq \mbb P_1\lbp\hat{\bm P}^{n-1}\in\tilde\Lambda_0^n \cup \mcal B_n'(\Gamma_0)\rbp \label{eq:error_n-1}
\end{align}
Similar as above, since $\sqrt{\eta_n}\sum_{i=1}^s \sqrt{2/\alpha_i}$ vanishes as $n$ goes to infinity, using Lemma~\ref{lemma:inf_continuity} and method of types, it can be shown that the resulting error exponent is at least $\inf_{\bm Q\in\mscr P_0\cup \Gamma_0} \sum_{i=1}^s \alpha_i \KLD{Q_i}{P_{1,i}}$.
On the other hand, when $\tau_n=n^2$, 
\begin{align}
    &\quad\ \mbb P_1\lbp\tau_n=n^2,\  \hat\theta =0 \rbp\\
    &= \sum_{\substack{\hat{\bm P}\in \bm{\mcal{P}}^{n^2}\\g_n(\hat{\bm P})<0}} 2^{-\sum_{i=1}^\ell \lce\alpha_i n\rce \KLD{\hat P_{i}}{P_{1,i}}-\sum_{i=\ell+1}^s \lce\alpha_i n^2\rce \KLD{\hat P_{i}}{P_{1,i}}}\\
    &\leq \prod_{i=1}^\ell (\lce\alpha_i n\rce+1)^{|\mcal X_i|}\prod_{i=\ell+1}^s (\lce\alpha_i n^2\rce+1)^{|\mcal X_i|}\times 2^{-n\mu_n(\bm P_1)},
\end{align}
where $\mu_n(\bm P_1) =$
\begin{equation}
     \inf_{\substack{\bm Q\in\mscr D\\g_n(\bm Q) < 0}} \Big( \sum_{i=1}^\ell \alpha_i \KLD{Q_i}{P_{1,i}} + n\sum_{i=\ell+1}^s \alpha_i \KLD{Q_i}{P_{1,i}} \Big).
\end{equation}
\begin{lemma}\label{lemma:limit}
Assume $\mscr P_0$ is bounded away from the boundary of $\mscr D$, and $\lambda:\mscr P_0 \to (0,\infty)$ can be extended to a continuous function $\bar\lambda:\closure{\mscr P_0}\to[0,\infty)$. Then for $\bm P_1\in\mscr P_1$, the sequence $\{\mu_n(\bm P_1)\}$ is non-decreasing in $n$, and
\begin{equation}
    \lim_{n\to\infty} \mu_n(\bm P_1) \geq \inf_{\bm Q\in\Omega'(\bm P_1)} \sum_{i=1}^s \alpha_i \KLD{Q_i}{P_{1,i}}.
\end{equation}
\end{lemma}
Intuitively, as $n$ grows, to minimize $\sum_{i=1}^\ell \alpha_i \KLD{Q_i}{P_{1,i}} + n\sum_{i=\ell+1}^s \alpha_i \KLD{Q_i}{P_{1,i}}$, $Q_i$ should be close to $P_{1,i}$ for $i=\ell+1,\dots,s$, otherwise $n\KLD{Q_i}{P_{1,i}}$ gets too large. Also, to have $g_n(\bm Q)$ less than $0$, $P_{0,i}'$ should be close to $Q_i$ for $i=\ell+1,\dots,s$. So the result should approach restricting $P_{0,i}'=Q_i=P_{1,i}$ for $i=\ell+1,\dots,s$.
The details are given in Appendix~\ref{app:limit}. 
    
By Lemma~\ref{lemma:limit}, we know the exponential type-II error rate when $\tau_n=n^2$. 
Combining the above results, \Cref{prop:composite_achievability} is proved.

\section{Discussion}\label{sec:discussion}
\subsection{Assumptions on $\lambda$}
\subsubsection{Continuity}
In our achievability part, \Cref{prop:composite_achievability}, it is assumed that $\lambda:\mscr P_0 \to (0,\infty)$ can be extended to a continuous function $\bar\lambda:\closure{\mscr P_0}\to[0,\infty)$, and this assumption is crucial when proving the limit of $\{\mu_n(\bm P_1)\}$ in \Cref{lemma:limit}.
An interesting question is what will happen for general $\lambda$. 
In Section~\ref{sec:achievability}, it is shown that our proposed test can achieve type-II error exponent
\begin{equation}
    \min\lbp \inf_{\bm Q\in\mscr P_0\cup \Gamma_0} \sum_{i=1}^s \alpha_i \KLD{Q_i}{P_{1,i}},\ \lim_{n\to\infty}\mu_n(\bm P_1)\rbp.
\end{equation}
When $\lambda$ does not satisfy the continuity assumption, there might be a gap between $\lim_{n\to\infty}\mu_n(\bm P_1)$ and the desired quantity. In the following, we would give a lower bound on $\lim_{n\to\infty}\mu_n(\bm P_1)$. Also, though we currently cannot improve the type-II error exponent to close the gap, we could in fact show the asymptotic optimality by increasing the type-I error exponent.

For simplicity, we introduce some notations.
Define the divergence ball centered at $\bm P_0\in\mscr P_0$ with radius $r$ as
\begin{equation}
    \mcal B_r(\bm P_0) = \big\{ \bm Q\in\mscr D \,\big\vert\, \sum_{i=1}^s \alpha_i \KLD{Q_i}{P_{0,i}} < r \big\}.
\end{equation}
We use $\mcal B_\lambda(\bm P_0)$ to denote $\mcal B_{\lambda(\bm P_0)}(\bm P_0)$ and let $\Omega_0=\bigcup_{\bm P_0'\in\mscr P_0}\mcal B_\lambda(\bm P_0')$. Notice that $\Gamma_0= \lbp\bm Q\in\mscr P_1\mv g_1(\bm Q) < 0\rbp=\mscr P_1\cap \Omega_0$. 

\begin{proposition}\label{prop:trivial}
    If $\mscr P_0$ is bounded away from the boundary of $\mscr D$ and $\lambda$ is unbounded, then $\mscr P_1\subseteq\Gamma_0$ and hence $e_1^*(\bm P_1)=0$ for all $\bm P_1\in\mscr P_1$.
\end{proposition}
\begin{proof}
    For any $\bm Q\in\mscr D$ and $\bm P_0\in\mscr P_0$, $\sum_{i=1}^s \alpha_i \KLD{Q_i}{P_{0,i}}\leq C$ for some $C>0$.
    Since we can find $\bm P_0\in\mscr P_0$ such that $\lambda(\bm P_0)$ is arbitrarily large, we have $\mscr D\subseteq\mcal B_\lambda(\bm P_0)$ and thus $\mscr P_1\subseteq\Gamma_0$.
\end{proof}

By \Cref{prop:trivial}, unbounded constraint functions are trivial. On the other hand, given any bounded $\lambda$, we can find a substituting $\tilde\lambda$ that can be extended to a continuous function $\bar\lambda:\closure{\mscr P_0}\to[0,\infty)$ and does not affect $\Omega_0$. 
\begin{proposition}\label{prop:substituting_lambda}
    Assume $\mscr P_0$ is bounded away from the boundary of $\mscr D$, then given any bounded $\lambda:\mscr P_0\to (0,\infty)$, there exist a function $\tilde\lambda:\mscr P_0\to (0,\infty)$ such that 
    \begin{itemize}
        \item $\tilde\lambda$ can be extended to a continuous function $\bar\lambda:\closure{\mscr P_0}\to[0,\infty)$,
        \item $\tilde\lambda(\bm P_0) \geq \lambda(\bm P_0)$ for all $\bm P_0\in\mscr P_0$,
        \item $\bigcup_{\bm P_0'\in\mscr P_0}\mcal B_{\tilde\lambda}(\bm P_0') = \Omega_0 = \bigcup_{\bm P_0'\in\mscr P_0}\mcal B_\lambda(\bm P_0')$.
    \end{itemize}
\end{proposition}
The idea is to take $\tilde\lambda(\bm P_0)=\sup\big\{ r>0 \,\big\vert\, \mcal B_r(\bm P_0) \subseteq \Omega_0\big\}$.
Detailed proof can be found in Appendix~\ref{app:substituting_lambda}. 
For the fixed-length setup ($\ell=s$), the optimal type-II error exponents \eqref{eq:optimal_fixed-length_e1} only have dependency on $\lambda$ through $\Omega_0$. 
This implies that substituting $\lambda$ with $\tilde\lambda$ does not affect the type-II error exponents.
Similarly, for the fully-sequential setup ($\ell=0$), 
notice that substituting $\lambda$ with $\tilde\lambda$ decreases the type-II error exponents, but \Cref{cor:fully-seq_CHT} shows 
the optimal type-II error exponents under $\tilde\lambda$ only have dependency on $\tilde\lambda$ through $\Gamma_0=\mscr P_1\cap\Omega_0$. As a result, the optimal type-II error exponents under $\lambda$ is sandwiched and we can obtain a slightly stronger version of \Cref{cor:fully-seq_CHT}. Specifically, we can drop the continuity assumption on $\lambda$.
\begin{corollary}[Fully-sequential composite hypothesis testing - relaxed]\label{cor:fully-seq_CHT_stronger}
    Suppose $\mscr P_0$ is bounded away from the boundary of $\mscr D$ and $\ell=0$, then the  optimal type-II error exponent of all tests that satisfy the universality constraints in \Cref{prop:composite_converse} and \Cref{prop:composite_achievability} is
    \begin{equation}
        e_1^*(\bm P_1)=\inf_{\bm Q\in\mscr P_0\cup \Gamma_0} \sum_{i=1}^s \alpha_i \KLD{Q_i}{P_{1,i}},
    \end{equation}
    and it can be achieved by a sequence of tests that do not depend on the actual $\bm{P}_1$.
\end{corollary}

In the case of binary classification, by \Cref{prop:substituting_lambda}, for \fixed\ and \fullseq\ setups, the type-I and optimal type-II error exponents under $\tilde\lambda$ is at least the optimal error exponents under $\lambda$. For \semione\ tests, let $\tilde\mu(P_0,P_1)$ denote the $\mu(P_0,P_1)$ under $\tilde\lambda$. It is clear that
\begin{equation}
    \tilde\mu(P_0,P_1)\leq\lim_{n\to\infty}\mu_n(P_0,P_1)\leq\mu(P_0,P_1).
\end{equation}
In other words, we cannot close the potential gap between $\lim_{n\to\infty}\mu_n(P_0,P_1)$ and $\mu(P_0,P_1)$. However we can instead increase the type-I error exponent without affecting the optimal type-II error exponents of \fixed\ and \fullseq\ tests, and in this case show the optimality of the corresponding type-II error exponent of \semione\ tests. The same argument works for \semitwo\ tests too.

Note that when there is a gap, it might be the converse bound that is not tight enough. In the proof of converse, it is assumed that we have perfect knowledge about the distribution of the testing sequence. If $\lambda$ is not continuous, any slight estimation error may cause significant difference. 

\subsubsection{Range}
Previously we assume that $\lambda$ is a strictly positive function. Now we discuss the case where $\lambda:\mscr D\to [0,\infty)$ can sometimes be $0$. Let $\mscr Q=\lbp\bm P_0\in\mscr P_0\mv\lambda(\bm P_0)=0\rbp$. Then simply consider the composite hypothesis testing problem with $\mscr P_0'=\mscr P_0\setminus \mscr Q$ and $\mscr P_1'=\mscr P_1\cup \mscr Q$. The idea is as follows. Since we only care about the error exponent, for any distribution $\bm P_0\in\mscr Q$, we do not have any requirement on error probability $\pi_0(\Phi_n|\bm P_0)$. The only constraint to be satisfied is that $\mbb E_0[\tau_n|\bm P_0]\leq n$. Putting these distributions into the alternative hypothesis does exactly what we need.

\subsection{Assumptions on the Underlying Distributions}

Apart from the constraint function, another possible relaxation is about the assumption on the underlying distributions. In the achievability part, it is assumed that $\mscr P_0$ is bounded away from the boundary of $\mscr D$. If we allow distributions to come from $\interior{\mscr D}$, it is possible that the sequence we construct in the proof converges to some distributions that do not have full support and that causes problems such as infinite KL divergences.

In \Cref{prop:composite_achievability}, additional assumptions are made in order to match the result with the converse bound in \Cref{prop:composite_converse}. Namely, a sufficient condition is that $\closure{\mscr P_0}\cap \mcal{S}(\bm P_1) = \closure{\mscr P_0\cap \mcal{S}(\bm P_1)}$. If there exist some $\bm Q\in\closure{\mscr P_0}\cap \mcal{S}(\bm P_1)$ that cannot be approximated by points in $\mscr P_0\cap \mcal{S}(\bm P_1)$, then our current proposed tests may result in a sequence converging to that point and lead to smaller type-II error exponents. It is unclear whether it is the converse or achievability bound that is too loose. As mentioned before, as long as the number of samples are finite, the underlying distributions cannot be perfectly known. Since the error exponent is an asymptotic notion, attempts to universally achieve the converse bound may further jeopardize the performance when the sample size is small. Note that it is also possible that the assumption of perfectly known distributions in the proof of converse is too strong. 

\subsection{Comparison with previous work}
Here we compare our test specialized to the \semione\ and \fullseq\ setups with existing tests in \cite{HaghifamTan_21,BaiZhou22,HsuWang_22}. 
For these two setups, given universality constraint on the expected stopping time $n$, the sets defined in \eqref{eq:def_Lambda} can be written as
\begin{align}
    &\Lambda_\theta^n = \Big\{ (Q,Q_0,Q_1)\,\Big\vert\nonumber\\
    &\inf_{(P_0',P_1')\in\distset} \KLD{Q}{P_\theta'} + \alpha\KLD{Q_0}{P_0'}+ \beta\KLD{Q_1}{P_1'} <  \eta_n\Big\}
\end{align}
for $\theta = 0,1$, 
where $\eta_n = \big[(d+2)\log n + d\log(\lce\alpha n\rce+1) + d\log(\lce\beta n\rce+1)\big]/(n-1)$.
Since the GJS divergence can be written as the following minimization problem $\GJS{P}{Q}{\alpha} = \min_{V\in\mcal{P}(\mcal{X})} \{\alpha\KLD{P}{V} + \KLD{Q}{V}\}$, we can further simplify the sets as
\begin{align}
    \Lambda_0 &= \big\{ (Q,Q_0,Q_1)\,\big\vert\,\GJS{Q_0}{Q}{\alpha} <  \eta_n\big\} \\
    \text{and }\Lambda_1 &= \big\{ (Q,Q_0,Q_1)\,\big\vert\,\GJS{Q_1}{Q}{\beta} <  \eta_n\big\}. \label{eq:Lambda_as_GJS}
\end{align}

The stopping time of our test is either $n-1$ or $n^2$, depending on whether the observed samples at time $n-1$ are typical with respect to any possible underlying distributions. With the expression in \eqref{eq:Lambda_as_GJS}, it can also be viewed as the test stops at time $n-1$ if the testing samples are close to at least one of the $P_0/P_1$-training samples, measured with GJS divergence.
This ensures the expected stopping time to be bounded by $n$ and allows more samples when the already observed ones cannot well represent the true underlying distributions.
As for the decision rule, the idea is to output $1$ whenever possible, without violating the universality constraint on the type-I error exponent. So we simply apply the optimal fixed-length decision rule in \cite{LevitanMerhav_02}. 

Our test is closely related to the \fullseq\ test in \cite{HsuWang_22}, with two main differences. 
In \cite{HsuWang_22}, after observing $n-1$ samples, the test stops if the testing samples are close enough to at least one of the $P_0/P_1$-training samples, measured with GJS divergence.
Hence the stopping time can be anything larger or equal than $n-1$. However, notice that when all the sequences are observed sequentially, the error is dominated by the error when stopping at time around $n$. So we can simplify the test to a two-phase one, similar to the almost fixed length test in \cite{BaiZhou22}. Another difference is that in \cite{HsuWang_22}, it is require that the test is efficient, that is, the error vanishes exponentially for any underlying distributions. This makes it incomparable to the \fixed\ setup with constant constraint on the type-I error exponent in \cite{Gutman_89}. In this work, the problem is resolved and we can show the benefit of sequentiality in a fair manner.

It is clear that our test is customized to satisfy the universality constraints. This also marks the main difference with the \semione\ tests in \cite{HaghifamTan_21,BaiZhou22}. In \cite{BaiZhou22}, they also have a two-phase test, but there are two differences regarding the stopping time. First, their stopping time can be $n$ or $kn$ with some constant $k$. In comparison, our second phase observes subsequently more samples, resulting in possibly lower error. Second, their test stops at $n$ if the testing samples are \emph{far} from at least one of the training samples, measured with GJS divergence. The problem is that when $P_0$ and $P_1$ are really close to each other, the testing samples will be close to both of the $P_0/P_1$-training samples with high probability, which makes the test stop at $kn>n$. So their test cannot satisfy the universality constraint on the expected stopping time.
In \cite{HaghifamTan_21}, the test stops when the testing samples are \emph{far enough} from at least one of the $P_0/P_1$-training samples, measured in GJS divergence. Similarly, the expected stopping time also depends on the underlying distributions.

\subsection{Efficient \semitwo\ Tests}\label{subsec:efficient_semitwo}
In Section~\ref{sec:results}, we compare the error exponent of efficient \semione\ and \fullseq\ tests. Specifically, \Cref{prop:efficient_compare} provides a necessary and sufficient condition on $\alpha$ and $\beta$ such that \semione\ tests can achieve the same error exponents as \fullseq\ tests, and this condition is universal for all the underlying distributions. In the example in Figure~\ref{fig:compare_xi}, there is a strict gap between \semitwo\ and \fullseq\ tests. Yet this is not always the case. 
We know that $\nu(P_0,P_1)$ does not depend on $\alpha,\beta$, while $\RenyiD{P_1}{P_0}{\frac{\beta}{1+\beta}}$ and $\RenyiD{P_0}{P_1}{\frac{\alpha}{1+\alpha}}$ go to $0$ as $\alpha,\beta\to0$. As a result, for any $(P_0,P_1)$, when $\alpha$ or $\beta$ is small enough, the \semitwo\ tests achieves $\lp\RenyiD{P_1}{P_0}{\frac{\beta}{1+\beta}}, \RenyiD{P_0}{P_1}{\frac{\alpha}{1+\alpha}}\rp$, the same as \fullseq\ tests. 

Next we provide a sufficient condition on $\alpha$ and $\beta$ such that there is a strict gap between \semitwo\ and \fullseq\ tests. First notice that for a specific pair of distributions $(P_0,P_1)\in\distset$, $\nu(P_0,P_1)$ only depends on $\lambda(P_0,P_1)$. Moreover, recall that the pair $\big(\lambda(P_0,P_1),\nu(P_0,P_1)\big)$ lies on the trade-off curve of error exponents for the binary hypothesis testing problem. Classical result shows that the curve can be characterized by the parametric curve
\begin{equation}
    \lbp \big( \KLD{P_\rho}{P_0}, \KLD{P_\rho}{P_1} \big)\mv \rho\in[0,1]\rbp,
\end{equation}
where $P_\rho$ is a tilted distribution of $P_0$ towards $P_1$, described by
\begin{equation}
    P_\rho(x) = \frac{P_0(x)^{1-\rho} P_1(x)^{\rho}}{\sum_{x'\in\mcal X}P_0(x')^{1-\rho} P_1(x')^{\rho}} \quad\text{for } x\in\mcal X.
\end{equation}
Also, note that by \eqref{eq:Renyi_minimizer}, the R\'enyi divergence can be expressed by
\begin{align}
    \RenyiD{P_1}{P_0}{\frac{\beta}{1+\beta}} &= \beta\KLD{P_{\frac{\beta}{1+\beta}}}{P_1} + \KLD{P_{\frac{\beta}{1+\beta}}}{P_0},\\
    \RenyiD{P_0}{P_1}{\frac{\alpha}{1+\alpha}} &= \alpha\KLD{P_{\frac{1}{1+\alpha}}}{P_0} + \KLD{P_{\frac{1}{1+\alpha}}}{P_1}.
\end{align}
When $\alpha=\beta=1$, we have
\begin{equation}
    \RenyiD{P_1}{P_0}{\frac{\beta}{1+\beta}}=\RenyiD{P_0}{P_1}{\frac{\alpha}{1+\alpha}}=\KLD{P_{\frac{1}{2}}}{P_1} + \KLD{P_{\frac{1}{2}}}{P_0}.
\end{equation}
Since $\lp \KLD{P_{\frac{1}{2}}}{P_0}, \KLD{P_{\frac{1}{2}}}{P_1} \rp$ is a point on the trade-off curve, we know that in this case the optimal error exponents of \semitwo\ tests are strictly smaller than \fullseq\ tests.

In general, to provide a necessary and sufficient condition on $\alpha,\beta$ such that there is a strict gap between \semitwo\ and \fullseq\ tests, one have to solve the equations
\begin{equation}
    \begin{cases}
        \KLD{P_\rho}{P_0} = \RenyiD{P_1}{P_0}{\frac{\beta}{1+\beta}}\\
        \KLD{P_\rho}{P_1} = \RenyiD{P_0}{P_1}{\frac{\alpha}{1+\alpha}}
    \end{cases}
\end{equation}
to obtain a relation between $\alpha$ and $\beta$. It is unclear whether the solution depends on $(P_0, P_1)$. So there may not exist a necessary and sufficient condition that does not depend on $(P_0,P_1)$ like \Cref{prop:efficient_compare}.

\subsection{Different Sequential Setups for Binary Classification}
Previously, we focus on \fixed, \semione, \semitwo, and \fullseq\ setups. Nevertheless, there are more possibilities, for example $X, T_0$ are sequentially observed and $T_1$ has fixed length.
In this case, the optimal type-II error exponent can be shown to be the same as that in the \fullseq\ setup, $\eseq(P_0,P_1)=\min\lbp \RenyiD{P_0}{P_1}{\frac{\alpha}{1+\alpha}},\ \kappa(P_0,P_1)\rbp$. This means the sequentiality in taking $P_0$-training samples is enough, and sequentially taking $P_1$-training samples does not further improve the error exponents.

However, when $X, T_1$ are sequentially observed and $T_0$ has fixed length, things are a bit different. Arrange the three sequences by $T_0,T_1,X$, then for $\bm P_1=(P_0,P_1,P_1)$, we have $\closure{\mscr P_0}\cap \mcal{S}(\bm P_1)=\{(P_1,P_1,P_1)\}$ and $\closure{\mscr P_0\cap \mcal{S}(\bm P_1)}=\emptyset$. Hence the condition in \Cref{prop:composite_achievability} is not satisfied. As a result, the upper bound is $\eseq$, while our test is only proved to achieve
\begin{align}
    &\min\Big\{ \RenyiD{P_0}{P_1}{\frac{\alpha}{1+\alpha}},\ \kappa(P_0,P_1),\nonumber\\
    &\hspace{1cm}\inf_{Q\in\mcal P(\mcal X):\alpha\KLD{Q}{P_1}\leq\bar\lambda(P_1,P_1)}\alpha\KLD{Q}{P_0}\Big\}.
\end{align}
If $\bar\lambda(P_1,P_1)=0$, then the third term is larger than $\RenyiD{P_0}{P_1}{\frac{\alpha}{1+\alpha}}$. Notice that for efficient tests, it is guaranteed that $\bar\lambda(P_1,P_1)=0$. But in general we cannot show the optimality of the achievability result.

\section{Conclusion}\label{sec:conclusion}
In this work, we consider a unified framework for binary classification with universality constraints on the expected stopping time and the type-I error exponent. Furthermore, we extend to a more general composite hypothesis testing problem and provide upper and lower bounds on the error exponents, along with a sufficient condition such that the bounds match. Applying the results to binary classification, we compare the optimal error exponents in different sequential setups in a fair manner to show the benefit of sequentiality. Under constant type-I error exponent constraint, sequentiality in taking testing samples suffices to achieve strict gain over fixed-length setup. On the other hand, for efficient tests, we derive necessary and sufficient condition such that sequentiality in taking training samples results in additional gain.

Throughout this work, we focus on distributions on fixed and finite alphabets.
It is interesting whether this kind of tests can be extended to more general parametric distributions such as exponential families. Since the method of types is heavily used in this work and is only applicable for distributions on finite support, an alternative tool for continuous distributions is required. 
In particular, kernel-based tests such as those based on the maximum mean discrepancy (MMD) metric are suitable for achievability and have been explored in the literature \cite{ZouLiang_17, SunZou_23,GerberJiang_23}. 
Another more practical direction is to consider data with high dimension or growing alphabet size.
In practice, our test is difficult to compute as it involves solving a lot of optimization problems. So one may attempt to design a test with lower computational complexity and simpler statistics.
Also, our focus is on the error exponents, which is an asymptotic notion. Although the test is asymptotically optimal, there is no performance guarantee for small sample size. A possible direction is to study the sample complexity under the assumption that the two underlying distributions are $\varepsilon$-separated (total variation at least $\varepsilon$), as in \cite{GerberPolyanskiy_24}. 

\begin{appendices}
    \section{Useful Tools}\label{sec:appendix_lemma}
Here we introduce some useful lemmas that will be used throughout the following proofs. These are well-known results, so the proofs are omitted.
\begin{lemma}[Optimal error exponent trade-off of binary hypothesis testing]\label{lemma:BHT_tradeoff}
    Let $P_0,P_1\in \interior{\mcal P(\mcal X)}$ be two distinct distributions, and $e_0\in(0,\KLD{P_1}{P_0}]$ be a constant. Then
    \begin{align}
        e_1(e_0)&:=\inf_{Q\in\mcal P(\mcal X):\KLD{Q}{P_0}<e_0}\KLD{Q}{P_1}\\
        &= \inf_{Q\in\mcal P(\mcal X):\KLD{Q}{P_0}\leq e_0}\KLD{Q}{P_1},
    \end{align}
    and the infimum is attained by $Q=P_\rho$ for some unique $\rho\in[0,1]$, where $P_\rho$ is a tilted distribution of $P_0$ towards $P_1$, described by
    \begin{equation}
        P_\rho(x) = \frac{P_0(x)^{1-\rho} P_1(x)^{\rho}}{\sum_{x'\in\mcal X}P_0(x')^{1-\rho} P_1(x')^{\rho}} \quad\text{for } x\in\mcal X.
    \end{equation}
    Moreover, if $P_0,P_1\in\mcal P_\varepsilon$, then by H\"older's inequality, $P_\rho\in\mcal P_\varepsilon$, thus
    \begin{align}
        e_1(e_0)&=\inf_{Q\in\mcal P_\varepsilon:\KLD{Q}{P_0}<e_0}\KLD{Q}{P_1} \\
        &= \inf_{Q\in\mcal P_\varepsilon:\KLD{Q}{P_0}\leq e_0}\KLD{Q}{P_1},
    \end{align}
    Lastly, consider the parametric curve $(e_0,e_1)=\big( \KLD{P_\rho}{P_0}, \KLD{P_\rho}{P_1} \big)$ for $\rho\in[0,1]$. % has two endpoints $(0,\KLD{P_0}{P_1})$ and $(\KLD{P_1}{P_0},0)$. 
    It can be shown that $\frac{\mathrm de_1}{\mathrm de_0}<0$ for $\rho\in(0,1)$, and $\frac{\mathrm de_1}{\mathrm de_0}=0$ at $\rho=1$.
\end{lemma}
\begin{lemma}\label{lemma:inf_over_union}
    Let $I$ be an index set and $A_i$ be a set for each $i\in I$. Let $A=\bigcup_{i\in I}A_i$ and $f:A\to [0,\infty)$ be a function. Then % $f:S\to \mbb R$, 
    %Let $S$ be a set and $f:S\to \mbb R$ be a function. If $S=\bigcup_{i\in I}S_i$ for some index set $I$, then
    \begin{equation}
        \inf_{x\in A} f(x)=\inf_{i\in I}\lp\inf_{x\in A_i} f(x)\rp.
    \end{equation}
\end{lemma}
\iffalse
\begin{lemma}\label{lemma:leq0}
    Let $K\subseteq \mbb R^m$ be a compact set and $S$ be a set. Let $f:K\times S\to\mbb R$ be a function. %For each $u\in K$, let $A_u=\lbp x\in S\mv f(u,x)\leq 0\rbp$. 
    If $f(u,x)$ is continuous in $u$ for all $x\in S$, then 
    %\[\bigcup_{u\in K}A_u=\lbp x\in S\mv \inf_{u\in K} f(u,x) \leq 0 \rbp.\]
    \[\bigcup_{u\in K}\lbp x\in S\mv f(u,x)\leq 0\rbp=\lbp x\in S\mv \inf_{u\in K} f(u,x) \leq 0 \rbp.\]
\end{lemma}
\fi

    \section{Proofs of General Composite Hypothesis Testing}\label{sec:appendix_general}
\subsection{Proof of Lemma~\ref{lemma:rewrite_Omega}}\label{app:rewrite_Omega}
By Lemma~\ref{lemma:inf_continuity}, 
\begin{align}
    g^{\bm P_1}(\bm Q) &= \inf_{\bm P_0'\in\mscr P_0\cap \mcal{S}(\bm P_1)} \Big( \sum_{i=1}^{\ell} \alpha_i \KLD{Q_i}{P_{0,i}'} - \lambda(\bm P_0') \Big)\\
    &= \inf_{\bm P_0'\in\closure{\mscr P_0\cap \mcal{S}(\bm P_1)}} \Big( \sum_{i=1}^{\ell} \alpha_i \KLD{Q_i}{P_{0,i}'} - \lambda(\bm P_0') \Big).
\end{align}
Notice that $g^{\bm P_1}(\bm Q)<0$ if and only if $\sum_{i=1}^{\ell} \alpha_i \KLD{Q_i}{P_{0,i}'} < \lambda(\bm P_0')$ for some $\bm P_0'\in\closure{\mscr P_0\cap \mcal{S}(\bm P_1)}$. Also, since $\closure{\mscr P_0\cap \mcal{S}(\bm P_1)}$ is compact and the KL divergence and $\lambda$ are continuous, $g^{\bm P_1}(\bm Q)\leq0$ if and only if $\sum_{i=1}^{\ell} \alpha_i \KLD{Q_i}{P_{0,i}'} \leq \lambda(\bm P_0')$ for some $\bm P_0'\in\closure{\mscr P_0\cap \mcal{S}(\bm P_1)}$. 
For simplicity, for $\bm P_0'\in\closure{\mscr P_0\cap \mcal{S}(\bm P_1)}$, let
\begin{align}
    \mcal B(\bm P_0') &= \lbp\bm Q\in \mcal{S}(\bm P_1)\mv \sum_{i=1}^{\ell} \alpha_i \KLD{Q_i}{P_{0,i}'} < \lambda(\bm P_0')\rbp,\\
    \bar{\mcal B}(\bm P_0') &= \lbp\bm Q\in \mcal{S}(\bm P_1)\mv \sum_{i=1}^{\ell} \alpha_i \KLD{Q_i}{P_{0,i}'} \leq \lambda(\bm P_0')\rbp.
\end{align}
Then $\Omega(\bm P_1) = \bigcup_{\bm P_0'\in\closure{\mscr P_0\cap \mcal{S}(\bm P_1)}}\mcal B(\bm P_0')$ and $\Omega'(\bm P_1) = \bigcup_{\bm P_0'\in\closure{\mscr P_0\cap \mcal{S}(\bm P_1)}}\bar{\mcal B}(\bm P_0')$.
Notice that if $\lambda(\bm P_0')>0$, then
\begin{equation}
    \inf_{\bm Q\in\mcal B(\bm P_0')} \sum_{i=1}^s \alpha_i \KLD{Q_i}{P_{1,i}} = \inf_{\bm Q\in\bar{\mcal B}(\bm P_0')} \sum_{i=1}^s \alpha_i \KLD{Q_i}{P_{1,i}}.
\end{equation}
On the other hand, if $\lambda(\bm P_0')=0$, then $\mcal B(\bm P_0')=\emptyset$ and $\bar{\mcal B}(\bm P_0')=\{\bm P_0'\}$.
For all $\bm P_0'\in\mscr P_0\cap \mcal{S}(\bm P_1)$, we have $\lambda(\bm P_0')>0$ and $\bm P_0'\in \mcal B(\bm P_0')$. Hence
\begin{align}
    &\quad\ \inf_{\bm Q\in\Omega(\bm P_1)} \sum_{i=1}^s \alpha_i \KLD{Q_i}{P_{1,i}} \\
    &\leq \inf_{\bm Q\in\mscr P_0\cap \mcal{S}(\bm P_1)} \sum_{i=1}^s \alpha_i \KLD{Q_i}{P_{1,i}} \\
    &= \inf_{\bm Q\in\closure{\mscr P_0\cap \mcal{S}(\bm P_1)}} \sum_{i=1}^s \alpha_i \KLD{Q_i}{P_{1,i}}.\label{eq:lambda=0}
\end{align}
By Lemma~\ref{lemma:inf_over_union}, 
\begin{align}
    &\quad\ \inf_{\bm Q\in\Omega(\bm P_1)} \sum_{i=1}^s \alpha_i \KLD{Q_i}{P_{1,i}}\\ 
    &= \inf_{\bm P_0'\in\closure{\mscr P_0\cap \mcal{S}(\bm P_1)}}\lp\inf_{\bm Q\in\mcal B(\bm P_0')} \sum_{i=1}^s \alpha_i \KLD{Q_i}{P_{1,i}}\rp\\
    &=\inf_{\substack{\bm P_0'\in\closure{\mscr P_0\cap \mcal{S}(\bm P_1)}\\\lambda(\bm P_0')>0}}\lp\inf_{\bm Q\in\mcal B(\bm P_0')} \sum_{i=1}^s \alpha_i \KLD{Q_i}{P_{1,i}}\rp.
\end{align}
Also,
\begin{align}
    &\inf_{\bm Q\in\Omega'(\bm P_1)} \sum_{i=1}^s \alpha_i \KLD{Q_i}{P_{1,i}}  \\
    &= \inf_{\bm P_0'\in\closure{\mscr P_0\cap \mcal{S}(\bm P_1)}}\lp\inf_{\bm Q\in\bar{\mcal B}(\bm P_0')} \sum_{i=1}^s \alpha_i \KLD{Q_i}{P_{1,i}}\rp \\
    &= \min\Bigg\{ \inf_{\substack{\bm P_0'\in\closure{\mscr P_0\cap \mcal{S}(\bm P_1)}\\\lambda(\bm P_0')>0}}\lp\inf_{\bm Q\in\bar{\mcal B}(\bm P_0')} \sum_{i=1}^s \alpha_i \KLD{Q_i}{P_{1,i}}\rp,\nonumber\\
    &\hspace{2.3cm}\inf_{\substack{\bm Q\in\closure{\mscr P_0\cap \mcal{S}(\bm P_1)}\\\lambda(\bm Q)=0}}\sum_{i=1}^s \alpha_i \KLD{Q_i}{P_{1,i}}\Bigg\}\label{eq:min_of_two_trash}\\
    &=\inf_{\bm Q\in\Omega(\bm P_1)} \sum_{i=1}^s \alpha_i \KLD{Q_i}{P_{1,i}}. 
\end{align}
The last equality holds since the second term in \eqref{eq:min_of_two_trash} is lower bounded by \eqref{eq:lambda=0}.

\subsection{Proof of Lemma~\ref{lemma: converse}}\label{app:converse_lemma}
Using data processing inequality of divergence:
\begin{align}
    &\BKLD{\mbb P_0\{\mcal{E}\}}{\mbb{P}_{1}\{\mcal{E}\}}
    \leq \KLD{\mbb{P}_0}{\mbb{P}_{1}}|_{\mcal{F}_\tau}\\
    &=\mbb E_0\Bigg[ \sum_{i=1}^\ell \sum_{k=1}^{N_i}\log\frac{P_{0,i}(X_{i,k})}{P_{1,i}(X_{i,k})} %\nonumber\\&\hspace{3cm}
    +\sum_{i=\ell+1}^s\sum_{k=1}^{\lce\alpha_i\tau\rce}\log\frac{P_{0,i}(X_{i,k})}{P_{1,i}(X_{i,k})} \,\Bigg|\, \bm P_0\Bigg]\\
    &=\sum_{i=1}^\ell N_i \KLD{P_{0,i}}{P_{1,i}} + \sum_{i=\ell+1}^s \mbb{E}_0\lb\lce\alpha_i\tau\rce|\bm P_0\rb \KLD{P_{0,i}}{P_{1,i}},
\end{align}
where the last equality is by Wald's identity. One can easily check that the assumptions for Wald's identity are satisfied since $\tau$ is a stopping time and $X_{i,k}$'s are i.i.d. sampled from a distribution with full support on a finite alphabet.

\subsection{Proof of Lemma~\ref{lemma:limit}}\label{app:limit}
Recall that
\begin{align}
    \mu_n(\bm P_1) &= \nonumber\\
    &\hspace{-0.5cm}\inf_{\substack{\bm Q\in\mscr D\\g_n(\bm Q) < 0}} \Big( \sum_{i=1}^\ell \alpha_i \KLD{Q_i}{P_{1,i}} + n\sum_{i=\ell+1}^s \alpha_i \KLD{Q_i}{P_{1,i}} \Big),\\
    g_n(\bm Q) &= \inf_{\bm P_0'\in\mscr P_0} \Big( \sum_{i=1}^\ell \alpha_i \KLD{Q_i}{P_{0,i}'} \nonumber\\
    &\hspace{1.3cm}+ n\sum_{i=\ell+1}^s \alpha_i \KLD{Q_i}{P_{0,i}'} - \lambda(\bm P_0') \Big),\\
    \mcal{S}(\bm P_1) &= \lbp \bm Q\in\mscr D \mv Q_i=P_{1,i}\ \forall\,i=\ell+1,\dots,s\rbp\\
    \bar\Omega (\bm P_1) &= \lbp\bm Q\in \mcal{S}(\bm P_1) \mv \bar g^{\bm P_1}(\bm Q) \leq 0\rbp,\\
    \bar g^{\bm P_1}(\bm Q) &= \inf_{\bm P_0'\in\closure{\mscr P_0}\cap \mcal{S}(\bm P_1)} \Big( \sum_{i=1}^{\ell} \alpha_i \KLD{Q_i}{P_{0,i}'} - \bar\lambda(\bm P_0') \Big),\\
    \Omega(\bm P_1) &= \lbp\bm Q\in \mcal{S}(\bm P_1) \mv g^{\bm P_1}(\bm Q) < 0\rbp,\\
    g^{\bm P_1}(\bm Q) &= \inf_{\bm P_0'\in\mscr P_0\cap \mcal{S}(\bm P_1)} \Big( \sum_{i=1}^{\ell} \alpha_i \KLD{Q_i}{P_{0,i}'} - \lambda(\bm P_0') \Big).
\end{align}
For simplicity, let
\begin{align}
    \mu(\bm P_1) &= \inf_{\bm Q\in\bar\Omega (\bm P_1)} \sum_{i=1}^s \alpha_i \KLD{Q_i}{P_{1,i}}\\
    &= \inf_{\bm Q\in\bar\Omega (\bm P_1)} \sum_{i=1}^\ell \alpha_i \KLD{Q_i}{P_{1,i}}.
\end{align}
The second equality holds because $\bar\Omega (\bm P_1)\subseteq \mcal{S}(\bm P_1)$. 
To show that $\{\mu_n(\bm P_1)\}$ is non-decreasing in $n$, observe that for any $\bm Q\in\mscr D$, $\{g_n(\bm Q)\}$ is non-decreasing in $n$. So
\begin{align}
    &\quad\ \mu_n(\bm P_1) \\
    &=\inf_{\substack{\bm Q\in\mscr D\\g_n(\bm Q) < 0}} \hspace{-0.2cm}\Big( \sum_{i=1}^\ell \alpha_i \KLD{Q_i}{P_{1,i}} + n\sum_{i=\ell+1}^s \alpha_i \KLD{Q_i}{P_{1,i}} \Big)\\
    &\leq \hspace{-0.2cm}\inf_{\substack{\bm Q\in\mscr D\\g_{n+1}(\bm Q) < 0}} \hspace{-0.2cm}\Big( \sum_{i=1}^\ell \alpha_i \KLD{Q_i}{P_{1,i}} + n\sum_{i=\ell+1}^s \alpha_i \KLD{Q_i}{P_{1,i}} \Big)\\
    &\leq \hspace{-0.2cm}\inf_{\substack{\bm Q\in\mscr D\\g_{n+1}(\bm Q) < 0}} \hspace{-0.2cm}\Big( \sum_{i=1}^\ell \alpha_i \KLD{Q_i}{P_{1,i}} %\nonumber\\&\hspace{3cm}
    + (n+1)\hspace{-0.2cm}\sum_{i=\ell+1}^s \alpha_i \KLD{Q_i}{P_{1,i}} \Big)\\
    &= \mu_{n+1}(\bm P_1).
\end{align}
If $\{\mu_n(\bm P_1)\}$ is unbounded, the statement is trivially true. Otherwise, by the monotone convergence theorem, $\{\mu_n(\bm P_1)\}$ converges. Denote the limit by $\lim_{n\to\infty}\mu_n(\bm P_1) = E$. For any $n$, there exist $\bm Q^n\in\mscr D$ and $\bm P_0^n\in\mscr P_0$ such that
\begin{align}
    \sum_{i=1}^\ell \alpha_i \KLD{Q_i^n}{P_{1,i}} + n\sum_{i=\ell+1}^s \alpha_i \KLD{Q_i^n}{P_{1,i}} 
    &\leq \mu_n(\bm P_1)+\frac{1}{n}\\
    &\leq E+\frac{1}{n},\label{eq:fn_bound}\\
    \sum_{i=1}^\ell \alpha_i \KLD{Q_i^n}{P_{0,i}^n} + n\sum_{i=\ell+1}^s \alpha_i \KLD{Q_i^n}{P_{0,i}^n} &- \lambda(\bm P_0^n) <0.\label{eq:gn_bound}
\end{align}
Let $\bar{\bm Q}\in\mscr D$ and $\bar{\bm P}_0\in\closure{\mscr P_0}$ be limit points of the sequences $\{\bm Q^n\}$ and $\{\bm P_0^n\}$.
For $i=\ell+1,\dots,s$, by \eqref{eq:fn_bound}, $\KLD{Q_i^n}{P_{1,i}}\leq \frac{1}{n}(E+\frac{1}{n})\to 0$ as $n\to\infty$, and hence $\bar{Q}_i = P_{1,i}$.
Since $\lambda$ can be extended to a continuous function $\bar\lambda:\closure{\mscr P_0}\to[0,\infty)$ and $\closure{\mscr P_0}$ is compact, we know that $\bar\lambda$ is bounded, and so is $\lambda$. By \eqref{eq:gn_bound}, it can be shown similarly that $\bar{Q}_i = \bar{P}_{0,i}$, thus $\bar{P}_{0,i}=P_{1,i}$. As a result, we have $\bar{\bm Q}, \bar{\bm P}_0\in \mcal{S}(\bm P_1)$. 
If $\closure{\mscr P_0}\cap \mcal{S}(\bm P_1)=\emptyset$, then $\bar{\bm P}_0$ makes a contradiction. Otherwise, by the continuity of KL divergence and $\bar\lambda$, we have
\begin{align}
    &\sum_{i=1}^\ell \alpha_i \KLD{\bar{Q}_i}{P_{1,i}}
    \leq E,\\
    \bar g^{\bm P_1}(\bar{\bm Q}) \leq &\sum_{i=1}^\ell \alpha_i \KLD{\bar{Q}_i}{\bar{P}_{0,i}}  - \bar\lambda(\bar{\bm P}_0) \leq 0.
\end{align}
If $\bar\Omega (\bm P_1)=\emptyset$, there is again a direct contradiction. Otherwise $E\geq \mu(\bm P_0)$ by the definition of $\mu(\bm P_0)$.

\subsection{Proof of Proposition~\ref{prop:substituting_lambda}}\label{app:substituting_lambda}
Recall that the divergence ball centered at $\bm P_0\in\mscr P_0$ with radius $r$ is
\begin{equation}
    \mcal B_r(\bm P_0) = \big\{ \bm Q\in\mscr D \,\big\vert\, \sum_{i=1}^s \alpha_i \KLD{Q_i}{P_{0,i}} < r \big\}.
\end{equation}
Also, $\Omega_0=\bigcup_{\bm P_0'\in\mscr P_0}\mcal B_\lambda(\bm P_0')$,
where $\mcal B_\lambda(\bm P_0)$ denotes $\mcal B_{\lambda(\bm P_0)}(\bm P_0)$ for simplicity.

If $\Omega_0=\mscr D$, it is done by simply taking $\tilde\lambda(\bm P_0) = \sup_{\bm P_0'\in\mscr P_0}\lambda(\bm P_0')$ as $\lambda$ is bounded. If $\Omega_0\neq\mscr D$, let $\bar\lambda(\bm P_0)=\sup\big\{ r\geq0 \,\big\vert\, \mcal B_r(\bm P_0) \subseteq\Omega_0 \big\}$ for $\bm P_0\in\closure{\mscr P_0}$. The function $\bar\lambda$ is well-defined because there exists $\bm Q\in\mscr D\setminus \Omega_0$ and $\sum_{i=1}^s \alpha_i \KLD{Q_i}{P_{0,i}}$ is bounded as $\closure{\mscr P_0}\subseteq \interior{\mscr D}$. 
Take $\tilde\lambda = \bar\lambda|_{\mscr P_0}$, and we first show that $\mcal B_{\tilde\lambda}(\bm P_0)\subseteq\Omega_0$ for each $\bm P_0\in\mscr P_0$. Suppose there exists $\bm Q\in \mcal B_{\tilde\lambda}(\bm P_0)\setminus\Omega_0$, then we have $\sum_{i=1}^s \alpha_i \KLD{Q_i}{P_{0,i}} < \tilde\lambda(\bm P_0)$. By definition, there exists $\sum_{i=1}^s \alpha_i \KLD{Q_i}{P_{0,i}}<r\leq \tilde\lambda(\bm P_0)$ such that $\mcal B_r(\bm P_0)\subseteq\Omega_0$. Hence $\bm Q\in \mcal B_r(\bm P_0)\subseteq\Omega_0$, which makes a contradiction. 
It remains to show that $\bar\lambda$ is continuous. For any $\bm P_0\in\closure{\mscr P_0}$, suppose $\bar\lambda$ is not continuous at $\bm P_0$. Then there exists $\epsilon>0$ such that $\forall\, m\in\mbb N$, $\exists\, \bm P_0^m\in\closure{\mscr P_0}$ with $\lnorm \bm P_0^m - \bm P_0\rnorm_1<\frac{1}{m}$ and $|\bar\lambda(\bm P_0^m) - \bar\lambda(\bm P_0)|\geq \epsilon$.
\setcounter{case}{0}
\begin{case}
    For infinitely many $m$, $\bar\lambda(\bm P_0^m)\leq \bar\lambda(\bm P_0) - \epsilon$. Let $r = \bar\lambda(\bm P_0) - \epsilon/2$ and focus on the infinite subsequence. For each $m$ in the subsequence, there exists $\bm Q^m\in\mcal B_r(\bm P_0^m)\cap\comp{\Omega_0}$.
    Let $\bar{\bm Q}$ be a limit point of $\{\bm Q^m\}$. 
    Since $\mcal B_\lambda(\bm P_0')$ is open for all $\bm P_0'\in\mscr P_0$, we know that $\Omega_0$ is open and hence $\comp{\Omega_0}$ is closed (and also bounded). Therefore $\bar{\bm Q}\in\comp{\Omega_0}$. However, we have $\sum_{i=1}^s \alpha_i \KLD{Q_i^m}{P_{0,i}^m}<r$ for each $m$ in the subsequence and by continuity of KL divergence,
    \begin{equation}
        \sum_{i=1}^s \alpha_i \KLD{\bar Q_i}{P_{0,i}}\leq r<\bar\lambda(\bm P_0) - \epsilon/4.
    \end{equation}
    This implies $\bar{\bm Q}\in \mcal B_{\bar\lambda(\bm P_0) - \epsilon/4}(\bm P_0)\subseteq \Omega_0$, which leads to contradiction.
\end{case}
\begin{case}
    For infinitely many $m$, $\bar\lambda(\bm P_0^m)\geq \bar\lambda(\bm P_0) + \epsilon$. Let $r = \bar\lambda(\bm P_0) + \epsilon/2$ and focus on the infinite subsequence. For each $m$ in the subsequence, $\mcal B_r(\bm P_0^m)\subseteq \Omega_0$. 
    Take $r'=\bar\lambda(\bm P_0)+\epsilon/4$ and notice that for each $i=1,\dots,s$,
    \begin{equation}
        \big | \KLD{Q_i}{P_{0,i}^m} - \KLD{Q_i}{P_{0,i}}\big|\leq \sum_{x\in\mcal X_i}\big|\log P_{0,i}^m(x)-\log P_{0,i}(x)\big|.
    \end{equation}
    Since $\bm P_0^m\to\bm P_0$, we can find $m$ large enough such that for all $\bm Q\in \mcal B_{r'}(\bm P_0)$,
    \begin{equation}
        \sum_{i=1}^s \alpha_i \KLD{Q_i}{P_{0,i}^m}\leq \sum_{i=1}^s \alpha_i \KLD{Q_i}{P_{0,i}} + \epsilon/8< r.
    \end{equation}
    Thus $\mcal B_{r'}(\bm P_0)\subseteq \mcal B_r(\bm P_0^m)\subseteq \Omega$, leading to contradiction.
\end{case}

    \section{Proofs of Binary Classification}\label{sec:appendix_classification}
\subsection{Proof of Proposition~\ref{prop:constant_compare}}\label{app:constant_compare}
When $\lambda(P_0,P_1)\equiv \lambda_0$, we have
\begin{align}
    &\quad\ g(Q_0,Q_1) \\
    &=\inf_{P_1'\in\mcal{P}_\varepsilon\setminus \{P_1\}} \Big(\alpha\KLD{Q_0}{P_1} + \beta\KLD{Q_1}{P_1'} - \lambda(P_1,P_1')\Big)\\
    &\geq \alpha\KLD{Q_0}{P_1}-\lambda_0
\end{align}
and hence
\begin{align}
    \mu(P_0,P_1) &= \inf_{\substack{Q_0,Q_1\in\mcal P(\mcal X)\\g(Q_0,Q_1) < 0}} \Big(\alpha\KLD{Q_0}{P_0} + \beta\KLD{Q_1}{P_1}\Big)\\
    &\geq \inf_{\substack{Q_0\in\mcal P(\mcal X)\\\alpha\KLD{Q_0}{P_1}<\lambda_0}} \alpha\KLD{Q_0}{P_0}\label{eq:ttttt}
\end{align} 
Moreover,
\begin{align}
    &\quad\ \kappa(P_0,P_1) \\
    &=\hspace{-0.2cm}\inf_{\substack{(Q_0,Q_1)\in\distset\\g_1(Q_1, Q_0,Q_1)<0}} \hspace{-0.2cm}\Big( \KLD{Q_1}{P_1} + \alpha\KLD{Q_0}{P_0} + \beta\KLD{Q_1}{P_1} \Big)\\
    &\leq \inf_{\substack{Q_0\in\mcal{P}_\varepsilon\setminus \{P_1\}\\g_1(P_1, Q_0,P_1)<0}} \alpha\KLD{Q_0}{P_0},\label{eq:jjjjj}
\end{align}
and
\begin{align}
    &\quad\ g_1(P_1,Q_0,P_1) \\
    &= \hspace{-0.15cm}\inf_{(P_0',P_1')\in\distset} \hspace{-0.15cm}\Big(\KLD{P_1}{P_0'} + \alpha\KLD{Q_0}{P_0'} + \beta\KLD{P_1}{P_1'}\Big) - \lambda_0\\
    &\leq\alpha\KLD{Q_0}{P_1}-\lambda_0.
\end{align}
With the help of Lemma~\ref{lemma:BHT_tradeoff}, the proof is complete by observing that \eqref{eq:jjjjj} is upper bounded by \eqref{eq:ttttt}.

\subsection{Proof of Proposition~\ref{prop:constant_strict_gain}}\label{app:constant_strict_gain}
First we rewrite $\efix$. By Lemma~\ref{lemma:inf_continuity},
\begin{align}
    &\quad\ g_1(Q,Q_0,Q_1) \\
    &= \hspace{-0.15cm}\inf_{(P_0',P_1')\in\distset} \hspace{-0.15cm}\Big(\KLD{Q}{P_0'} + \alpha\KLD{Q_0}{P_0'} + \beta\KLD{Q_1}{P_1'}\Big) - \lambda_0\\
    &= \hspace{-0.15cm}\inf_{P_0',P_1'\in\mcal P_\varepsilon} \hspace{-0.15cm}\Big(\KLD{Q}{P_0'} + \alpha\KLD{Q_0}{P_0'} + \beta\KLD{Q_1}{P_1'}\Big) - \lambda_0.
\end{align}
For $P_0',P_1'\in\mcal P_\varepsilon$, let $\mcal B(P_0',P_1') =$
\begin{align}
    \Big\{ &Q,Q_0,Q_1\in\mcal P(\mcal X)\,\Big|\nonumber\\
    &\hspace{0.5cm}\KLD{Q}{P_0'} + \alpha\KLD{Q_0}{P_0'} + \beta\KLD{Q_1}{P_1'} <\lambda_0 \Big\}
\end{align}
By the continuity and convexity of KL divergence, we know that $\closure{\mcal B(P_0',P_1')} = $
\begin{align}
    \Big\{ &Q,Q_0,Q_1\in\mcal P(\mcal X)\,\Big|\nonumber\\
    &\hspace{0.5cm}\KLD{Q}{P_0'} + \alpha\KLD{Q_0}{P_0'} + \beta\KLD{Q_1}{P_1'} \leq\lambda_0 \Big\}.
\end{align}
Notice that 
\begin{align}
    &\quad\ \Big\{ Q,Q_0,Q_1\in\mcal P(\mcal X)\,\Big|\, g_1(Q,Q_0,Q_1) < 0\Big\} \\
    &= \bigcup_{P_0',P_1'\in\mcal P_\varepsilon}\mcal B(P_0',P_1')
\end{align}
By Lemma~\ref{lemma:inf_over_union} and Lemma~\ref{lemma:BHT_tradeoff}, 
\begin{align}
    &\quad\ \efix(P_0,P_1) \\
    &= \inf_{\substack{Q,Q_0,Q_1\in\mcal P(\mcal X)\\g_1(Q,Q_0,Q_1) < 0}} \Big(\KLD{Q}{P_1} + \alpha\KLD{Q_0}{P_0} + \beta\KLD{Q_1}{P_1}\Big) \\
    &= \inf_{P_0',P_1'\in\mcal P_\varepsilon} \bigg\{\inf_{(Q,Q_0,Q_1)\in \mcal B(P_0',P_1')} \Big(\KLD{Q}{P_1} + \alpha\KLD{Q_0}{P_0} \nonumber\\
    &\hspace{4.9cm}+ \beta\KLD{Q_1}{P_1}\Big)\bigg\} \\
    &= \inf_{P_0',P_1'\in\mcal P_\varepsilon} \bigg\{\inf_{(Q,Q_0,Q_1)\in \closure{\mcal B(P_0',P_1')}\cap \mcal P_\varepsilon^3} \Big(\KLD{Q}{P_1} \nonumber\\
    &\hspace{3.9cm}+ \alpha\KLD{Q_0}{P_0} + \beta\KLD{Q_1}{P_1}\Big)\bigg\} \\
    &= \inf_{\substack{Q,Q_0,Q_1\in\mcal P_\varepsilon\\g_1(Q,Q_0,Q_1) \leq 0}} \hspace{-0.2cm}\Big(\KLD{Q}{P_1} + \alpha\KLD{Q_0}{P_0} + \beta\KLD{Q_1}{P_1}\Big). \label{eq:rewrite_e1}
\end{align}
The last equality follows from the fact that since $\mcal P_\varepsilon^2$ is compact and the KL divergence is continuous in the concerned region, $g_1(Q,Q_0,Q_1)\leq 0$ if and only if $\KLD{Q}{P_0'} + \alpha\KLD{Q_0}{P_0'} + \beta\KLD{Q_1}{P_1'} \leq \lambda_0$ for some $P_0',P_1'\in\mcal P_\varepsilon$.

Notice that for $Q,Q_0,Q_1\in\mcal P_\varepsilon$, 
\begin{align}
    &\quad\ g_1(Q,Q_0,Q_1) \\
    &= \inf_{P_0',P_1'\in\mcal P_\varepsilon} \hspace{-0.1cm}\Big(\KLD{Q}{P_0'} + \alpha\KLD{Q_0}{P_0'} + \beta\KLD{Q_1}{P_1'}\Big) - \lambda_0 \\
    &= \inf_{P_0'\in\mcal P_\varepsilon} \Big(\KLD{Q}{P_0'} + \alpha\KLD{Q_0}{P_0'}\Big) - \lambda_0 \\
    &=\GJS{Q_0}{Q}{\alpha} - \lambda_0. \label{eq:rewrite_g1}
\end{align}
The last equality follows from the fact that the GJS divergence can be written as the following minimization problem $\GJS{P}{Q}{\alpha} = \min_{V\in\mcal{P}(\mcal{X})} \{\alpha\KLD{P}{V} + \KLD{Q}{V}\}$, with the minimizer $\frac{\alpha P+Q}{\alpha+1}$.
Combining \eqref{eq:rewrite_e1} and \eqref{eq:rewrite_g1}, 
\begin{align}
    &\quad\ \efix(P_0,P_1)\\
    &= \inf_{\substack{Q,Q_0\in\mcal P_\varepsilon\\ \GJS{Q_0}{Q}{\alpha} \leq \lambda_0}} \hspace{-0.2cm}\Big(\KLD{Q}{P_1} + \alpha\KLD{Q_0}{P_0} + \beta\KLD{Q_1}{P_1}\Big) \\
    &=\inf_{\substack{Q,Q_0,Q_1\in\mcal P_\varepsilon\\ \GJS{Q_0}{Q}{\alpha} \leq \lambda_0}} \Big(\KLD{Q}{P_1} + \alpha\KLD{Q_0}{P_0}\Big).\label{eq:rewrite_e1_final}
\end{align}

Next we rewrite $\kappa(P_0,P_1)$. 
By \eqref{eq:rewrite_g1}, we know that for $Q_0,Q_1\in\mcal P_\varepsilon$, $g_1(Q_1, Q_0,Q_1) = \GJS{Q_0}{Q_1}{\alpha} - \lambda_0$.
Hence, for any $Q\in\mcal P_\varepsilon$, the pair $(Q,Q)$ can be approximated by $(Q_0,Q_1)\in\distset$ with $g_1(Q_1, Q_0,Q_1)<0$, since we can simply choose $Q_0=Q$ and $Q_1\in\mcal P_\varepsilon\setminus\{Q\}$ arbitrarily close to $Q$.
By Lemma~\ref{lemma:inf_continuity}, 
we can slightly change the set where the infimum is taken over in $\kappa(P_0,P_1)$.
\begin{align}
    &\quad\ \kappa(P_0,P_1) \\
    &= \hspace{-0.2cm}\inf_{\substack{(Q_0,Q_1)\in\distset\\g_1(Q_1, Q_0,Q_1)<0}} \hspace{-0.2cm}\Big( \KLD{Q_1}{P_1} + \alpha\KLD{Q_0}{P_0} + \beta\KLD{Q_1}{P_1} \Big)\\
    &= \inf_{\substack{Q_0,Q_1\in\mcal P_\varepsilon\\g_1(Q_1, Q_0,Q_1)<0}} \hspace{-0.2cm}(1+\beta)\KLD{Q_1}{P_1} + \alpha\KLD{Q_0}{P_0}.
\end{align}
Then following similar arguments in rewriting $\efix(P_0,P_1)$, we get the desired result
\begin{equation}
    \kappa(P_0,P_1) = \inf_{\substack{Q_0,Q_1\in\mcal P_\varepsilon\\\GJS{Q_0}{Q_1}{\alpha}\leq\lambda_0}} \hspace{-0.2cm}(1+\beta)\KLD{Q_1}{P_1} + \alpha\KLD{Q_0}{P_0}.
\end{equation}
It is then obvious that if $\GJS{P_0}{P_1}{\alpha}\leq \lambda_0$, then $\efix(P_0,P_1)$ and $\kappa(P_0,P_1)$ are $0$.
In the following we focus on the case $\GJS{P_0}{P_1}{\alpha}> \lambda_0$. 
Notice that the infimum in \eqref{eq:rewrite_e1_final} is attainable since $\mcal P_\varepsilon$ is compact and the GJS divergence is continuous. Also, it is $0$ only if $Q=P_1$ and $Q_0=P_0$. However, $\GJS{P_0}{P_1}{\alpha}> \lambda_0$, so we have $\efix(P_0,P_1)>0$.
To show that $\efix(P_0,P_1)<\kappa(P_0,P_1)$ and $\kappa(P_0,P_1)$ is strictly increasing in $\beta$, it suffices to show that for any $b\geq 1$, the following infimum is not attained by $Q_1=P_1$:
\begin{equation}
    \inf_{\substack{Q_0,Q_1\in\mcal P_\varepsilon\\\GJS{Q_0}{Q_1}{\alpha}\leq\lambda_0}} b\KLD{Q_1}{P_1} + \alpha\KLD{Q_0}{P_0}.\label{eq:inf_problem}
\end{equation}
Suppose the infimum is attained by $(Q_0, P_1)$. 
Let $V^*=\frac{\alpha Q_0+P_1}{\alpha+1}$, $c_0=\alpha\KLD{Q_0}{V^*}$ and $c_1 = \KLD{P_1}{V^*}$.
Since $\lambda_0>0$, it can be shown that $Q_0\neq P_1$. Thus $V^*\neq P_1$ and $c_1>0$. Also, by the above arguments, we know that $Q_0\neq P_0$. As a result, we must have $c_0<\alpha\KLD{P_0}{V^*}$. For $y\geq 0$, let
\begin{align}
    Q_0^y&=\argmin_{Q\in\mcal P_\varepsilon: \alpha\KLD{Q}{V^*}\leq c_0+y} \alpha\KLD{Q}{P_0},\\
    Q_1^y&=\argmin_{Q\in\mcal P_\varepsilon: \KLD{Q}{V^*}\leq c_1-y} b\KLD{Q}{P_1},
\end{align}
and $f(y) = b\KLD{Q_1^y}{P_1} + \alpha\KLD{Q_0^y}{P_0}$. Notice that $f(0)\leq \alpha\KLD{Q_0}{P_0}$. By Lemma~\ref{lemma:BHT_tradeoff}, we have 
\begin{equation}
    \frac{\mathrm d}{\mathrm dy} b\KLD{Q_1^y}{P_1}\Big|_{y=0} = 0\quad\text{and}\quad \frac{\mathrm d}{\mathrm dy} \alpha\KLD{Q_0^y}{P_0}\Big|_{y=0} < 0.
\end{equation}
Hence $f'(0)<0$, and there exists some small $y>0$ such that $f(y)<f(0)$. Also,
\begin{align}
    \GJS{Q_0^y}{Q_1^y}{\alpha} &\leq \alpha\KLD{Q_0^y}{V^*} + \KLD{Q_1^y}{V^*} \\
    &\leq c_0+y+c_1-y \\
    &= \GJS{Q_0}{Q_1}{\alpha}\leq\lambda_0.
\end{align}
It is clear that $Q_1^y\neq P_1$ and this contradicts with the assumption that the infimum in \eqref{eq:inf_problem} is attained by $(Q_0, P_1)$.

Finally, to show strict gain of sequentiality, it remains to show that $\efix(P_0,P_1)<\RenyiD{P_0}{P_1}{\frac{\alpha}{1+\alpha}}$. Let $\bar P = \argmin_{V\in\mcal P_\varepsilon} \alpha\KLD{V}{P_0}+ \KLD{V}{P_1}$. Since $\lambda(\bar P, P_1)=\lambda_0>0$, we can pick $Q=\eta P_1 + (1-\eta)\bar P$, $Q_0=\eta P_0 + (1-\eta)\bar P$ for some $\eta>0$ small enough such that $(Q,Q_0,P_1)\in\mcal B(\bar P, P_1)$ and then
\begin{align}
    \efix(P_0,P_1)&\leq \KLD{Q}{P_1} + \alpha\KLD{Q_0}{P_0} + \beta\KLD{P_1}{P_1}\\
    &<\KLD{\bar P}{P_1} + \alpha\KLD{\bar P}{P_0}=\RenyiD{P_0}{P_1}{\frac{\alpha}{1+\alpha}}.
\end{align}
As a remark, $\efix(P_0,P_1)<\RenyiD{P_0}{P_1}{\frac{\alpha}{1+\alpha}}$ holds in general and does not limit to the case of constant $\lambda$.

\subsection{Proof of Proposition~\ref{prop:kappa=infty}}\label{app:kappa=infty}
Following the proof of converse in \cite{HsuWang_22}, we can obtain $\lambda(P_0,P_1)\leq e_0(P_0,P_1)\leq \RenyiD{P_1}{P_0}{\frac{\beta}{1+\beta}}$, where the first inequality follows from the universality constraint on the type-I error exponent. Then for $(Q_0,Q_1)\in\distset$,
\begin{align}
    &\quad\ g_1(Q_1,Q_0,Q_1)\\
    &= \inf_{(P_0',P_1')\in\distset} \Big(\KLD{Q_1}{P_0'} + \alpha\KLD{Q_0}{P_0'} + \beta\KLD{Q_1}{P_1'} \nonumber\\
    &\hspace{5.5cm}- \lambda(P_0',P_1')\Big) \\
    &\geq \inf_{(P_0',P_1')\in\distset} \Big(\KLD{Q_1}{P_0'} + \alpha\KLD{Q_0}{P_0'} + \beta\KLD{Q_1}{P_1'} \nonumber\\
    &\hspace{4.8cm}- \RenyiD{P_1'}{P_0'}{\frac{\beta}{1+\beta}}\Big) \\
    &\geq \inf_{(P_0',P_1')\in\distset} \alpha\KLD{Q_0}{P_0'}\label{eq:kappa=infty1}\\
    &\geq 0 
\end{align}
Note that \eqref{eq:kappa=infty1} follows from the fact that R\'enyi divergence can be written as a minimization problem. We then have $\kappa(P_0,P_1)=\infty$ as it is the infimum of an empty set.

\subsection{Proof of Proposition~\ref{prop:efficient_compare}}\label{app:efficient_compare}
First we show that if $\alpha\beta\geq1$, then $\mu(P_0,P_1)\geq \RenyiD{P_0}{P_1}{\frac{\alpha}{1+\alpha}}$.
Recall that
\begin{align}
    \mu(P_0,P_1) &= \inf_{\substack{Q_0,Q_1\in\mcal P(\mcal X)\\g(Q_0,Q_1) < 0}} \Big(\alpha\KLD{Q_0}{P_0} + \beta\KLD{Q_1}{P_1}\Big),\\
    g(Q_0,Q_1) &= \nonumber\\
    &\hspace{-0.5cm}\inf_{P_1'\in\mcal{P}_\varepsilon\setminus \{P_1\}} \alpha\KLD{Q_0}{P_1} + \beta\KLD{Q_1}{P_1'} - \lambda(P_1,P_1').
\end{align}
Now $\lambda(P_0,P_1)=\RenyiD{P_1}{P_0}{\frac{\beta}{1+\beta}}$ for all $(P_0,P_1)\in\distset$. Hence it is equivalent to show that if for some $Q_0,Q_1\in\mcal P(\mcal X)$,
\begin{equation}
    \alpha\KLD{Q_0}{P_0} + \beta\KLD{Q_1}{P_1} < \RenyiD{P_0}{P_1}{\frac{\alpha}{1+\alpha}}, 
    \label{eq:ineq1}
\end{equation}
then for all $P_1'\in\mcal{P}_\varepsilon\setminus \{P_1\}$, we have
\begin{equation}
    \alpha\KLD{Q_0}{P_1} + \beta\KLD{Q_1}{P_1'} \geq \RenyiD{P_1'}{P_1}{\frac{\beta}{1+\beta}}.
\end{equation}
Since the R\'{e}nyi divergence can be written as a minimization problem \eqref{eq:Renyi}, we upper bound the RHS of \eqref{eq:ineq1} by $\KLD{Q_0}{P_1} + \alpha\KLD{Q_0}{P_0}$ and get $\beta \KLD{Q_1}{P_1} < \KLD{Q_0}{P_1}$. For all $P_1'\in\mcal{P}_\varepsilon\setminus \{P_1\}$,
\begin{align}
    &\quad\ \alpha\KLD{Q_0}{P_1} + \beta\KLD{Q_1}{P_1'}\\
    &\geq 
    \alpha\beta\KLD{Q_1}{P_1} + \beta\KLD{Q_1}{P_1'}\\
    &\geq 
    \KLD{Q_1}{P_1} + \beta\KLD{Q_1}{P_1'}\\
    &\geq \RenyiD{P_1'}{P_1}{\frac{\beta}{1+\beta}}.
\end{align}
Second, we prove that when $\alpha\beta<1$, there exists $(P_0,P_1)\in\distset$ such that $\mu(P_0,P_1)< \RenyiD{P_0}{P_1}{\frac{\alpha}{1+\alpha}}$.
Specifically, we want to find $P_1\in\mcal P_\varepsilon$, $P_0,P_1'\in\mcal{P}_\varepsilon\setminus \{P_1\}$, and $Q_0,Q_1\in\mcal P(\mcal X)$ such that
\begin{align}
    &g(Q_0,Q_1)\\
    &\leq \alpha\KLD{Q_0}{P_1} + \beta\KLD{Q_1}{P_1'} - \beta\KLD{V_\beta}{P_1'} - \KLD{V_\beta}{P_1}\\
    &<0,
\end{align}
and
\begin{align}
    \mu(P_0,P_1)&= \alpha\KLD{Q_0}{P_0} + \beta\KLD{Q_1}{P_1} \\
    &< \alpha\KLD{V_\alpha}{P_0} + \KLD{V_\alpha}{P_1},
\end{align}
where
\begin{align}
    V_\alpha &= \argmin_{V\in\mcal{P}(\mcal{X})} \lbp \alpha\KLD{V}{P_0} + \KLD{V}{P_1} \rbp,\\
    \text{ and }
    V_\beta &= \argmin_{V\in\mcal{P}(\mcal{X})} \lbp \beta\KLD{V}{P_1'} + \KLD{V}{P_1} \rbp.
\end{align}
Take $Q_0=V_\alpha$ and $Q_1=V_\beta$, then it suffices to have $\alpha\KLD{V_\alpha}{P_1}<\KLD{V_\beta}{P_1}$ and $\beta\KLD{V_\beta}{P_1}<\KLD{V_\alpha}{P_1}$. Observe that $(P_0,\alpha)$ and $(P_1',\beta)$ play symmetric roles. WLOG assume $\alpha\leq\beta$ and get $\alpha<1$. Arbitrarily pick $(P_0,P_1)\in\distset$, and now the goal is to find $P_1'\in\mcal{P}_\varepsilon\setminus \{P_1\}$ such that $\alpha\KLD{V_\alpha}{P_1}<\KLD{V_\beta}{P_1}< \frac{1}{\beta}\KLD{V_\alpha}{P_1}$. This is possible since $\alpha< \frac{1}{\beta}$ and $V_\alpha\neq P_1$.
The closed-form expression of $V_\beta$ can be derived and notice that it is continuous in $P_1'$. Hence $\KLD{V_\beta}{P_1}$ is continuous in $P_1'$. When $P_1'=P_1$, $\KLD{V_\beta}{P_1}=0$. When $P_1'=P_0$, since $\alpha\leq\beta$, it can be shown that $\KLD{V_\alpha}{P_1}\leq \KLD{V_\beta}{P_1}$.
By the intermediate value theorem, there exists $P_1'$ between $P_1$ and $P_0$ with the desired property.

    \section{Proofs of Technical Tools}\label{sec:appendix_tools}
\subsection{Proof of Lemma~\ref{lemma: continuity}}\label{app:continuity}
Given any $\epsilon>0$, we aim to find $\delta>0$ such that for all $p, q\in (0,\delta)$, $|h(p,q) -1|\leq\epsilon$. For $p, q\in (0,\delta)$, if $\delta<0.5$,
\begin{align}
    &\quad\ \BKLD{1-p}{q} \\
    &\geq \BKLD{1-\delta}{q} \\
    &= (1-\delta)\log\frac{1-\delta}{q} + \delta\log\frac{\delta}{1-q}\\
    &= (1-\delta)\log\frac{1}{q} + (1-\delta)\log(1-\delta) + \delta\log\delta + \delta\log\frac{1}{1-q}\\
    &\geq (1-\delta)\log\frac{1}{q} -\frac{2}{e\ln2}
\end{align}
The last inequality follows from the fact that for all $x\in(0,1),\;x\log x\geq-\frac{2}{e\ln2}$
For the upper bound,
\begin{align}
    \BKLD{1-p}{q} &= (1-p)\log\frac{1-p}{q} + p\log\frac{p}{1-q} \\
    &\leq \log\frac{1}{q} + \delta\log\frac{1}{1-\delta}.
\end{align}
Hence
\begin{equation}
    \frac{\BKLD{1-p}{q}}{\log\frac{1}{q}} \geq 1-\delta - \frac{2}{e\ln2}\frac{1}{\log \frac{1}{q}} \geq 1-\delta - \frac{2}{e\ln2}\frac{1}{\log \frac{1}{\delta}},
\end{equation}
and
\begin{equation}
    \frac{\BKLD{1-p}{q}}{\log\frac{1}{q}} \leq 1 + \frac{\delta\log\frac{1}{1-\delta}}{\log\frac{1}{q}} \leq 1 + \frac{\delta\log\frac{1}{1-\delta}}{\log\frac{1}{\delta}}.
\end{equation}
It suffices to choose $\delta>0$ small enough such that 
\begin{equation}
    \delta - \frac{2}{e\ln2}\frac{1}{\log \delta} < \epsilon \text{ and } \frac{\delta\log(1-\delta)}{\log\delta} <\epsilon.
\end{equation}

\subsection{Proof of Lemma~\ref{lemma:inf_continuity}}\label{app:inf_continuity}
For the first part, clearly, $\inf_{x\in A} f(x) \geq \inf_{x\in\closure{A}} f(x)$. For $y\in\closure{A}$, choose a sequence $\{y_n\}$ such that $y_n\in A$ and $y_n\to y$. By the continuity of $f$ and the definition of infimum, we have
\begin{equation}
    f(y) = \lim_{n\to\infty}f(y_n)\geq \inf_{x\in A} f(x).
\end{equation}
Since this holds for all $y\in\closure{A}$, it follows that $\inf_{x\in\closure{A}} f(x) \geq \inf_{x\in A} f(x)$.

For the second part, observe that for any $\epsilon>0$, we have $\inf_{x\in \mcal B'_\epsilon(A)\cap S} f(x) \leq \inf_{x\in\closure{A}} f(x)<\infty$. Also, $\inf_{x\in \mcal B'_\epsilon(A)\cap S} f(x)$ is non-decreasing as $\epsilon\to0$. Hence we know the limit exists. Assume
\begin{equation}
    \lim_{\epsilon\to0}\inf_{x\in \mcal B'_\epsilon(A)\cap S} f(x) = a < \inf_{x\in\closure{A}} f(x).
\end{equation}
For each $n\in\mbb N$, there exists $y_n\in\mcal B'_{1/n}(A)\cap S$ such that $f(y_n)\leq \inf_{x\in \mcal B'_{1/n}(A)\cap S} f(x) + 1/n$. Since $y_n\in\mcal B'_{1/n}(A)$, there exists $x_n\in A$ such that $\lnorm y_n-x_n\rnorm_1\leq 1/n$. Consider a convergent subsequence $\{y_{n_i}\}$ and let $y\in S$ denote the limit point. Since $\lnorm y-x_{n_i}\rnorm_1 \leq \lnorm y-y_{n_i}\rnorm_1 + \lnorm y_{n_i}-x_{n_i}\rnorm_1$ converges to $0$ as $i$ goes to infinity, it follows that $y\in\closure A$.
By the continuity of $f$, we have
\begin{equation}
    f(y)=\lim_{i\to\infty} f(y_{n_i})\leq a < \inf_{x\in\closure{A}} f(x),
\end{equation}
which makes a contradiction.

\end{appendices}

\bibliographystyle{IEEEtran}

\begin{IEEEbiographynophoto}{Ching-Fang Li}
(Student Member, IEEE) received the B.S. degree in electrical engineering and mathematics, and the M.S. degree in electrical engineering from National Taiwan University, Taiwan, in 2021 and 2024, respectively. She is currently a Ph.D. student in the department of electrical engineering at Stanford University, CA, USA.
Her research fields mainly lie in information theory.
\end{IEEEbiographynophoto}

\begin{IEEEbiographynophoto}{I-Hsiang Wang}(Member, IEEE) received the B.Sc. degree in Electrical Engineering from National Taiwan University, Taiwan, in 2006. He received a Ph.D. degree in Electrical Engineering and Computer Sciences from the University of California at Berkeley, USA, in 2011. From 2011 to 2013, he was a postdoctoral researcher at \`{E}cole Polytechnique F\`{e}d\`{e}rale de Lausanne, Switzerland. Since 2013, he has been at the Department of Electrical Engineering in National Taiwan University, where he is now a professor. His research interests include network information theory, networked data analysis, and statistical learning. He was a finalist of the Best Student Paper Award of IEEE International Symposium on Information Theory, 2011. He received the 2017 IEEE Information Theory Society Taipei Chapter and IEEE Communications Society Taipei/Tainan Chapters Best Paper Award for Young Scholars.
\end{IEEEbiographynophoto}

\end{document}